\documentclass[a4paper,12pt,notitlepage]{article}
\usepackage[utf8]{inputenc}
\usepackage[T1]{fontenc}
\usepackage{fullpage}
\usepackage{times}
\usepackage{amsmath}
\usepackage{amsfonts}
\usepackage{amssymb}
\usepackage[makeroom]{cancel}
\usepackage{amsthm}
\usepackage{graphicx}
\usepackage{float}
\usepackage{multirow}

\usepackage{mathtools}
\DeclarePairedDelimiter{\ceil}{\lceil}{\rceil}

\usepackage{bbm}
\usepackage{braket}
\usepackage{ytableau,tikz,varwidth}
\usepackage[enableskew]{youngtab}
\usepackage{hyphenat}
\usepackage{verbatim}
\usepackage{subcaption}
\usepackage{hyperref}
\usepackage{soul}
\usepackage{etoolbox}
\usepackage{nameref}
\usepackage{array,multirow,graphicx}
\usetikzlibrary{tikzmark,decorations.pathreplacing,calc}
\usepackage{changepage}

\usepackage[sorting = none]{biblatex}
\AtEveryBibitem{\clearfield{doi}}
\AtEveryBibitem{\clearfield{url}}
\AtEveryBibitem{\clearfield{month}}

\usepackage{listings}
\usepackage{xcolor}
\definecolor{codegreen}{rgb}{0,0.6,0}
\definecolor{codegray}{rgb}{0.5,0.5,0.5}
\definecolor{codepurple}{rgb}{0.58,0,0.82}
\definecolor{backcolour}{rgb}{0.95,0.95,0.92}
\lstdefinestyle{mystyle}{
    backgroundcolor=\color{white},   
    commentstyle=\color{codegreen},
    keywordstyle=\color{magenta},
    numberstyle=\tiny\color{codegray},
    stringstyle=\color{codepurple},
    basicstyle=\ttfamily\footnotesize,
    breakatwhitespace=false,         
    breaklines=true,                 
    captionpos=b,                    
    keepspaces=true,                 
    numbers=left,                    
    numbersep=5pt,                  
    showspaces=false,                
    showstringspaces=false,
    showtabs=false,                  
    tabsize=2
}
\graphicspath{{plots/}}

\theoremstyle{definition}

\allowdisplaybreaks
\DeclareRobustCommand{\bbone}{\text{\usefont{U}{bbold}{m}{n}1}}

\renewcommand{\S}{\mathcal{S}}
\def\mc {\mathcal}

\newcommand{\prob}{\mathbb{P}}
\newcommand{\eigval}[2]{\lambda_{#1}(#2)}

\newcommand{\floor}[1]{\ensuremath\lfloor #1\rfloor}

\DeclarePairedDelimiterX\bk[2]{\langle}{\rangle}{#1 \delimsize\vert #2}
\DeclarePairedDelimiterX\kb[2]{\vert }{\vert }{#1 \delimsize\rangle\langle#2}

\newcommand*{\Zero}{\textbf{0}}

\newcommand{\blockdiagonal}[3]{
\begin{bmatrix}
    #1      & \Zero     & \cdots    & \Zero     \\
    \Zero   & #2        & \cdots    & \Zero     \\
    \vdots  & \vdots    & \ddots    & \vdots    \\
    \Zero   & \Zero     & \cdots    & #3        \\
\end{bmatrix}
}

\newcommand{\Ei}[1]{\mathrm{Ei}\left(#1\right)}

\def\={\;=\;} \def\+{\,+\,}

\newcommand{\E}[1]{\underset{U \sim #1}{\mathbb{E}}}
\newcommand{\err}[1]{{\color{black}#1}}

\newtheorem{theorem}{Theorem}
\newtheorem{lemma}[theorem]{Lemma}

\newtheorem{corollary}[theorem]{Corollary}

\newtheorem{example}{Example}
\newtheorem{fact}{Fact}
\newtheorem{Conjecture}{Conjecture}

\usepackage{tikz}

\apptocmd\appendix{%
  \addcontentsline{toc}{chapter}{Appendix}%
  \counterwithin{equation}{section}%
  \counterwithin{figure}{section}%
  \counterwithin{table}{section}%
}{}{}

\begin{document}

\title{A random matrix model for random approximate $t$-designs }
\author{Piotr Dulian$^1$ and Adam Sawicki $^{1}$}
%\affiliation{}

\date{%
    $^1$ Center for Theoretical Physics, Polish Academy of Sciences,\\ Al. Lotnik\'ow 32/46, 02-668 Warszawa, Poland,
    %\today
}
\maketitle

\begin{abstract}
For a Haar random set $\mathcal{S}\subset U(d)$ of quantum gates we consider the uniform measure $\nu_\mc{S}$ whose support is given by $\mathcal{S}$. The measure $\nu_\mc{S}$ can be regarded as a $\delta(\nu_\mc{S},t)$-approximate $t$-design, $t\in\mathbb{Z}_+$. We propose a random matrix model that aims to describe the probability distribution of $\delta(\nu_\mathcal{S},t)$ for any $t$. Our model is given by a block diagonal matrix whose blocks are independent, given by Gaussian or Ginibre ensembles, and their number, size and type is determined by $t$. We prove that, the operator norm of this matrix, $\delta({t})$, is the random variable to which $\sqrt{|\mathcal{S}|}\delta(\nu_\mc{S},t)$ converges in distribution when the number of elements in $\mc{S}$ grows to infinity. Moreover, we characterize our model giving explicit bounds on the tail probabilities $\mathbb{P}(\delta(t)>2+\epsilon)$, for any $\epsilon>0$. We also show that our model satisfies the so-called spectral gap conjecture, i.e. we prove that with the probability $1$ there is $t\in\mathbb{Z}_+$ such that $\sup_{k\in\mathbb{Z}_{+}}\delta(k)=\delta(t)$. Numerical simulations give convincing evidence that the proposed model is actually almost exact for any cardinality of $\mc{S}$. The heuristic explanation of this phenomenon, that we provide, leads us to conjecture that the tail probabilities $\mathbb{P}(\sqrt{\mathcal{S}}\delta(\nu_\mathcal{S},t)>2+\epsilon)$ are bounded from above by the tail probabilities $\mathbb{P}(\delta(t)>2+\epsilon)$ of our random matrix model. In particular our conjecture implies that a Haar random set $\mathcal{S}\subset U(d)$ satisfies the spectral gap conjecture with the probability $1$.
\end{abstract}

\section{Introduction and Main Results}
Approximate $t$-design are ensembles of unitaries that (approximately) recover Haar averages of polynomials in entries of unitaries up to the order $t$. Their relation to the notion of epsilon nets was recently given in \cite{nets}. Although there are methods of constructing exact $t$- designs \cite{Yoshifumi2},  their implementation on near term quantum devices is fraught with inevitable noise and errors that effectively change them  into approximate $t$-designs. Approximate $t$-designs find numerous applications throughout quantum information, including randomized benchmarking \cite{Gambetta2014}, efficient estimation of properties of quantum states \cite{EffLearning2020}, decoupling \cite{Decoupl2013}, information transmission \cite{InfTransmission2009}, quantum state discrimination \cite{StateDiscrimination2005}, criteria for universality of quantum gates \cite{Sawicki22} and complexity growth \cite{Suskind2018,ChaosDesign2017,MO1,jonas2}. Recently, there was a lot of interest in efficient implementations of pseudo-random  unitaries. First, it was shown in \cite{HArrowHasstings2008,HH09} that Haar random sets of gates from $U(d)$ form rather good $t$-designs for $t=O(d^{\frac{1}{6}}/\log(d))$. Moreover, in \cite{BHH2016} it was shown that random circuits built form Haar-random 2-qubit gates acting (according to the specified layout) on $N$-qubit systems of the  depth polynomial in $N$ form approximate $t$-designs. These results were later improved in \cite{Harrow2018} and \cite{Haferkamp2022randomquantum,PhysRevA.104.022417,Yoshifumi1}, where even faster convergence in $n$ was proved. Moreover, the authors of \cite{Qhomeopathy2020} showed that random circuits constructed from Clifford gates and a small number of non-Clifford gates can be used to efficiently generate approximate $t$-designs (see also \cite{OLIVIERO2021127721,Leone2021quantumchaosis} where the authors show that in practice one needs an extensive number of non-Clifford resources to reproduce features of $t$-designs).

In this paper we propose a random matrix model that describes approximate $t$-designs constructed from Haar random sets of gates. In order to explain our main results we need to first introduce the main object of the paper. Let  $\{\mathcal{S},\nu_{\mathcal{S}}\}$ be an ensemble of quantum gates, where $\mathcal{S}$ is a finite subset of $U(d)$ and $\nu_{\mathcal{S}}$ is the uniform measure with the support given by $\mathcal{S}$. In order to simplify notation we will denote the cardinality of the set $\mc{S}$ by $|\mc{S}|$ or just by $\mc{S}$. We define the \emph{moment operator} associated with any measure $\nu$ on $U(d)$ by:
\begin{gather}\label{def:intro_moment-t}
T_{\nu,t}\coloneqq\int_{U(d)} d\nu(U) U^{t, t},\,\,\mathrm{where}\,
U^{t,t}=U^{\otimes t}\otimes \bar{U}^{\otimes t}.
\end{gather}
An ensemble for which 
\begin{gather}
\delta(\nu_\S,t):=\left\|T_{\nu_\S,t}-T_{\mu,t}\right\|=\delta<1,
\end{gather}
is called {\it $\delta$-approximate $t$-design}, where $\|\cdot\|$ is the operator norm and  $\mu$ is the Haar measure on $U(d)$ with normalization $\mu(U(d))=1$. One can also consider 
\begin{gather}
\delta(\nu_\S):=\mathrm{sup}_{t}\delta(\nu_{\mathcal{S}},t)\in[0,1],
\end{gather}
which, by the Peter-Weyl theorem, is equal to the operator norm of the moment operator 
\begin{gather}\label{def:intro_moment}
   (T_{\nu_\S}f)(U)=\int d\nu_\S(V) f(V^{-1}U),
\end{gather}
acting on functions $f\in L^2(U(d))$ that have vanishing mean value, $\int_{U(d)} d\mu(U) f(U)=0$. A long standing conjecture, known as the spectral gap conjecture, states that for any $\mc{S}$ that is universal $1$ is not an accumulation point of the spectrum of $T_{\nu_\S}$. In other words one has $1-\delta({\nu_{\mc{S}}})>0$. So far this was proven rigorously when the matrices from $\mc{S}$ have algebraic entries  \cite{Bourgain2011, naud2016}.

In this paper we address an important question of how good $t$-designs are random sets of gates. In particular we consider two natural gate-sets $\mc{S}$:
\begin{itemize}
    \item $\mathcal{S}=\{U_1,\ldots,U_n\}$, where $U_k$'s are independent and Haar random unitaries from $U(d)$. Such $\mc{S}$ will be called {\it Haar random  gate-set}
    \item $\mathcal{S}=\{U_1,\ldots,U_n\}\cup\{U_1^{-1},\ldots,U_n^{-1}\}$, where $U_k$'s are independent and Haar random unitaries from $U(d)$. Such $\mc{S}$ will be called {\it symmetric Haar random  gate-set}.
\end{itemize}

\noindent Clearly, when $\mc{S}$ is a (symmetric) Haar random gate-set $\delta(\nu_\mc{S},t)$ is a random variable. Our goal is to describe probability distribution of  this random varible and in particular to find concrete explicit bounds on the tail probability $\mathbb{P}(\delta(\nu_\S,t)>\alpha)$, for $\alpha>0$. Some residual results of this sort can be found in \cite{HArrowHasstings2008}. The bounds obtained there, however, depend on undetermined and potentially large constants and what is more important they can be applied only when $t=O(d^{\frac{1}{6}}/\log(d))$. Moreover, bounds on the tail probability $\mathbb{P}(\delta(\nu_\S,t)>\alpha)$, for $\alpha>0$ have been recently rigorously derived in \cite{concentration} using matrix concentrations inequalities \cite{tropp2015}. In this paper we present a new approach to the problem which identifies moment operators as elements of a well known Random Matrix Ensembles. Our results  show much faster decay of the tail probabilities of $\delta(\nu_\S,t)$ than any previous results.

First, in Section \ref{sec:moments}, using representation theory of the unitary group we provide more useful, for our purposes, formula for $\delta(\nu_\mc{S},t)$. This relies on the fact that $T_{\nu_\mc{S},t}$ can be brought to a block diagonal form, where blocks are labeled by sequences of $d$ non-increasing integers $\lambda\in \Lambda_t$ that correspond to irreducible representations $\pi_\lambda:U(d)\rightarrow GL(d_\lambda)$ appearing in the decomposition $U^{t,t}\simeq \bbone^{\oplus m_0}\oplus\bigoplus_{\lambda\in\Lambda_t}\left (\pi_{\lambda}\right)^{\oplus m_\lambda}$. More precisely in order to determine $\delta(\nu_{\mc{S}},t)$ we consider a block diagonal operator
\begin{gather*}
    \sum_{U \in \mathcal{S}}\nu_\mc{S}(U) \blockdiagonal{\pi_{\lambda_1}(U)}{\pi_{\lambda_2}(U)}{\pi_{\lambda_k}(U)} = \blockdiagonal{T_{\nu_\mc{S},\lambda_1}}{T_{\nu_\mc{S},\lambda_2}}{T_{\nu_\mc{S},\lambda_k}},
\end{gather*}
where $\lambda\in \Lambda_t^{\mathrm{ess}}$ and $\Lambda_t$, $\Lambda_t^\mathrm{ess}$, are defined in Section \ref{sec:moments}. Denoting $\delta(\nu_{\mathcal{S}},\lambda):=\|T_{\nu_\mc{S},\lambda}\|$ we can write
\begin{gather}\label{def:intro_moment-lambda}
\delta(\nu_\S,t)=\mathrm{sup}_{\lambda\in\Lambda_t^{\mathrm{ess}}}\delta(\nu_{\mathcal{S}},\lambda),\,\,\,\delta(\nu_\mc{S}):=\mathrm{sup}_{t}\delta(\nu_{\mathcal{S}},t).
\end{gather}

%\begin{gather*}
%    \label{eq:big-lambda}
% \Lambda_t=\left\{\lambda=(\lambda_1,\ldots,\lambda_d)|\,\lambda\neq 0,\,\lambda\in\mathbb{Z}^d,\,\forall_k\,\lambda_k\geq\lambda_{k+1},\,\Sigma(\lambda):=\sum_{i=1}^d\lambda_i=0,\,\Sigma(\lambda_{+})\leq t\right\}
%\end{gather*}
%As we are interested in the norm $T_{\nu_\mc{S},t}-T_{\mu,t}$, we can consider a block diagonal matrix whose blocks are as blocks of $T_\{\nu_{\mc{S}},t}$ albeit there no multiplicities and the trivial representation do not appear. Then we can define 

%The $\lambda$'s that appear in $T_{\nu_{\mc{S}},t}$ are those form the set
%\begin{gather*}
%    \label{eq:big-lambda}
%    \lambda\in\Lambda_t=\left\{\lambda=(\lambda_1,\ldots,\lambda_d)|\,\lambda\in\mathbb{Z}^d,\,\forall_k\,\lambda_k\geq\lambda_{k+1},\,\Sigma(\lambda):=\sum_{i=1}^d\lambda_i=0,\,\Sigma(\lambda_{+})\leq t\right\}
%\end{gather*}
%and given by. 

\noindent We note that $T_{\nu_{\mathcal{S}},\lambda}$ is either (1) a real (symmetric) random matrix, when $\pi_\lambda$ is a real representation and $\mc{S}$ is a (symmetric) Haar random gate-set, or (2) A complex (hermitian) random matrix, when $\pi_\lambda$ is a complex representation and $\mc{S}$ is a (symmetric) Haar random gate-set. Using multidimensional Central Limit Theorem and the orthogonality relations for irreducible representations of the unitary group that we describe in Section \ref{sec:tools} we show in Sections \ref{sec:asym-proof} and \ref{sec:sym-proof} our first main result:
\begin{theorem}\label{intro_thm1}
When $|\mc{S}|\rightarrow\infty$:
\begin{enumerate}
    \item $\sqrt{|\mc{S}|}T_{\nu_\mc{S},\lambda}$ converges in distribution to a random matrix $T_\lambda$ from: (1) the standard Real (Complex) Ginibre Ensemble when $\mc{S}$ is Haar random gate-set and $\pi_\lambda$ is a real (complex) representation, (2) the standard Gaussian Orthogonal (Unitary) Ensemble when  $\mc{S}$ is symmetric Haar random gate-set and $\pi_\lambda$ is a real (complex) representation.
    \item Operator $ {\color{black}\sqrt{|\mc S|}}\bigoplus_{\lambda\in\Lambda_t^{\mathrm{ess}}}T_{\nu_{\mathcal{S}},\lambda}$  converges in distribution to a block diagonal matrix 
\begin{gather}\label{intro-model1} T_t=\bigoplus_{\lambda\in\Lambda_t^{\mathrm{ess}}}T_{\lambda}.
\end{gather}
with blocks $T_\lambda$ belonging to the ensembles listed above. Moreover $T_\lambda$'s with distinct $\lambda$'s are independent. 
\end{enumerate}
\end{theorem}
Theorem \ref{intro_thm1} gives a nice Random Matrix model which describes properties of approximate $t$-designs in the regime of large $\mc{S}$. Moreover, we can characterize rigorously many important properties of this model. In particular we focus on random variables:
\begin{gather}
    \delta(\lambda):=\|T_\lambda\|=\lim_{|\mc{S}|\rightarrow\infty}\sqrt{|\mc{S}|}\delta(\nu_{\mc{S}},\lambda),\,\,
    \delta(t):=\|T_t\|=\lim_{|\mc{S}|\rightarrow\infty}\sqrt{|\mc{S}|}\delta(\nu_{\mc{S}},t),\\
    \delta:=\lim_{t\rightarrow\infty}\|T_t\|=\lim_{|\mc{S}|\rightarrow\infty}\sqrt{|\mc{S}|}\delta(\nu_{\mc{S}}).
\end{gather}
Making use of concentration inequalities for Gaussian random variables and following the reasoning which allows to bound the norm of matrices from Ginibre and Gaussian ensembles given in \cite{szarek2001Chapter8,SzarekBook} and described in Section \ref{sec:tools} we derive the following tail bounds that constitute our next main result
\begin{theorem}\label{thm:intro_bounds_T_t}
In the (symmetric) Haar random setting we have the following tail bounds for $\delta(t)$: 
\begin{gather}\label{eq:intro_bounds_T_t}
      \mathbb{P}\left(\delta(t)>2+\epsilon\right )\leq\sum_{\lambda\in\Lambda_t^{\mathrm{ess}}}F_{\ast}(d_\lambda,\epsilon)
      \end{gather}
where $\epsilon>0$ and functions $F_\ast$ are defined by: $F_O(d_\lambda,\epsilon)=\frac{1}{2}e^{-\frac{d_\lambda\epsilon^2}{4}}$, 
$F_U(d_\lambda,\epsilon)=F_R(d_\lambda,\epsilon):=\frac{1}{2}e^{-\frac{d_\lambda\epsilon^2}{2}}$, $  F_C(d_\lambda,\epsilon):=\frac{1}{2}e^{-d_\lambda\epsilon^2}$. Moreover, $\ast=R$ ($\ast=O$) when $\lambda$ corresponds to a real representation and $\mc{S}$ is (symmetric) Haar random gate-set and $\ast=C$ ($\ast=U$) when $\lambda$ corresponds to a complex representation and $\mc{S}$ is (symmetric) Haar random gate-set.
\end{theorem}
\noindent Moreover, using Theorem \ref{thm:intro_bounds_T_t} and the Borel-Cantelli Lemmas that we review in Section \ref{sec:tools} we prove in Section \ref{sec:model} that for our Random Matrix model the spectral gap conjecture is fulfilled with the probability $1$. 
\begin{theorem}\label{thm:intro_gap}
    In both Haar random and symmetric Haar random settings with the probability $1$ there exists finite $t\in \mathbb{Z}_{+}$ such that  $\delta=\delta(t)$. 
\end{theorem}
\noindent
In the same Section we also prove concrete bounds for the tail probability $\mathbb{P}\left(\delta>2+\epsilon\right )$.
\begin{theorem}\label{thm:intro_ineq}
The tail probability of $\delta$ is bounded by:
\begin{gather*}
    \mathbb{P}\left(\delta>2+\epsilon\right )\leq \frac{e^{-\epsilon^2/(2c^2)}}{2(e^{\epsilon^2/c^2}-1)},\;\mathrm{if}\;d=2\;\mathrm{and}\;\epsilon>0,\\
    \mathbb{P}\left( \delta > 2 + \epsilon \right) \leq
    e^{-\frac{d(d+1)}{4c^2} \epsilon^2} \left( e^{2\pi\sqrt{\frac{d+2}{3}}} \frac{60+100\pi\sqrt{3d+6}}{d+2} - b \right) + \\
    +\frac{60}{d^2} \left( 2 + \frac{\sqrt{2\pi}c}{\epsilon} \right) e^{-\frac{3}{4c^2}d(d-1)\epsilon^2 + 2\pi\sqrt{\frac{d}{3}}}, \; \mathrm{if} \; d \ge 2 \; \mathrm{and} \; \epsilon>c \left(\frac{4\pi^2}{3d(d-1)^2}\right)^{1/4},
\end{gather*}
where 
\begin{equation*}
    b = 10 \left[8\pi^2 \Ei{2\sqrt{\frac{2}{3}}\pi} - e^{2\sqrt{\frac{2}{3}}\pi}(3+2\sqrt{6}\pi) \right] \approx 3855.93,
\end{equation*}
and $c=1$ in the Haar random setting and $c=\sqrt{2}$ in the symmetric Haar random setting.
\end{theorem}

As a direct consequence of Theorem \ref{intro_thm1} for any $t\in \mathbb{Z}_+$ and given precision $\eta>0$ there is $N_0=N(d,\eta,t)$ such that for a (symmetric) Haar random gate-set $\mc{S}$ with $|\mc{S}|>N_0$ we have
\begin{equation} \label{eq:intro-ine}
    \mathbb{P}\left( \sqrt{\mathcal{S}} \delta(\nu_\mc{S},t)>2+\epsilon\right )\leq\eta+\mathbb{P}\left(\delta(t)>2+\epsilon\right ).
\end{equation}
\textcolor{black}{In order to understand better how the convergence depends on $|\mathcal{S}|$ and $d_\lambda$ in Section \ref{sec:numerics} we compare $\delta(\lambda)$ and $\delta(t)$ with the numerical simulations  of $\sqrt{\mathcal{S}}\delta({\nu_{\mathcal{S}},\lambda})$ and $\sqrt{\mathcal{S}}\delta({\nu_{\mathcal{S}},t})$ for 
\begin{enumerate}
    \item one qubit $d=2$: $1\leq t\leq 500$ which means irreducible representations of $U(2)$ with the dimension $3 \leq d_\lambda \leq 1001$, 
    \item two qubits $d=4$:  $1\leq t\leq 6$ which means irreducible representations of $U(4)$ with the dimension $15 \leq d_\lambda \leq 5200$.
\end{enumerate} 
Our numerical simulations show that 
\begin{equation} \label{eq:intro-ine-1}
    \mathbb{P}\left( \sqrt{\mathcal{S}} \delta(\nu_\mathcal{S},\lambda)>2+\epsilon\right )\leq\mathbb{P}\left(\delta(\lambda)>2+\epsilon\right ),
\end{equation}
for $3\leq|\mathcal{S}|\leq 26$ (see Figures \ref{fig:conj}, \ref{fig:conj_t} and \ref{fig:conj_t_d2}). A careful analysis of Figures \ref{fig:conj}, \ref{fig:conj_t} and \ref{fig:conj_t_d2} reveals that \ref{eq:intro-ine-1} becomes much sharper for smaller number of generators (for any $d_\lambda$). This suggests that there is another mechanism, not related to central limit theorem, that plays a significant role in \ref{eq:intro-ine-1}. Moreover our calculations shows that in the regime of large block dimensions, that is $d_\lambda \rightarrow \infty$, random variables $\sqrt{|\S|}\delta({\nu_{\mc{S}},\lambda})$ and $\delta(\lambda)$ concentrate around $2\sqrt{\frac{|\mc S | - 1}{|\mc S|}}$ and $2$ respectively (see Figures \ref{fig:averages}, \ref{fig:medians} and \ref{fig:renorm}). Since the former is always smaller than the latter this behaviour implies that for dimensions $d_\lambda$ bigger than those we could check numerically the inequality \eqref{eq:intro-ine-1} is even sharper.
In Section \ref{sec:numerics} we provide extended heuristic explanation of this phenomena that uses spectral properties of $T_{\nu_\S,t}$ and $T_{\nu_\S,\lambda}$ described in Section \ref{sec:spectral_properties}. This heuristic arguments and numerical results lead to the following conjecture:} 
\begin{Conjecture}\label{conjecture}
Let $\epsilon>0$. For any (symmetric) Haar random $\mc{S}\subset U(d)$ with at least $3$ elements ($2$ elements and their inverses) we have: 
\begin{gather}
    \mathbb{P}\left(\delta(\nu_\mc{S},\lambda)>\frac{2+\epsilon}{\sqrt{|\mc{S}|}}\right )\leq F_{\ast}(d_\lambda,\epsilon),
    \end{gather}
where $\ast\in\{O,U,R,C\}$ and the correspondence between $\ast$, $\lambda$ and the type of a gate-set is as in Theorem \ref{thm:intro_bounds_T_t}. Moreover, for any fixed cardinality of $\mc{S}$ there is $\lambda_0$ such that for any $\lambda$ with $d_\lambda>d_{\lambda_0}$ we have 
\begin{gather}
    \mathbb{P}\left(\delta(\nu_\mc{S},\lambda)>\frac{2+\epsilon}{\sqrt{|\mc{S}|}}\right )\leq  \mathbb{P}\left(\delta(\lambda)>2+\epsilon\right ).
    \end{gather}
\end{Conjecture}
As a direct consequence of this conjecture we can rewrite Theorem \ref{thm:intro_bounds_T_t} and Theorem \ref{thm:intro_ineq} replacing $\delta(t)$ by $\sqrt{|\mc{S}|}\delta(\nu_{\mc{S}},t)$. Moreover, assuming Conjecture \ref{conjecture} holds we prove in Section \ref{sec:model} that for a (symmetric) Haar random gate-set $\mc{S}$ with the probability $1$ there exists finite $t\in \mathbb{Z}_{+}$ such that $\delta(\nu_{\mc{S}})=\delta(\nu_{\mc{S}},t)$.

%and that the tail probability of $\delta(\nu_{\mc{S}})$ is bounded by:
%\begin{gather*}
%    \mathbb{P}\left(\delta(\nu_\mc{S})>\frac{2+\epsilon}{\sqrt{\mc{S}}}\right )\leq \frac{e^{-\epsilon^2/(2c^2)}}{2(e^{\epsilon^2/c^2}-1)},\;\mathrm{if}\;d=2\;\mathrm{and}\;\epsilon>0,\\
 %   \mathbb{P}\left( \delta(\nu_\mc{S}) > \frac{2 + \epsilon}{\sqrt{\mc{S}}} \right) \leq
 %   e^{-\frac{d(d+1)}{4c^2} \epsilon^2} \left( e^{2\pi\sqrt{\frac{d+2}{3}}} \frac{60+100\pi\sqrt{3d+6}}{d+2} - b \right) + \\
 %   +\frac{60}{d^2} \left( 2 + \frac{\sqrt{2\pi}c}{\epsilon} \right) e^{-\frac{3}{4c^2}d(d-1)\epsilon^2 + 2\pi\sqrt{\frac{d}{3}}}, \; \mathrm{if} \; d \ge 2 \; \mathrm{and} \; \epsilon>c \left(\frac{4\pi^2}{3d(d-1)^2}\right)^{1/4},
%\end{gather*}
%where 
%\begin{equation*}
 %   b = 10 \left[8\pi^2 \Ei{2\sqrt{\frac{2}{3}}\pi} - e^{2\sqrt{\frac{2}{3}}\pi}(3+2\sqrt{6}\pi) \right] \approx 3855.93,
%\end{equation*}
%and $c=1$ in the Haar random setting and $c=\sqrt{2}$ in the symmetric Haar random setting.
Throughout the paper we denote our results by Theorems and Lemmas, and by Facts we denote results taken from the literature.  

\section{Moment operators}\label{sec:moments}
Let  $\{\mathcal{S},\nu_{\mathcal{S}}\}$ be an ensemble of quantum gates, where $\mathcal{S}$ is a finite subset of $U(d)$ and $\nu_{\mathcal{S}}$ is any measure on $\mathcal{S}$. We define the \emph{moment operator} associated with a measure $\nu$ by:
\begin{gather}\label{def:moment-t}
T_{\nu,t}\coloneqq\int_{G_d} d\nu(U) U^{ t,t},\,\,\mathrm{where}\,
U^{t,t}=U^{\otimes t}\otimes \bar{U}^{\otimes t}.
\end{gather}
In the following we will be interested in 
\begin{gather}
\delta(\nu_\S,t):=\left\|T_{\nu_\S,t}-T_{\mu,t}\right\|,
\end{gather}
where $\|\cdot\|$ is the operator norm and  $\mu$ is the Haar measure on $U(d)$ with the normalization $\mu(U(d))=1$. In the remaining part of the paper we will often use the following notation 
\begin{gather}
    \E{\nu} f(U):=\int_{U(d)}d\nu(U)f(U).
\end{gather}

Let us first note that the map $U\mapsto U^{t,t}$ is a representation of the unitary group $U(d)$. This representation is reducible and decomposes into irreducible representations of $U(d)$. Those are usually labeled by sequences of $d$ non-increasing integers called highest weights. To simplify our description we introduce the following notation. For  $\lambda = \left( \lambda_1, ..., \lambda_d \right) \in \mathbb{Z}^d$, that satisfies $\lambda_{k}\geq \lambda_{k+1}$ for any $k\in\{1,\ldots,d-1\}$, we denote 
\begin{itemize}
    \item by $\pi_\lambda$ the corresponding irreducible representation  and by $d_\lambda := \mathrm{dim}\pi_\lambda$ its dimension given by:
    \begin{gather}
    \label{eq:weyl_dimension}
    d_\lambda=\prod_{1\leq i<j\leq d}\frac{\lambda_i-\lambda_j+j-i}{j-i},
\end{gather}
    \item by $l(\lambda) := d$ the {\it length} of $\lambda$, 
    \item by $\Sigma(\lambda) := \sum_{i=1}^{d}\lambda_i$ the sum of the entries of $\lambda$, 
    \item by $\|\lambda\|_1\coloneqq\sum_{k=1}^d|\lambda_k|$ the sum of the absolute values of the entries of $\lambda$,
    \item by $\lambda_+$ the subsequence of positive integers in $\lambda$.
\end{itemize}
It is worth mentioning here that the irreducible representation $\pi_\lambda$ of the unitary group $U(d)$ gives rise to an irreducible representation of the special unitary group $SU(d)$ by the restriction $\pi_\lambda|_{SU(d)}$. In the literature irreducible representations of $SU(d)$ are typically labeled by sequences of integers called highest weights and denoted by $\lambda^s$ or by Young diagrams denoted by $\lambda^Y$. The relations between the highest weight $\lambda$ of $\pi_\lambda$ and the highest weight $\lambda^s$ and the Young diagram $\lambda^Y$ of $\pi_\lambda|_{SU(d)}$ are given by:
\begin{align}\label{eq:standard-young}
 \lambda^Y &= (\lambda_1-\lambda_d, \lambda_2-\lambda_d, ..., \lambda_{d-1}-\lambda_d,0),\\
    \lambda^s &= (\lambda_1-\lambda_2, \lambda_2-\lambda_3, ..., \lambda_{d-1}-\lambda_d,0).
\end{align}
As the last integer in the above formulas is always zero  typically it is omitted. For the further discussion and details see Lemma 5 in our recent work \cite{concentration}. 

\begin{fact}(\cite{Stroomer})\label{fact:moemnt_decomposition}
Irreducible representations that appear in the decomposition of $U\mapsto U^{\otimes t}\otimes \bar{U}^{\otimes t}$ are $\pi_\lambda$ with $l(\lambda)=d$, $\Sigma(\lambda)=0$ and $\Sigma(\lambda_+) \leq t$. That is we have 
\begin{gather}\label{decomposition}
    U^{\otimes t}\otimes \bar{U}^{\otimes t}\simeq \bbone^{\oplus m_0}\oplus \bigoplus_{\lambda\in \Lambda_t}\pi_\lambda(U)^{\oplus m_{\lambda}}\simeq  \left(U\otimes \bar{U}\right)^{\otimes t},
\end{gather}
where 
\begin{gather}
    \label{eq:big-lambda}
    \Lambda_t=\left\{\lambda=(\lambda_1,\ldots,\lambda_d)|\,\lambda\in\mathbb{Z}^d,\lambda\neq 0,\,\forall_k\,\lambda_k\geq\lambda_{k+1},\,\Sigma(\lambda)=0,\,\Sigma(\lambda_{+})\leq t\right\},
\end{gather}
and $\bbone$ stands for the trivial representation and $m_0$ is its  multiplicity and $m_{\lambda}$ is the multiplicity of $\pi_\lambda$.
\end{fact}
The representations occurring in decomposition \eqref{decomposition} are in fact irreducible representation of the projective unitary group, $PU(d)=U(d)/\sim$, where $U\sim V$ iff $U=e^{i\phi} V$. One can show that every irreducible representation of $PU(d)$ arises this way for some, possibly large, $t$ \cite{Dieck}. For $t=1$ decomposition (\ref{decomposition}) is particularly simple and reads $U\otimes \bar{U}\simeq\mathrm{Ad}_U\oplus \bbone$, where $\mathrm{Ad}_U$ is the adjoint representation of $U(d)$ and $PU(d)\simeq\mathrm{Ad}_{U(d)}$\footnote{By $\mathrm{Ad}_U$ we mean the matrix $\mathrm{Ad}_U:\mathfrak{su}(d)\rightarrow \mathfrak{su}(d)$, $\mathrm{Ad}_U(X)=UXU^{-1}.$}. 

We further note that $U\mapsto U^{\otimes t}\otimes \bar{U}^{\otimes t}$ is a real representation and as such can be decomposed into the direct sum of real irreducible representations.  When acting on a complex vector space, an irreducible real representation $\pi(U)$ can be either 
\begin{enumerate}
    \item irreducible on both complex and real vector spaces
    \item irreducible when acting on a real vector space but when acting on a complex vector space a direct sum $\pi(U)=\pi_\lambda(U)\oplus\overline{\pi_\lambda(U)}$, where $\pi_\lambda$ is a complex irreducible representation of $U(d)$ and $\overline{\pi_\lambda}$ is the conjugate representation of $\pi_\lambda$. 
\end{enumerate}
It is also known that $\overline{\pi_\lambda(U)}=\pi_{\lambda^\ast}(U)$, where $\lambda^\ast=-(\lambda_d,\lambda_{d-1},\ldots,\lambda_2,\lambda_1)$. Summing up irreducible representations showing up in the decomposition \eqref{decomposition} are either real or complex representations (there are no quaternion representations). Furthermore for $\lambda\in\Lambda_t$, the representation $\pi_\lambda$ is real iff $\lambda_{i}=-\lambda_{d-i+1}$, for any $i\in\{1,2,\ldots, d\}$ and is complex otherwise (see also Proposition 26.24 \cite{fulton1991representation} for analogous conditions when $\lambda$ is arbitrary). In the following we will denote by $\Lambda_t^r$ a subset of $\lambda$'s in $\Lambda_t$ that correspond to real representations and by $\Lambda_t^c$ a subset of $\lambda$'s in $\Lambda_t$ that correspond to complex representations. For any irreducible representation $\pi_{\lambda}$, $\lambda\in\Lambda_t$ we define
\begin{gather}\label{def:moment-lambda}
    T_{\nu_{\mathcal{S}},\lambda}\coloneqq \E{\nu_{\mathcal{S}}}\pi_\lambda(U),\,\,\,\delta(\nu_{\mathcal{S}},\lambda)\coloneqq  \|T_{\nu_{\mathcal{S}},\lambda}\|.
\end{gather}
One easily see that
\begin{gather}
    \delta(\nu_{\mathcal{S}},\lambda)= \left \|\E{\nu_{\mathcal{S}}}\pi_\lambda(U)\right\|\leq \E{\nu_{\mathcal{S}}}\|\pi_\lambda(U)\|=1
\end{gather}
Thus $\delta_{\nu_\S,\lambda}\in[0,1]$. It follows directly from Fact \ref{fact:moemnt_decomposition} that 
\begin{gather}\label{eq:delta_t}
    T_{\nu_{\mathcal{S}},t}\simeq \bbone^{\oplus m_0}\oplus\bigoplus_{\lambda\in\Lambda_t}\left (T_{\nu_{\mathcal{S}},\lambda}\right)^{\oplus m_\lambda}.
\end{gather}
Using the fact that $\E{\mu}\pi_\lambda(U)=0$ for any nontrivial irreducible representation $\pi_\lambda$ we see that $ T_{\mu,t}=\bbone^{\oplus m_0}$. Thus we have: 
\begin{gather}
    \label{eq:delta_t_sup}
   \delta(\nu_\S,t)= \left\|T_{\nu_{\mathcal{S}},t}-T_{\mu,t} \right\|=\left\|\bigoplus_{\lambda\in\Lambda_t}\left (T_{\nu_{\mathcal{S}},\lambda}\right)^{\oplus m_\lambda}\right\|=\mathrm{sup}_{\lambda\in\Lambda_t}\|T_{\nu_{\mathcal{S}},\lambda}\|=\mathrm{sup}_{\lambda\in\Lambda_t}\delta(\nu_{\mathcal{S}},\lambda).
\end{gather}
Thus $ \delta(\nu_\S,t)\in [0,1]$. An ensemble for which $\delta(\nu_\S,t)=\delta<1$ is called {\it $\delta$-approximate $t$-design}. 
One can also define 
\begin{gather}
\delta(\nu_\S):=\mathrm{sup}_{t}\delta(\nu_{\mathcal{S}},t)\in[0,1],
\end{gather}
which, by the Peter-Weyl theorem, is equal to the operator norm of the moment operator $T_{\nu_\S}$ acting on functions $f\in L^2(PU(d))$ that have vanishing mean value, $\int d\mu(U) f(U)=0$ (see \cite{https://doi.org/10.48550/arxiv.2201.11774} for detailed derivation):
\begin{gather}\label{def:moment}
   (T_{\nu_\S}f)(U)=\int d\nu_\S(V) f(V^{-1}U).
\end{gather}
Finally we note that since $T_{\nu_\mc{S},\lambda^\ast}=\overline {T}_{\nu_\S,\lambda}$ we have $\|T_{\nu_\S,\lambda}\|=\|T_{\nu_\S,\lambda^\ast}\|$. In order to remove this  redundancy we define $\tilde{\Lambda}_t^c$ to be a subset of $\Lambda_t^c$  that for any pair $\lambda,\lambda^\ast\in\Lambda_t^c$ contains either $\lambda$ or $\lambda^\ast$, but not both. Then
\begin{gather}
    \label{eq:def_lam_ess}
    \Lambda_t^{\mathrm{ess}}:=\Lambda_t^r\cup\tilde{\Lambda}_t^c,\\
    \Lambda^{\mathrm{ess}} := \bigcup_{t} \Lambda_t^{\mathrm{ess}}.
\end{gather}
In particular 
\begin{gather}\label{eq:delta_S_t}
\delta(\nu_\S,t)=\mathrm{sup}_{\lambda\in\Lambda_t^{\mathrm{ess}}}\delta(\nu_{\mathcal{S}},\lambda),\\
\delta(\nu_\mathcal{S}) = \sup_t \delta(\nu_\mathcal{S}, t).
\end{gather}
We will also need some properties of partitions. Recall that a partition of a positive integer $k$ is a way of writing $k$ as a sum of positive integers. We will denote by $p_{n}(k)$ the number of partitions of $k$ with exactly $n$ nonzero integers and by $\tilde{p}_{n}(k)$ the number of partitions of $k$ with at most $n$ nonzero integers. Finally by $p(k)$ we will denote the number of all partitions of the integer $k$. We have the following two Facts:
\begin{fact}\label{fact:partitions1}(\cite{concentration})
Assume $\lambda\in\tilde{\Lambda}_t$. Then $\|\lambda\|_1=2k$, where the integer $k$ satisfies $1\leq k\leq t$. Moreover, the number of distinct irreducible representations $\pi_\lambda$ with $\|\lambda\|_1 =2k$ is given by
\begin{gather}\label{smalld}
\alpha_{2k}=\left\{
	    \begin{array}{ll}
		    p(k)^2  & d\geq 2k,\\
		    \sum_{n=1}^{k} p_n(k)\tilde{p}_{d-n}(k)& k+1 \leq d<2k,\\
      \sum_{n=1}^{d-1} p_n(k)\tilde{p}_{d-n}(k)& 2\leq d\leq k.
	    \end{array}
    \right.
\end{gather}
Moreover, for any $d$ we have $\alpha_{2k}\leq p(k)^2$.
\end{fact}

\begin{fact}\label{fact:partitions-properties}
(\cite{ORUC2016355}) Let $k$ and $n$ be positive integers. Then
\begin{gather}
    \frac{0.0036}{k}e^{\pi\sqrt{\frac{2k}{3}}}\leq p(k)\leq \frac{5.44}{k}e^{\pi\sqrt{\frac{2k}{3}}},\\\nonumber
    p_n(k)\leq \frac{1}{n!(n-1)!}\left(k+\frac{n(n-3)}{4}\right)^{n-1},\,4\leq n\leq k.\\\nonumber
    %p_n(k)\leq \frac{5.44}{k-n}e^{\pi\sqrt{\frac{2(k-n)}{3}}},\,1\leq n\leq k-1\\\nonumber
    %\tilde{p}_n(k)=\sum_{i=1}^np_i(k)\leq \frac{1}{2\sqrt{3}}\left(e^{\pi\sqrt{\frac{2(k-n)}{3}}}\log\left(1-\frac{1}{\pi\sqrt{\frac{2(k-n)}{3}}}\right)-\frac{1}{2}e^{\pi\sqrt{\frac{2k}{3}}}\log\left(1-\frac{2}{\pi\sqrt{\frac{2n}{3}}}\right)\right)\nonumber
    \end{gather}
\end{fact}

\begin{lemma}\label{lemma:dim-bound-new}
Let $\lambda\in\Lambda_t$ be such that $\|\lambda\|_1=2k$ and $l(\lambda) = d \ge 2$. Then $d_\lambda \geq 2k$. Moreover, 
\begin{gather*}
d_\lambda\geq \left\{
	    \begin{array}{ll}
		    \binom{d+1}{2} & 2k\leq d,\\
		    \binom{d}{2} + (d-1)2k & 2k \geq d.\\
	    \end{array}
    \right.
\end{gather*}
\end{lemma}
\begin{proof}
Let  
\begin{gather*}
    \eta(n) := (\underbrace{1, ..., 1}_{n}, \underbrace{0, \ldots, 0 }_{d-n}) \quad \mathrm{for} \quad n=1 , ..., d-1,
\end{gather*}
be the so-called fundamental weights of $SL(d,\mathbb{C})$ \cite{fulton1991representation}. For any $\lambda=(\lambda_1,\ldots,\lambda_d)$ that belongs to  $\Lambda_t$ we define $\lambda^{Y}:=(\lambda_1-\lambda_d,\ldots,\lambda_{d-1}-\lambda_d,0)$. We note that $\lambda^{Y}$ can be viewed as the Young diagram corresponding to the irreducible representation of $SU(d)$, and hence $SL(d,\mathbb{C})$, that arises as the restriction $\pi_\lambda|_{SU(d)}$. One can easily see that $\lambda^Y$ can be written as 
\begin{gather}    
    \lambda^{Y} = \sum_{n=1}^{d-1} \lambda^s_n \eta(n), \quad\, \lambda^s_n:=\lambda_n-\lambda_{n+1}.
\end{gather}
Let us also define
\begin{gather}
h(\lambda) := \sum_{n=1}^{d-1} \lambda^s_n=\lambda_1-\lambda_d.
\end{gather}
Lemmas 2.1 and 2.2 from \cite{goldstein2016} combined with the fact that $d_\lambda=d_{\lambda^Y}$ ensure that
\begin{equation}
\label{eq:dimension_bound}
    d_\lambda \ge \min_{1 \le n \le d-1} d_{h(\lambda) \cdot \eta(n)} = d_{h(\lambda) \cdot \eta(m)},
\end{equation}
where $1 \le m \le d-1$ is such that:
\begin{equation}
\label{eq:minimal_dim}
    \min_{1 \le n \le d-1} d_{\eta(n)} = d_{\eta(m)}.
\end{equation}
\noindent Therefore in order to use \eqref{eq:dimension_bound} we need to find out which  $d_{\eta(n)}$ is the smallest one. For this purpose we use the dimension formula \eqref{eq:weyl_dimension} and get
\begin{gather*}
    d_{\eta(n)} = \prod_{1 \le i < j \le d} \frac{j-i+\eta(n)_i-\eta(n)_j}{j-i}  = \prod_{1 \le i \le n} \prod_{n < j \le d} \frac{j-i+1}{j-i} = \\
    = \prod_{1 \le i \le n} \frac{(d+1-i)!(n-i)!}{(n+1-i)!(d-i)!} = \prod_{1 \le i \le n} \frac{d+1-i}{n+1-i} = \frac{d!}{(d-n)!n!} = \binom{d}{n},
\end{gather*}
which is more or equal to $d$ for $1 \le n \le d-1$. Therefore:
\begin{equation*}
    \min_{1 \le n \le d-1} d_{\eta(n)} \ge d = d_{\eta(1)} = d_{\eta(d-1)},
\end{equation*}
and $m$ from \eqref{eq:dimension_bound} and \eqref{eq:minimal_dim} is equal to $1$ (or $d-1$).
\noindent If $\| \lambda \|_1 = 2k \le d$ then clearly:
\begin{equation}
\label{eq:h_bound1}
    h(\lambda) = \lambda_1-\lambda_d \ge 1-(-1) = 2 \ge \frac{4k}{d},
\end{equation}
and if $2k > d$ then:
\begin{equation*}
    \zeta(n) = (\underbrace{\frac{k}{n}, ..., \frac{k}{n}}_{n}, \underbrace{-\frac{k}{d-n}, ..., -\frac{k}{d-n} }_{d-n}) \quad \mathrm{for} \quad n=1 , ..., d-1,
\end{equation*}
is an element of $\mathbb{R}^d$ with the smallest possible $h(\zeta(n))$ that satisfies $l(\zeta(n)_+)=n$, $\zeta(n)_i \ge \zeta(n)_{i+1}$, $\Sigma(\zeta(n))=0$ and $\| \zeta(n) \|_1 = 2k$. Thus
\begin{equation}
\label{eq:h_bound2}
    h(\lambda) \ge \min_{1\le n \le d-1} h(\zeta(n)) = \min_{1\le n \le d-1} \frac{kd}{n(d-n)} = \frac{kd}{\floor{\frac{d}{2}} \ceil{\frac{d}{2}}} \ge \frac{4k}{d}.
\end{equation}

\noindent Now we can prove our thesis for $2k \le d$. Using \eqref{eq:dimension_bound} and \eqref{eq:h_bound1} we get that:
\begin{gather*}
    d_\lambda \ge d_{h(\lambda) \cdot \eta(1)} \ge d_{2 \cdot \eta(1)} = \prod_{1 \le i < j \le d} \frac{j-i+2(\eta(n)_i-\eta(n)_j)}{j-i} =\\ 
    =\prod_{1 < j \le d} \frac{j-1+2}{j-1} = \frac{(d+1)!}{2!(d-2)!} = \binom{d+1}{2},
\end{gather*}
which is bigger than $d$ when $d \ge 2$. Hence:
\begin{equation*}
    d_\lambda \ge d \ge 2k.
\end{equation*}
In case $2k > d$ we use \eqref{eq:h_bound2} to obtain:
\begin{gather}
\label{eq:def_f}
    d_{h(\lambda) \cdot \eta(1)} = \prod_{1 < j \le d} \frac{j-1+h(\lambda)}{j-1} \ge \prod_{1 < j \le d} \frac{j-1+\frac{4k}{d}}{j-1} = \frac{ \prod_{1 < j \le d} [ d(j-1)+4k ]}{(d-1)!d^{d-1}}=:f(k)
\end{gather}
Function $f$ defined in \eqref{eq:def_f} can be naturally extended to $\mathbb{R}$. Under this extension $f$ is a polynomial of rank $d-1$ with positive coefficients. From \eqref{eq:dimension_bound} we have $d_\lambda \ge f(k)$ so we want to prove that $f(k) \ge 2k$ for any $k \ge \frac{d}{2}$. We first note that at $k = \frac{d}{2}$ this inequality is satisfied:
\begin{equation*}
    f\left(\frac{d}{2}\right) = \binom{d}{2} \ge d.
\end{equation*}
Next we compute the derivative of $f$:
\begin{gather*}
    f'(k) = 4\sum_{i=2}^{d} \frac{ \prod_{\substack{1 < j \le d \\ j \neq i}} [ d(j-1) + 4k ]}{(d-1)!d^{d-1}} = 4f(k) \sum_{i=2}^{d} \frac{1}{d(i-1)+4k} \ge 4f(k) \frac{d-1}{d(d-1)+4k}.
\end{gather*}
Using \eqref{eq:def_f} we get:
\begin{gather*}
    f'(k) \ge \frac{4 \prod_{1 < j \le d-1} [ d(j-1)+4k ]}{(d-2)!d^{d-1}} \ge \frac{4 \prod_{1 < j \le d-1} [ d(j-1)+2d ]}{(d-2)!d^{d-1}} = \frac{4}{d} \binom{d}{2}=2(d-1) > 1,
\end{gather*}
where in the second inequality we used our assumption that $2k > d$. In summary, we have that $f\left(\frac{d}{2}\right) > d$ and that for any $k > \frac{d}{2}$ it holds $f'(k) > 1$. It follows that for $k \ge \frac{d}{2}$:
\begin{equation*}
    f(k) \ge f\left(\frac{d}{2}\right) + f'\left(\frac{d}{2}\right)k \ge \binom{d}{2} + (d-1) 2k > 2k.
\end{equation*}

\end{proof}

\section{Spectral measures and moment operators}\label{sec:spectral_properties}
We say that the measure $\nu_\S$ is symmetric, if for every $U$ in $\S$ the inverse $U^{-1}$ is also in $\S$. For symmetric measures moment operators $T_{\nu_\S}$, $T_{\nu_\S,t}$ and $T_{\nu_\S,\lambda}$ defined by \eqref{def:moment}, \eqref{def:moment-t}  and \eqref{def:moment-lambda} are bounded self-adjoint operators and therefore their spectra are well defined. In this section we are interested in asymptotic properties of these spectra when the size of moment operators grows to infinity.  

Recall that for a self-adjoint $n\times n$ matrix $H_n$ the spectral measure of any interval $[a,b]\subset\mathbb{R}$ is defined by:
\begin{equation*}
    \sigma_{H_n}([a,b]) := \frac{1}{n} (\# \; of \; H_n \; \mathrm{eigenvalues} \;  \mathrm{in} \; [a,b]),
\end{equation*}
For $\eigval{1}{H_n} \le ... \le \eigval{n}{H_n}$ eigenvalues of $H_n$ the $m$-th moment of $\sigma_{H_n}$ is:
\begin{equation*}
    \sigma_{H_n}^{(m)} = \int x^m d\sigma_{H_n}(x) = \frac{1}{n} \sum_{i=1}^n \eigval{i}{H_n}^m = \frac{1}{n} \mathrm{Tr}(H_n^m) = \mathrm{tr}(H_n^m).
\end{equation*}
Recall that a compactly supported measure is determined by its moments. Thus if for a sequence $\{H_n\}_{n=1}^\infty$ of commonly bounded self-adjoint matrices there is measure $\sigma_H$ for which 
\begin{gather}
    \sigma_{H}^{(m)}=\lim_{n\rightarrow \infty}\sigma_{H_n}^{(m)},
\end{gather}
then this measure is unique and 
\begin{gather*}
    \int f(x) d\sigma_{H_n}(x) \xrightarrow{n \rightarrow \infty} \int f(x) d\sigma_{H}(x),
\end{gather*}
where $f$ is any continuous function, that is $\sigma_{H_n}$ converges weakly to $\sigma_{H}$. In case when $H_n$ is a random matrix we define the averaged spectral measure:
\begin{equation*}
    \bar{\sigma}_{H_n}([a,b]) := \mathbb{E}[\sigma_{H_n}([a,b])]
\end{equation*} 
One easily checks that $m$-th moment of $\bar{\sigma}_{H_n}$ is given by
\begin{gather}
    \bar{\sigma}_{H_n}^{(m)}= \int x^m d\bar{\sigma}_{H_n}(x)=\mathbb{E}\mathrm{tr}(H_n^m)=\mathbb{E}\sigma_{H_n}^{(m)}.
\end{gather}
Thus if for a sequence $\{H_n\}_{n=1}^\infty$ of commonly bounded random selfadjoint matrices there is measure $\bar\sigma_H$ for which 
\begin{gather}
    \bar\sigma_{H}^{(m)}=\lim_{n\rightarrow \infty}\bar\sigma_{H_n}^{(m)},
\end{gather}
then this measure is unique and 
\begin{gather*}
    \int f(x) d\bar\sigma_{H_n}(x) \xrightarrow{n \rightarrow \infty} \int f(x) d\bar\sigma_{H}(x),
\end{gather*}
where $f$ is any continuous function, that is $\bar\sigma_{H_n}$ converges weakly to $\bar\sigma_{H}$. 

Let us consider sequence $\{T_{\nu_\S,t}\}_{t=1}^\infty$, where $\S\subset U(d)$ is a finite symmetric set. In order to simplify notation let us denote by $\sigma_{\S,t}$ the spectral measure of $T_{\nu_\S,t}$. Let $\S_k=\{U_1U_2\ldots U_k | U_i\in \S\}$ and $\nu_\S^{\ast k}$ be the $k$-fold convolution of $\nu_\S$. Of course the support of $\nu_\S^{\ast k}$ is given by $\S_k$. We have
\begin{gather}
    \sigma_{\S,t}^{(m)}=\frac{\mathrm{Tr}\left(T^m_{\nu_\S,t}\right)}{d^{2t}}=\frac{1}{d^{2t}}\mathrm{Tr}\left(\int_{U(d)}d\nu_\S^{\ast m}(U) U^{ t,t}\right)= \sum_{U_1,\ldots , U_m \in \S }\nu_\S^{\ast m}(U_1\ldots U_m)\mathrm{tr}\left(\left(U_1\ldots U_m\right)^{t,t}\right).
\end{gather}
Note next that for any $U\in U(d)$ that has eigenvalues $\lambda_j(U)=e^{i\phi_j}$ we have
\begin{gather}
    \mathrm{tr}{U^{t,t}}=|\mathrm{tr}(U)|^{2t}=\left|\frac{e^{\phi_1}+\ldots+e^{i\phi_d}}{d}\right|^{2t}\xrightarrow{t \rightarrow \infty} \left\{
	    \begin{array}{ll}
		    1  & \text{if $U \propto I$ },\\
		    0 & \text{otherwise}.
	    \end{array}
    \right.
\end{gather}
Assuming the group generated by $\S$ is free we get
\begin{gather}
    \mathrm{tr}{(U_1\ldots  U_m )^{t,t}}\xrightarrow{t \rightarrow \infty} \left\{
	    \begin{array}{ll}
		    1  & \text{if $U_1\ldots  U_m =I$ },\\
		    0 & \text{otherwise}.
	    \end{array}
    \right.
\end{gather}
Thus in case when $\nu_\S$ is the uniform measure on $\S$ we have 
\begin{gather}
    \label{T-measure-mom}
    \lim_{t \rightarrow \infty} \sigma_{\S,t}^{(m)} = \frac{1}{|\S|^m} \sum_{\substack{U_1, ... , U_m \in \S \\ U_1 ... U_m = \mathbb{I}}} 1.
\end{gather}
The above summation can be interpreted as the number of walks of length $m$ that begin and end in some vertex $v$ on $|S|$-regular tree. This problem was solved in \cite{kesten} and 
\begin{gather}
     \lim_{t \rightarrow \infty} \sigma_{\S,t}^{(m)}=\left\{
	    \begin{array}{ll}
		    \frac{1}{|\S|^m}\sum_{j=1}^{m/2}{m-j\choose m/2}\frac{j}{m-j}|\S|^j(|\S-1|)^{m/2-j}  & \text{if $m$-even},\label{kesten-moments}\\
		    0 & \text{otherwise}.
	    \end{array}
    \right.
\end{gather}
The author of \cite{kesten} also showed that there is a measure $\sigma_{\S}$ such that $\sigma_{\S}^{(m)}= \lim_{t \rightarrow \infty} \sigma_{\S,t}^{(m)}$. This measure is know as the Kesten-McKay measure \cite{lubotzky,karol1} and is given by:
\begin{gather}
    \label{kesten}
    d\sigma_{\S}(x) = \frac{|\S| \sqrt{\delta_{\mathrm{opt}}^2(\S) - x^2}}{2 \pi (1-x^2)}\mathbf{1}_{[-\delta_{\mathrm{opt}}(\S) ,\delta_{\mathrm{opt}}(\S) ]}dx,
\end{gather}
where $\delta_{\mathrm{opt}}(\S) := \frac{2\sqrt{|\S| - 1}}{|\S|}$. We also note that the $m$-th moment of the spectral measure of $T_{\nu_{\S}}^2$, i.e. $\lim_{t\rightarrow \infty}\frac{\mathrm{Tr}\left(T^{2m}_{\nu_\S,t}\right)}{d^{2t}}$ is given by 
\begin{gather}\label{eq:moments-square}
     \frac{1}{|\S|^{2m}}\sum_{j=1}^{m}{2m-j\choose m}\frac{j}{2m-j}|\S|^j(|\S-1|)^{m-j}.
\end{gather} 
Using the symmetry of measure \eqref{kesten} making the change of variables $y=x^2$ in \eqref{kesten} one can easily see that when $t\rightarrow \infty$ the spectral measure of $T_{\nu_{\S},t}^2$ is given by
\begin{gather}\label{eq:spectral-measure-square}
    \frac{|\S| \sqrt{\delta_{\mathrm{opt}}^2(\S)y - y^2}}{2 \pi (1-y)y}\mathbf{1}_{[0 ,\delta_{\mathrm{opt}}^2(\S) ]}dy.
\end{gather}
One can verify that the $m$-th moment of \eqref{eq:spectral-measure-square}, i.e.
\begin{gather}
    \int_{0}^{\delta^2_{\mathrm{opt}}(\S)}y^{m-1} \frac{|\S| \sqrt{\delta_{\mathrm{opt}}^2(\S)y - y^2}}{2 \pi (1-y)}dy,
\end{gather}
is indeed given by \eqref{eq:moments-square} for every $m\in\mathbb{Z}_{+}$. We are now ready to treat the case when $\mc{S}$ is not symmetric. To this end assume $\S=\{U_1,\ldots,U_n\}$ is such that $\S\cup\S^{-1}$ generates a free group. Repeating similar arguments as for $T_{\nu_\S,t}$, one sees that the $m$-th moment of the spectral measure of $T_{\nu_\S,t}T_{\nu_\S,t}^\ast$ can be calculated as the number of closed paths of the length $2m$, starting and ending at the same vertex of the infinite $|\S|$-regular tree. Thus it is given by the formula \eqref{eq:moments-square}, where $|S|=n$. Therefore the spectral measure of $T_{\nu_\S,t}T_{\nu_\S,t}^\ast$, when $t\rightarrow \infty$ is given by \eqref{eq:spectral-measure-square}. Hence the distribution of singular values of $T_{\nu_\S}$, or in other words the spectral measure of $\left(T_{\nu_\S}T_{\nu_\S}^\ast\right)^{\frac{1}{2}}$ is given by: 
\begin{gather}\label{eq:kesten-final}
     \frac{|\S| \sqrt{\delta_{\mathrm{opt}}^2(\S) - x^2}}{ \pi (1-x^2)}\mathbf{1}_{[0 ,\delta_{\mathrm{opt}}(\S) ]}dx.
\end{gather}
Obviously formula \eqref{eq:kesten-final} gives also the distribution of the singular values of $T_{\nu_\S}$ when $\S$ is symmetric. As a conclusion we get
\begin{fact}
    Let $\S\subset U(d)$ be a finite set and assume that the group generated by $\S\cup\S^{-1}$ is free. Then $\delta(\nu_\S)\in [\delta_{\mathrm{opt}}(\S),1]$.
\end{fact}

Finally we note that similar reasoning can be {\it mutatis mutandis} repeated for $T_{\nu_\S,\lambda}$. Let us denote by $\sigma_{\S,\lambda}$ the spectral measure of $T_{\nu_{\S},\lambda}$. For example when $d=2$ we have  
\begin{gather}
    \sigma_{\S,\lambda}^{(m)}=\frac{1}{d_\lambda}\mathrm{Tr}\left(\int_{U(d)}d\nu_\S^{\ast m}(U) \pi_\lambda(U)\right)=\frac{1}{d_\lambda} \sum_{U_1,\ldots , U_m \in \S }\nu_\S^{\ast m}(U_1\ldots U_m)\chi_\lambda(U_1\ldots U_m),
\end{gather}
where $\chi_\lambda(U)=\mathrm{Tr}(\pi_\lambda(U))$ is the character of the representation $\pi_\lambda$. In case of $d=2$ we have $\lambda=(k,-k)$, $d_\lambda=2k+1$ and
\begin{gather}
    \chi_\lambda(U)=\frac{\sin(2k+1)\theta}{\sin\theta},
\end{gather}
where $\theta$ is such that the spectrum of $U$ is $\{e^{i\theta},e^{-i\theta}\}$. It is clear that when $k\rightarrow \infty$
\begin{gather}
    \lim_{k\rightarrow \infty}\frac{1}{2k+1}\chi_\lambda(U)\xrightarrow{k \rightarrow \infty} \left\{
	    \begin{array}{ll}
		    1  & \text{if $U\propto I$ },\\
		    0 & \text{otherwise}.
	    \end{array}
    \right.
\end{gather}
Therefore, continuing like for $\sigma_{\mc{S},t}$, we obtain that in the limit $k\rightarrow\infty$ the distribution of the singular values of $T_{\nu_\S,\lambda}$ converges weakly to the measure given by \eqref{eq:kesten-final}. Similar, although much more technical reasoning can be also repeated when $d>2$.

%TODO -- write K-M for single $\lambda$
%We next examine moments of 
%\begin{gather}\label{def:t-t}
%   T_{\nu_\S,t,\bar{s}}=\bigoplus_{\lambda\in\Lambda_{t}\setminus \Lambda_{s}}\left (T_{\nu_{\mathcal{S}},\lambda}\right)^{\oplus m_\lambda},
%\end{gather}
%where $t>s$. Matrix \ref{def:t-t} acts $(d^{2t}-d^{2s})$-dimensional space. Let us denote by $\sigma_{\S,t,\bar{s}}$ the spectral measure of \ref{def:t-t}. We have the following relation that follows from the structure of the  $T_{\nu_\S,t}$, $T_{\nu_\S,s}$ and $T_{\nu_\S,t,\bar{s}}$
%\begin{gather}
%\sigma_{\S,t}^{(m)}=\frac{\sigma_{\S,s}^{(m)}}{d^{2(t-s)}}+\frac{d^{2t}-d^{2s}}{d^{2t}}\sigma_{\S,t,\bar{s}}^{(m)}.
%\end{gather}
%Thus 
%\begin{gather}
%\sigma_{\S,t,\bar{s}}^{(m)}=\frac{d^{2t}}{d^{2t}-d^{2s}}\sigma_{\S,t}^{(m)}-\frac{d^{2s}}{d^{2t}-d^{2s}}\sigma_{\S,s}^{(m)}
%\end{gather}
%Hence for any fixed $m$, $t\rightarrow \infty$ and $s=f(t)$, such that, $t>f(t)$ for all $t$ we get
%\begin{gather}
%\sigma_{\S,t,\overline{f(t)}}^{(m)}\xrightarrow{t \rightarrow \infty}\sigma_\S^{(m)}.
%\end{gather}
%In particular 
%\begin{gather}
%    \sigma_{\S,t,\overline{t-1}}^{(m)}\xrightarrow{t \rightarrow \infty}\sigma_\S^{(m)}.
%\end{gather} 
%Note that in case when $d=2$, the operator $ T_{\nu_\S,t,\overline{t-1}}=T_{\nu_S,\lambda}$, where $\lambda=(t,-t)$. Thus we get that the spectral measure of $T_{\nu_S,\lambda}$, denoted by $\sigma_{\S,\lambda}$, also converges weakly to Kesten-McKay distribution $\sigma_\S$ when the dimension $d_\lambda \rightarrow \infty$. This   suggests that for $d>2$ we should expect similar (see appendix for details).
\section{Asymptotic behaviour of $\delta(\nu_\mathcal{S},t)$ for large $\mathcal{S}$}
\label{sec:asymp_beh}

In Section $\ref{sec:spectral_properties}$ we showed that for  a universal (symmetric) set $\mc{S}$ such that ($\mc{S}$) $\S\cup\S^{-1}$ generates a free group, when $t\rightarrow \infty$ the singular values of $T_{\nu_\mc{S},t}$ are distributed in the interval $[0,\delta_{\mathrm{opt}}(\mathcal{S})]$ with the density given by \eqref{eq:kesten-final}. We first note that when $|\mc{S}|\rightarrow \infty$ then $\delta_{\mathrm{opt}}(\mathcal{S}) \rightarrow 0$. Thus in order to capture asymptotic properties of $\delta(\nu_\mc{S},t)$ we normalize it by the factor $\sqrt{\mc{S}}$, that is in the following we will consider operator $\sqrt{\mc{S}}T_{\nu_\mc{S},t}$. The distribution of singular values of this normalized operator can be easily obtained from \eqref{eq:kesten-final} by the change of variables $y=\sqrt{\S}x$ and is given by the density:
\begin{gather}\label{normalized-kesten}
   \frac{1}{\pi(1-\frac{y^2}{\mc{S}})}\sqrt{\frac{4\left(\mc{S}-1\right)}{\mc{S}}-y^2}\mathbf{1}_{[0,2\sqrt{\frac{\mc{S}-1}{\mc{S}}}]}.
\end{gather}
When $\S\rightarrow \infty$ one easily shows that \eqref{normalized-kesten} converges to the quarter-circle distribution with the density function
\begin{gather}\label{semicircle}
   \sigma_{QC}(y)=\frac{1}{\pi}\sqrt{4-y^2}\mathbf{1}_{[0,2]}.
\end{gather}
Unfortunately, this limiting behaviour of the spectrum of $T_{\nu_\mc{S},t}$ (or $T_{\nu_\mc{S},\lambda}$) does not fully determine distribution of the norm. For example, it is known that for all properly normalized random matrix $\beta$-ensembles the limiting spectral measure is given by \eqref{semicircle} but the distributions of the largest/smallest eigenvalues or of the norm depend on $\beta$. In the next sections we determine ensembles that correspond to the distribution of the norm of  $\sqrt{\mc{S}}T_{\nu_\mc{S},\lambda}$ for (symmetric) Haar random  gate-sets. In particular we show that they depend on the type of representation $\lambda$ and on the symmetry of $\mc{S}$.

\subsection{Tools}\label{sec:tools}
In this section we explain tools that will be used in the subsequent sections to study the distribution of $\sqrt{\mc{S}}T_{\nu_\mc{S},\lambda}$  when $\mathcal{S}\rightarrow\infty$. 

\paragraph{Orthogonality relations} Let $\pi_{\lambda}$ and $\pi_\eta$ be two irreducible representation of $U(d)$. It is easy to check that for any linear map $\Psi:V_\lambda\rightarrow V_\eta$ the map
\begin{gather}
    \Psi_0=\E{\mu}\left(\pi_\eta(U^{-1}) \Psi \pi_\lambda(U)\right),
\end{gather}
satisfies $\pi_\eta(V^{-1})\Psi_0\pi_\lambda(V)=\Psi_0$, for any $V\in U(d)$. Thus by Schur's lemma $\Psi_0=0$ if $\eta\neq \lambda$ and $\Psi_0=\frac{\mathrm{Tr}\Psi}{d_\lambda}\bbone_{d_\lambda}=\mathrm{tr} \Psi \bbone_{d_\lambda}$ if $\lambda =\eta$. Choosing $[\Psi]_{j^\prime k^\prime}=\delta_{jj^\prime}\delta_{kk^\prime}$ for some fixed $j$ and $k$ we get\footnote{For a matrix $A$ we denote by $[A]_{ij}$ its entry $(i,j)$.} 

\begin{gather}
    \E{\mu}\left(\left[\pi_\eta(U^{-1})\right]_{ij} \left[\pi_\lambda(U)\right]_{kl}\right)=\frac{\delta_{\eta,\lambda}\delta_{il}\delta_{jk}}{d_\lambda}.
\end{gather}
We also know that there is a basis in $V_\lambda$ such that: 
\begin{enumerate}
    \item $\pi_\lambda(U^{-1})=\pi_\lambda(U)^T$, when $\lambda\in\Lambda_t^r$,
    \item $\pi_\lambda(U^{-1})=\pi_\lambda(U)^\dagger$, when $\lambda\in\Lambda_t^c$.
\end{enumerate}
Therefore 
\begin{gather}\label{eq:orthogonality-real}
     \E{\mu}\left(\left[\pi_\lambda(U)\right]_{ji} \left[\pi_\lambda(U)\right]_{kl}\right)=\frac{\delta_{jk}\delta_{il}}{d_\lambda}\,\,\mathrm{if}\,\lambda\in\Lambda_t^r,
\end{gather}   
and
\begin{gather}\label{eq:orthogonality-complex1}
     \E{\mu}\left(\overline{\left[\pi_\lambda(U)\right]}_{ji} \left[\pi_\lambda(U)\right]_{kl}\right)=\frac{\delta_{jk}\delta_{il}}{d_\lambda}\,\,\mathrm{if}\,\lambda\in\Lambda_t^c.
\end{gather}
In addition to \eqref{eq:orthogonality-complex1} we have one more relation for complex representations. If $\lambda\in \Lambda_{t}^c$ then $\pi_{\lambda^\ast}$ is not equivalent to $\pi_{\lambda}$. On the other hand, in a basis in which $\pi_\lambda(U)$ is unitary we have $\pi_{\lambda^\ast}(U)=\overline{\pi_\lambda(U)}$. Thus 
\begin{gather}\label{eq:orthogonality-complex2}
    0=\E{\mu}\left(\left[\pi_{\lambda^\ast}(U^{-1})\right]_{ij} \left[\pi_\lambda(U)\right]_{kl}\right)= \E{\mu}\left(\left[\overline{\pi_{\lambda}(U^{-1})}\right]_{ij} \left[\pi_\lambda(U)\right]_{kl}\right)=\\\nonumber=\E{\mu}\left(\left[\pi_{\lambda}(U)\right]_{ji} \left[\pi_\lambda(U)\right]_{kl}\right),\,\,\mathrm{for}\,\mathrm{all}\, i,j,k,l,\,\mathrm{if}\,\lambda\in\Lambda_t^c.
\end{gather}
Relations \eqref{eq:orthogonality-real}, \eqref{eq:orthogonality-complex1} and \eqref{eq:orthogonality-complex2} will play a central role in the next sections.

\paragraph{Central Limit Theorem} Let us next note that for a (symmetric) Haar random gate-set $\mc{S}$ the entries of $T_{\nu_{\mc{S}},\lambda}$ are given by sums of $n$ independent random variables. Therefore we recall the following variant of the Central Limit Theorem \cite{feller1957introduction}:
\begin{fact}\label{thm:CLT}
Let $\overline{X}=(X_1,X_2,\ldots,X_d)$ be a random vector in $\mathbb{R}^d$ with finite second moments and $\mathbb{E}(X_i)=0$, for every $i$. Let $\mathrm{Cov}(\overline{X})$ be the covariance matrix of $\overline{X}$, i.e. a matrix whose $ij^{\mathrm{th}}$ entry is 
\begin{gather}
    \Sigma_{ij}:=\left[\mathrm{Cov}(X)\right]_{ij}=\mathbb{E}\left(X_iX_j\right).
\end{gather}
Let 
\begin{gather}
\overline{M}_n=\frac{\overline{X}_1+\overline{X}_2+\ldots+\overline{X}_n}{\sqrt{n}},
\end{gather}
be the sum of $n$ IID (Independent Identically Distributed) copies of $\overline{X}$. Then $\overline{M}_n$ converges in distribution to the random vector $\overline{M}$ which is distributed according to the multivariate normal distribution $N(0,\Sigma)$ with the density
\begin{gather}\label{eq:density-normal}
    f(\overline{M})=\frac{1}{\left((2\pi)^d\det{\Sigma}\right)^{1/2}}e^{-\frac{1}{2}\overline{M}^T\Sigma^{-1} \overline{M}}.
\end{gather}
\end{fact}
\noindent 

\paragraph{Vectorization} In order to change matrices into real vectors we will use the so-called {\it vectorization} technique. 
\begin{itemize}
    \item For any real symmetric $n\times n$ matrix $A$
    \begin{gather}
    \label{eq:vectorization-symmetric}
    \mathrm{vec}_O(A)=\left(A_{11},\ldots,A_{1n},A_{22},\ldots A_{2n},\ldots,A_{nn}\right)^T
\end{gather}
\item For any complex hermitian $n\times n$ matrix $A$
\begin{gather}
    \mathrm{vec}_U(A)=\left(\mathrm{Re}A_{11},\mathrm{Im}A_{11},\ldots,\mathrm{Re}A_{1n},\mathrm{Im}A_{1n},\ldots,\mathrm{Re}A_{nn},\mathrm{Im}A_{nn}\right)^T
\end{gather}
\item For any real $n\times n$ matrix $A$
\begin{gather}
    \mathrm{vec}_R(A)=\left(A_{11},\ldots,A_{1n},A_{21},\ldots A_{2n},\ldots,A_{n1},\ldots,A_{nn}\right)^T
\end{gather}
    \item For any complex $n\times n$ matrix $A$
\begin{gather}
   \mathrm{vec}_C(A)= \left(\mathrm{Re}A_{11},\mathrm{Im}A_{11},\ldots,\mathrm{Re}A_{1n},\mathrm{Im}A_{1n},\ldots,\mathrm{Re}A_{n1},\mathrm{Im}A_{n1},\ldots, \mathrm{Re}A_{nn},\mathrm{Im}A_{nn}\right)^T
\end{gather}   
\end{itemize}
\paragraph{Gaussian and Ginibre ensembles} Next we recall some basic definitions of random matrix ensembles. Let $H_N = (H_{i, j})_{i, j=1}^N$ be $N\times N$ matrix. We distinguish four ensembles:
\begin{itemize}
    \item Gaussian Orthogonal Ensemble (GOE$_N$) real symmetric matrices:
    
    \begin{itemize}
        \item $H_{k, l} = \frac{1}{\sqrt{N}} N(0, (1+\delta_{k, l}))$ for $k\leq l$,
        \item entries $H_{k, l}$  are independent for $k\leq l$,
        \item probability measure: $ d \prob(H_N) \sim \exp{(-\frac{N}{4} \mathrm{Tr}(H_N^2))} \prod_{i \leq j} dH_{i, j} $,
    \end{itemize}
    \item Gaussian Unitary Ensemble (GUE$_N$) hermitian matrices:
    \begin{itemize}
        \item $H_{k, l} = \frac{1}{\sqrt{N}} \left(X_{k,l}+iY_{k,l}\right)$ for $k\neq l$, where $X_{k,l}$ and $Y_{kl}$ are independent $N(0,\frac{1}{2})$ and $H_{k, k}$  is  $N(0,1)$, 
        \item entries $H_{k, l}$  are independent for $k\leq l$,
        \item probability measure: $ d \prob(H_N) \sim \exp{(-\frac{N}{2} \mathrm{Tr}(H_N^2))} \prod_{i \leq j} d X_{i, j}d Y_{i,j} $,
        \end{itemize}
    \item Real Ginibre ensemble RG$_N$:
    \begin{itemize}
        \item $H_{k, l} = \frac{1}{\sqrt{N}} N(0,1)$  
        \item entries $H_{k, l}$  are independent 
        \item probability measure: $ d \prob(H_N) \sim \exp{(-\frac{N}{2} \mathrm{Tr}(H_N^tH_N))} \prod_{1\leq i, j\leq d_\lambda} dH_{i, j} $, 
    \end{itemize}  
    \item Complex Ginibre ensemble (CG$_N$)
    \begin{itemize}
        \item entries: $H_{k, l} = \frac{1}{\sqrt{N}}\left( X_{k,l}+iY_{k,l}\right)$, where $X_{k,l}$ and $Y_{kl}$ are independent $ N(0,\frac{1}{2})$ 
        \item entries $H_{k, l}$  are independent 
        \item probability measure: $ d \prob(H_N) \sim \exp{(-N \mathrm{Tr}(H_N^\dagger H_N))} \prod_{1\leq i,j\leq d_\lambda} dX_{i, j} dY_{i, j}$, 
    \end{itemize}  
\end{itemize}

\paragraph{Bounds on the norm} Using vectorization one has the following correspondence:

\begin{itemize}
    \item $H_N$ form GOE$_N$ -- $\mathrm{vec}_O(H_N)$ is a real Gaussian vector in $\mathbb{R}^{N_O}$, where $N_O=N(N+1)/2$, distributed according to $N(0,\Sigma_O)$, where $\Sigma_O$ is diagonal $N_O\times N_O$ matrix whose diagonal elements are either $\frac{1}{N}$ or $\frac{2}{N}$.
    \item $H_N$ form GUE$_N$ -- $\mathrm{vec}_U(H_N)$ is a real Real Gaussian vector in $\mathbb{R}^{N_U}$, where $N_U=N^2$, distributed according to $N(0,\Sigma_U)$, where $\Sigma_U$ is diagonal $N_U\times N_U$ matrix whose diagonal elements are either $\frac{1}{2N}$ or $\frac{1}{N}$.
    \item $H_N$ form RG$_N$ -- $\mathrm{vec}_R(H_N)$ is a real Gaussian vector in $\mathbb{R}^{N_R}$, where $N_R=N^2$, distributed according to $N(0,\Sigma_R)$, where $\Sigma_R$ is $N_R\times N_R$ matrix given by $\Sigma_R=\frac{1}{N}\bbone_{N_R}$.
    \item $H_N$ form CG$_N$ -- $\mathrm{vec}_C(H_N)$ is a  Real Gaussian vector in $\mathbb{R}^{N_C}$, where $N_C=2N^2$, distributed according to $N(0,\Sigma_C)$, where $\Sigma_C$ is $N_C\times N_C$ matrix given by $\Sigma_C=\frac{1}{2N}\bbone_{N_C}$.
\end{itemize} 
In this paragraph we present short derivation of the bounds on the probability  $\mathbb{P}\left(\|H_N\|>2+\epsilon\right)$ originally presented in \cite{szarek2001Chapter8,SzarekBook}. Recall that the standard real Gaussian vector $X_N$ in $\mathbb{R}^{N}$ is distributed according $N(0,\bbone_N)$. Let $\mathbb{P}_N$ be the probability measure of $N(0,\bbone_N)$. We say that $f:\mathbb{R}^N\rightarrow \mathbb{R}$ is $L$-Lipschitz iff:
\begin{gather}
    |f(X)-f(Y)|\leq L \|X-Y\|_2\,\mathrm{for}\,\mathrm{all}\,X,Y\in\mathbb{R}^N,
\end{gather}
where $\|X\|_2^2=\sum_{i=1}^N X_i^2$. We have the following two useful facts \cite{szarek2001Chapter8,SzarekBook}:
\begin{fact}\label{fact:gauss_concentration}
Let $f:\mathbb{R}^N\rightarrow \mathbb{R}$ be $L$-Lipschitz and $X$ be the standard real Gaussian vector in $\mathbb{R}^{N}$, i.e. $X\simeq N(0,\bbone_N)$. Then for $\epsilon>0$
\begin{gather}
    \mathbb{P}_N\left(f(X_N)>M_{f(X_N)}+\epsilon\right ) \err{\leq \frac{1}{2}\left( 1 - \mathrm{erf}\left(\frac{\epsilon}{\sqrt{2}L} \right) \right)} \leq \frac{1}{2}e^{-\frac{\epsilon^2}{2L^2}},
\end{gather}
where $M_{f(X_N)}$ is the median of the random variable $f(X_N)$ \err{and erf is the error function:
\begin{equation}
    \mathrm{erf}(x) = \frac{2}{\sqrt{\pi}} \int_0^x e^{-t^2}dt.
\end{equation}}
\end{fact}
\begin{fact}
\label{fact:med_exp}
Let $f:\mathbb{R}^N\rightarrow \mathbb{R}$ be a convex function and $X_N$ be the standard real Gaussian vector in $\mathbb{R}^{N}$. Then $M_{f(X_N)}\leq \mathbb{E}\left(f(X_N)\right)$
\end{fact}
\noindent Combining the above two Facts we get 
\begin{corollary}\label{norm_bound}
Let $f:\mathbb{R}^N\rightarrow \mathbb{R}$ be convex and $L$-Lipschitz and $X_N$ be the standard real Gaussian vector in $\mathbb{R}^{N}$. Then
\begin{gather}
    \mathbb{P}_N\left(f(X_N)>\mathbb{E}\left(f(X_N)\right)+\epsilon\right )\leq \frac{1}{2}e^{-\frac{\epsilon^2}{2L^2}}
\end{gather}
\end{corollary}
\noindent Our goal is to bound  $\mathbb{P}\left(\|H_N\|>\mathbb{E}\left(\|H_N\|\right)+\epsilon\right )$, where $H_N$ is from one of the ensembles listed in the previous paragraph. For this propose we consider the function $f_\ast:\mathbb{R}^{N_\ast}\rightarrow \mathrm{Mat}_N$ given by

\begin{gather}\label{eq:f-ast}
    f_\ast(X_{N_\ast}):=\|\mathrm{vec_\ast}^{-1}(\Sigma_\ast^{\frac{1}{2}}X_{N_\ast})\|,
\end{gather}
where $\ast\in\{O,U,R,C\}$. It is easy to see that \eqref{eq:f-ast} is equal to $\|H_N\|$, where $H_N$ form  GOE$_N$, GUE$_N$, RG$_N$ or CG$_N$. Moreover we easily see that
 \begin{gather}
   \left|\|\mathrm{vec}_\ast^{-1}(\Sigma_\ast^{\frac{1}{2}}X_{N_\ast})\|-\|\mathrm{vec}_\ast^{-1}(\Sigma_\ast^{\frac{1}{2}}Y_{N_\ast})\|\right|\leq\|\mathrm{vec}_\ast^{-1}(\Sigma_\ast^{\frac{1}{2}}(X_{N_\ast}-Y_{N_\ast}))\|\leq\\\leq \|\mathrm{vec}_\ast^{-1}(\Sigma_\ast^{\frac{1}{2}}(X_{N_\ast}-Y_{N_\ast}))\|_{\mathrm{HS}}\leq L_\ast\|X_{N_\ast}-Y_{N_\ast}\|_2,
\end{gather}
where $L_O=\sqrt{\frac{2}{N}}$, $L_U=L_R=\frac{1}{\sqrt{N}}$,  $L_C=\frac{1}{\sqrt{2N}}$ and $\|X\|_{\mathrm{HS}}^2=\sum_{i,j=1}^N|X_{ij}|^2$. One can also easily verify that $f_\ast$ is a convex function. Thus making use of Corollary \ref{norm_bound} we get
\begin{gather}
    \mathbb{P}\left(\|H_N\|>\mathbb{E}\left(\|H_N\|\right)+\epsilon\right )=\mathbb{P}_N\left(f_\ast(X_{N_\ast})>\mathbb{E}\left(f_\ast(X_{N_\ast})\right)+\epsilon\right )\leq\frac{1}{2}e^{-\frac{\epsilon^2}{2L_\ast^2}}
\end{gather}
For $H_N$ in GOE$_N$, GUE$_N$, RG$_N$ or CG$_N$ we know that $\mathbb{E}\left(\|H_N\|\right)\leq 2$ for every $N\geq 2$, hence:

\begin{fact}\label{fact:tail_bounds}
(\cite{szarek2001Chapter8,SzarekBook}) We have the following tail bounds:
\begin{itemize}
    \item $   \mathbb{P}\left(\|H_N\|>2+\epsilon\right )\leq\frac{1}{2}e^{-\frac{N\epsilon^2}{4}} =: F_O(N,\epsilon)$, when $H_N$ in GOE$_N$, 
    \item $ \mathbb{P}\left(\|H_N\|>2+\epsilon\right )\leq\frac{1}{2}e^{-\frac{N\epsilon^2}{2}} =: F_U(N,\epsilon)$, when $H_N$ in GUE$_N$, 
    \item $   \mathbb{P}\left(\|H_N\|>2+\epsilon\right )\leq\frac{1}{2}e^{-\frac{N\epsilon^2}{2}} =: F_R(N,\epsilon)$, when $H_N$ in RG$_N$, 
    \item  $   \mathbb{P}\left(\|H_N\|>2+\epsilon\right )\leq\frac{1}{2}e^{-N\epsilon^2} =: F_C(N,\epsilon)$, when $H_N$ in CG$_N$.
\end{itemize}
\end{fact}
\paragraph{Finitely and infinitely often}
We will also need the notion of events that occur infinitely and finitely often. Let $(\Omega,\mathcal{F}, \mathbb{P})$ be the probability space, i.e. $\Omega$ is a space of elementary events aka a sample space, $\mathcal{F}$ is the sigma algebra of events and $\mathbb{P}:\mathcal{F}\rightarrow [0,1]$ is a probability measure. Let $\{A_n\}_{n=1}^\infty$ be a sequence of events. We say that 
\begin{itemize}
    \item Events in the sequence $\{A_n\}_{n=1}^\infty$ occur infinitely often iff $A_n$ occurs for infinite number of indices $n\in\mathbb{Z}_+$. We denote this by $\{A_n\,\,i.o.\}$. More precisely 
    
    $$\{A_n\,\,i.o.\}:=\{\omega\in\Omega:\omega\in A_n\,\mathrm{for\,infinite\, number\, of\, indices}\, n\}.$$
     \item Events in the sequence $\{A_n\}_{n=1}^\infty$ occur finitely often iff $A_n$ occurs for at most finite number of indices $n\in\mathbb{Z}_+$. We denote this by $\{A_n\,\,f.o.\}$. More precisely 
    $$\{A_n\,\,f.o.\}:=\{\omega\in\Omega:\omega\in A_n\,\mathrm{for\,at\,most\, finite\,number\, of\, indices}\, n\}.$$
\end{itemize}
Moreover, using formulas from \cite{feller1957introduction}
\begin{gather}
    \{A_n\,\,i.o.\}=\bigcap_{m=1}^\infty\bigcup_{n\geq m}A_n,\,\,\,\, \{A_n\,\,f.o.\}=\bigcup_{m=1}^\infty\bigcap_{n\geq m}A_n^c,
\end{gather}
it can be verified that both $\{A_n\,\,i.o.\}$ and $\{A_n\,\,f.o.\}$ are elements of $\mathcal{F}$ and $\{A_n\,\,i.o.\}^c=\{A_n\,\,f.o.\}$. We also have:  
\begin{fact}\label{fact:B-C}
(Borel-Cantelli lemmas) Let $\{A_n\}_{n=1}^\infty$ be a sequence of events. Then:
\begin{enumerate}
    \item If $\sum_n\mathbb{P}(A_n)<\infty$ then $\mathbb{P}(\{A_n\,\,f.o.\})=1$.
    \item If  $\sum_n\mathbb{P}(A_n)=\infty$ and events $\{A_n\}_{n=1}^\infty$ are independent then $\mathbb{P}(\{A_n\,\,f.o.\})=0$.  
\end{enumerate}
\end{fact}

\subsection{Haar random gate-sets}\label{sec:asym-proof}
In this section we consider a Haar random  gate-set $\mathcal{S}=\{U_1,\ldots,U_n\}$, where $U_k$'s are independent and Haar random unitaries from $U(d)$. Using tools from Section \ref{sec:tools} we look for the limiting properties of $\sqrt{\mc{S}}T_{\nu_{\mc{S}},\lambda}$ when $\mc{S}\rightarrow \infty$. In particular we show that $\sqrt{\mc{S}}T_{\nu_{\mc{S}},\lambda}$ converges in distribution to a random matrix from either the real Ginibre ensemble if $\lambda\in \Lambda_t^r$ or to a random matrix from the complex Ginibre ensemble if $\lambda\in \Lambda_t^c$.   

\subsubsection{Real representations}
For a Haar random $U$ and $\lambda\in \Lambda_t^r$ 
the vectorization $\mathrm{vec}_R\left(\pi_\lambda(U)\right)$ is a random vector in $\mathbb{R}^{d_\lambda^2}$ with $\E{\mu}\left([\pi_\lambda(U)]_{ij}\right)=0$. Using \eqref{eq:orthogonality-real} we can calculate the covariance matrix of $\mathrm{vec}_R\left(\pi_\lambda(U)\right)$:
\begin{gather}
    \Sigma_{(ij),(kl)}=\left[\mathrm{Cov}\left(\mathrm{vec}_R\left(\pi_\lambda(U)\right)\right)\right]_{(ij),(kl)}=\E{\mu}\left([\pi_\lambda(U)]_{ij}[\pi_\lambda(U)]_{kl}\right)=\frac{\delta_{ik}\delta_{jl}}{d_\lambda}.
\end{gather}
Hence the covariance matrix is $d_\lambda^2\times d_\lambda^2$ matrix given by $\Sigma=\frac{1}{d_\lambda}\bbone_{d_\lambda ^2}$. For a Haar random gate-set $\mc{S}$, using the Central Limit Theorem (see Fact \ref{thm:CLT}) the random operator
\begin{gather}
    \sqrt{\mc{S}}T_{\nu_{\mc{S}},\lambda}=\frac{1}{\sqrt{\mc{S}}}\sum_{U\in \mc{S}}\pi_\lambda(U),
\end{gather}
when $\mathcal{S}\rightarrow \infty$, converges in distribution to a random operator $T_{\lambda}$ whose entries are distributed according to $N(0,\frac{1}{d_\lambda}\bbone_{d_\lambda ^2})$. Thus using \eqref{eq:density-normal} the limiting density is given by the density of the real Ginibre ensemble RG${_{d{_\lambda}}}$, i.e. is proportional to
\begin{gather}
   e^{-\frac{d_\lambda}{2}\mathrm{Tr}(T_\lambda^T T_\lambda)}\prod_{1\leq i, j\leq d_\lambda} d[T_\lambda]_{i j}.
\end{gather}
\subsubsection{Complex representations}\label{subsection:complex1}
For a Haar random $U$ and $\lambda\in \Lambda_t^c$  we first decompose $\pi_\lambda(U)$ into a sum of two matrices $x(U)$ and $y(U)$ corresponding to the real and imaginary parts, i.e. $\pi_\lambda(U)=x(U)+iy(U)$. Obviously, $\E{\mu}\left([x(U)]_{ij}\right)=0=\E{\mu}\left([y(U)]_{ij}\right)$. One easily sees that $\mathrm{vec}_C\left(\pi_\lambda(U)\right)$ is a random vector in $\mathbb{R}^{2d_\lambda^2}$. From \eqref{eq:orthogonality-complex1} and \eqref{eq:orthogonality-complex2} we have:
%\begin{gather}
%\int_{U(d)}d\mu(U)\left[x(U)-iy(U)\right]_{ji} %\left[x(U)+iy(U)\right]_{kl}=\frac{\delta_{jk}\delta_{il}}{d_\lambda}\,\%,\mathrm{if}\,\lambda\in\Lambda_t^c
%\end{gather}
\begin{gather}
\E{\mu}\left(\left[x(U)\right]_{ji} \left[x(U)\right]_{kl}+\left[y(U)\right]_{ji} \left[y(U)\right]_{kl}\right)= \frac{\delta_{jk}\delta_{il}}{d_\lambda},\\
\E{\mu}\left(\left[x(U)\right]_{ji} \left[y(U)\right]_{kl}-\left[y(U)\right]_{ji} \left[x(U)\right]_{kl}\right)=0,\\
\E{\mu}\left(\left[x(U)\right]_{ji} \left[x(U)\right]_{kl}-\left[y(U)\right]_{ji} \left[y(U)\right]_{kl}\right)=0,\\
\E{\mu}\left(\left[x(U)\right]_{ji} \left[y(U)\right]_{kl}+\left[y(U)\right]_{ji} \left[x(U)\right]_{kl}\right)=0.
\end{gather}
Using the above four relations we can calculate the following expectations:
\begin{gather}\label{eq:expected1}
\E{\mu}\left(\left[x(U)\right]_{ji} \left[x(U)\right]_{kl}\right)=\frac{\delta_{jk}\delta_{il}}{2d_\lambda},\\\nonumber
\E{\mu}\left(\left[y(U)\right]_{ji} \left[y(U)\right]_{kl}\right)=\frac{\delta_{jk}\delta_{il}}{2d_\lambda},\\\nonumber
\E{\mu}\left(\left[x(U)\right]_{ji} \left[y(U)\right]_{kl}\right)=0.
\end{gather}
The covariance matrix of $\mathrm{vec}_{C}\left(\pi_\lambda(U)\right)$ is determined by the above expected values and is $2d_\lambda^2\times 2d_\lambda^2$ matrix given by $\Sigma=\frac{1}{2d_\lambda}\bbone_{2d_\lambda ^2}$. Using the Central Limit Theorem (see Fact \ref{thm:CLT}) the random operator
\begin{gather}
    \sqrt{\mc{S}}T_{\nu_{\mc{S}},\lambda}=\frac{1}{\sqrt{\mc{S}}}\sum_{U\in \mc{S}}\pi_\lambda(U),
\end{gather}
when $\mathcal{S}\rightarrow \infty$, converges in distribution to a random operator $T_{\lambda}=X_\lambda+iY_\lambda$ whose real and imaginary parts of entries are distributed according to $N(0,\frac{1}{2d_\lambda}\bbone_{2d_\lambda ^2})$. Thus using \eqref{eq:density-normal} the limiting density is given by the density of the complex Ginibre ensemble CG${_{d{_\lambda}}}$, i.e. is proportional to
\begin{gather}
   e^{-d_\lambda\mathrm{Tr}(T_\lambda^\dagger T_\lambda)}\prod_{1\leq i, j\leq d_\lambda} d [X_\lambda]_{i j}\prod_{1\leq i, j\leq d_\lambda}d[Y_\lambda]_{i,j}.
\end{gather}
   
\subsection{Symmetric Haar random gate-sets}\label{sec:sym-proof}
In this section we consider a symmetric Haar random  gate-set  $\mathcal{S}=\{U_1,\ldots,U_n\}\cup\{U_1^{-1},\ldots,U_n^{-1}\}$, where $U_k$'s are independent and Haar random unitaries from $U(d)$. Using tools from Section \ref{sec:tools} we look for the limiting distribution of $\sqrt{\mc{S}}T_{\nu_{\mc{S}},\lambda}$ when $\mc{S}\rightarrow \infty$.  In particular we show that $\sqrt{\mc{S}}T_{\nu_{\mc{S}},\lambda}$ converges in distribution to a random matrix from the Gaussian orthogonal ensemble if $\lambda\in \Lambda_t^r$ or to a random matrix from the Gaussian unitary ensemble if $\lambda\in \Lambda_t^c$.  
\subsubsection{Real representations}
For a Haar random $U$ and $\lambda\in \Lambda_t^r$ the matrix $\pi_\lambda(U)+\pi_\lambda(U^{-1})$ is a symmetric matrix. Therefore, using \eqref{eq:vectorization-symmetric} the vectorization $\mathrm{vec}_O\left(\pi_\lambda(U)+\pi_\lambda(U^{-1})\right)$ is a random vector in $\mathbb{R}^{d_\lambda(d_\lambda+1)/2}$ with $\E{\mu}\left([\pi_\lambda(U)+\pi_\lambda(U^{-1})]_{ij}\right)=0$. Using \eqref{eq:orthogonality-real} we can calculate the covariance matrix of $\frac{1}{\sqrt{2}}\mathrm{vec}_O\left(\pi_\lambda(U)+\pi_\lambda(U^{-1})\right)$:
\begin{gather}
    \Sigma_{(ij),(kl)}=\frac{1}{2}\left[\mathrm{Cov}\left(\mathrm{vec}_O\left(\pi_\lambda(U)+\pi_\lambda(U^{-1})\right)\right)\right]_{(ij),(kl)}=\\=\frac{1}{2}\E{\mu}\left(\left([\pi_\lambda(U)]_{ij}+[\pi_\lambda(U)]_{ji}\right)\left([\pi_\lambda(U)]_{kl}+[\pi_\lambda(U)]_{lk}\right)\right)=\\
    =\frac{\delta_{ik}\delta_{jl}+\delta_{il}\delta_{jk}}{d_\lambda}.
\end{gather}
Taking into account that $i\leq j$ and $k\leq l$ we see that the covariance matrix is a diagonal matrix. For a Haar random symmetric gate-set  $\mathcal{S}=\{U_1,\ldots,U_n\}\cup\{U_1^{-1},\ldots,U_n^{-1}\}$, using the Central Limit Theorem (see fact \ref{thm:CLT}) the random operator
\begin{gather}
    \sqrt{\mc{S}}T_{\nu_{\mc{S}},\lambda}=\frac{\sqrt{2}}{\sqrt{\mc{S}}}\sum_{i=1}^{n}\frac{\left(\pi_\lambda(U_i)+\pi_\lambda(U_i^{-1})\right )}{\sqrt{2}},
\end{gather}
when $\mathcal{S}\rightarrow \infty$, converges in distribution to a random operator $T_{\lambda}$ whose entries are distributed according to $N(0,\Sigma )$. Thus using \eqref{eq:density-normal} the limiting density is given by the density of the Gaussian orthogonal ensemble GOE${_{d{_\lambda}}}$, i.e. is proportional to
\begin{gather}
   e^{-\frac{d_\lambda}{4}\mathrm{Tr}(T_\lambda^2)}\prod_{1\leq i\leq j\leq d_\lambda} d[T_\lambda]_{i j}.
\end{gather}
\subsubsection{Complex representations}
For a Haar random $U$ and $\lambda\in \Lambda_t^c$  let us first decompose $\pi_\lambda(U)$  ($\pi_\lambda(U^{-1})$) into a sum of two matrices $x(U)$ ($x(U^{-1})$) and $y(U)$ ($y(U^{-1})$) corresponding to the real and imaginary parts, i.e. $\pi_\lambda(U)=x(U)+iy(U)$, $\pi_\lambda(U^{-1})=x(U^{-1})+iy(U^{-1})$. The expected values of $x(U)$, $x(U^{-1})$, $y(U)$, $y(U^{-1})$ are all zeros. We also have
\begin{gather}
    [x(U^{-1})]_{ij}=[x(U)]_{ji},\\\nonumber
    [y(U^{-1})]_{ij}=-[y(U)]_{ji}.
\end{gather}
The covariance matrix of 
\begin{gather}\label{eq:vec2}
    \mathrm{vec}_U\left(\frac{\pi_\lambda(U)+\pi_\lambda(U^{-1})}{\sqrt{2}}\right),
\end{gather}
is determined by 
\begin{gather}
\frac{1}{2}\E{\mu}\left((\left[x(U)\right]_{ij}+\left[x(U)\right]_{ji}) \left(\left[x(U)\right]_{kl}+\left[x(U)\right]_{lk}\right)\right)=\frac{\delta_{ik}\delta_{jl}+\delta_{il}\delta_{jk}}{2d_\lambda},\\\nonumber
\frac{1}{2}\E{\mu}\left((\left[y(U)\right]_{ij}-\left[y(U)\right]_{ji})\left(\left[y(U)\right]_{kl}-\left[y(U)\right]_{lk}\right)\right)=\frac{\delta_{ik}\delta_{jl}-\delta_{il}\delta_{jk}}{2d_\lambda},\\\nonumber
\frac{1}{2}\E{\mu}\left((\left[x(U)\right]_{ij}+\left[x(U)\right]_{ji})\left(\left[y(U)\right]_{kl}-\left[y(U)\right]_{lk}\right)\right)=0,
\end{gather}
where we used relations \eqref{eq:expected1}. Taking into account that $i\leq j$ and $k\leq l$ we see that the  covariance matrix is a diagonal matrix. For a Haar random symmetric gate-set  $\mathcal{S}=\{U_1,\ldots,U_n\}\cup\{U_1^{-1},\ldots,U_n^{-1}\}$, using the Central Limit Theorem (see Fact \ref{thm:CLT}) the random operator
\begin{gather}
    \sqrt{\mc{S}}T_{\nu_{\mc{S}},\lambda}=\frac{\sqrt{2}}{\sqrt{\mc{S}}}\sum_{i=1}^{n}\frac{\left(\pi_\lambda(U_i)+\pi_\lambda(U_i^{-1})\right )}{\sqrt{2}},
\end{gather}
when $\mathcal{S}\rightarrow \infty$, converges in distribution to a random operator $T_{\lambda}=X_\lambda+iY_\lambda$ whose real and imaginary parts of entries are distributed according to $N(0,\Sigma )$. Thus using \eqref{eq:density-normal} the limiting density is given by the density of the Gaussian unitary ensemble GUE${_{d{_\lambda}}}$, i.e. is proportional to
\begin{gather}
   e^{-\frac{d_\lambda}{2}\mathrm{Tr}(T_\lambda^2)} \prod_{1\leq i \leq j\leq d_\lambda} d [X_\lambda]_{i j}\prod_{1\leq i \leq j\leq d_\lambda}d[Y_\lambda]_{i,j}.
\end{gather}

\section{The random matrix model and its properties}\label{sec:model}
As we have seen in the previous sections a matrix $\sqrt{\mc{S}}T_{\nu_\mc{S},\lambda}$ converges, when $\mc{S}\rightarrow \infty$, in distribution to a random matrix $T_\lambda$ from: (1) real Ginibre ensemble when $\mc{S}$ is Haar random and $\lambda\in \Lambda^r_t$, (2) Gaussian orthogonal ensemble when  $\mc{S}$ is symmetric Haar random and $\lambda\in \Lambda^r_t$, (3) complex  Ginibre ensemble when $\mc{S}$ is Haar random and $\lambda\in \Lambda^c_t$, (4) Gaussian unitary ensemble when  $\mc{S}$ is symmetric Haar random and $\lambda\in \Lambda^c_t$. Therefore when $\mc{S}\rightarrow \infty$  the matrix $T_{\nu_{\mathcal{S}},t}$ given by \eqref{eq:delta_t} converges in distribution to a block diagonal matrix $\bbone^{m_0}\oplus\bigoplus_{\lambda\in\Lambda_t}\left (T_{\lambda}\right)^{\oplus m_\lambda}$, with blocks $T_\lambda$ belonging to the ensembles listed above. By similar arguments as in Section \ref{sec:moments}, the norm $\|\bbone^{m_0}\oplus\bigoplus_{\lambda\in\Lambda_t}\left (T_{\lambda}\right)^{\oplus m_\lambda}-\bbone^{m_0}\|$ is the same as the the norm of $T_t=\bigoplus_{\lambda\in\Lambda_t^{\mathrm{ess}}}T_{\lambda}$. It is also easy to see that for any $\lambda\neq \eta$, where $\lambda,\eta \in \Lambda_t^{\mathrm{ess}}$ blocks $T_\lambda$ and $T_\eta$ are independent. This follows directly from the Central Limit Theorem (Fact \ref{thm:CLT}) and the orthogonality relations given in Section \ref{sec:tools}.  We next introduce 
\begin{gather} \delta(\lambda):=\|T_\lambda\|,\,\,\delta(t)=\sup_{\lambda\in\Lambda_t^{\mathrm{ess}}}\delta(\lambda),\,\,\delta=\sup_t\delta(t). 
\end{gather}
Summing up we proved that:
\begin{theorem}\label{thm1}
When $|\mc{S}|\rightarrow\infty$:
\begin{enumerate}
    \item $\sqrt{|\mc{S}|}T_{\nu_\mc{S},\lambda}$ converges in distribution to a random matrix $T_\lambda$ from: (1) the real (complex) Ginibre ensemble RG$_{d_\lambda}$ (CG$_{d_\lambda}$) when $\mc{S}$ is Haar random gate-set and $\pi_\lambda$ is a real (complex) representation, (2) the  Gaussian orthogonal (unitary) ensemble GOE$_{d_\lambda}$ (GUE$_{d_\lambda}$), when  $\mc{S}$ is symmetric Haar random gate-set and $\pi_\lambda$ is a real (complex) representation.
    \item Operator ${\color{black}\sqrt{\S}}\bigoplus_{\lambda\in\Lambda_t^{\mathrm{ess}}}T_{\nu_{\mathcal{S}},\lambda}$  converges in distribution to a block diagonal matrix 
\begin{gather}\label{model1} T_t=\bigoplus_{\lambda\in\Lambda_t^{\mathrm{ess}}}T_{\lambda},
\end{gather}
with blocks $T_\lambda$ belonging to the ensembles listed above. Moreover $T_\lambda$'s with distinct $\lambda$'s are independent. 
\end{enumerate}
\end{theorem}

\noindent Next we charactarize properties of our random matrix model, i.e. the probability distributions of : $\delta(\lambda)$, $\delta(t)$ and $\delta$. The tail bounds for the norm of any block $T_\lambda$ are given in Fact \ref{fact:tail_bounds}, where $N$ should be replaced by $d_\lambda$. Moreover, since blocks are independent we can formulate the following theorem.
\begin{theorem}\label{thm:bounds_T_t}
In the symmetric Haar random setting we have the following tail bounds for $\delta(t)$: 
\begin{gather}
        \label{eq:bounds_T_t}
      \mathbb{P}\left(\delta(t)>2+\epsilon\right )\leq\sum_{\lambda\in\Lambda_t^{\mathrm{r}}}F_{O}(d_\lambda,\epsilon)+\frac{1}{2}\sum_{\lambda\in \Lambda_t^c}F_{U}(d_\lambda,\epsilon)\\
      \label{eq:bounds_T_t_2}
      \mathbb{P}\left(\delta(t)>2+\epsilon\right )\leq 1-\exp\left(-\sum_{\lambda\in\Lambda_t^{\mathrm{r}}}F_O(d_\lambda,\epsilon)-\frac{1}{2}\sum_{\lambda\in\Lambda_t^{\mathrm{c}}}F_U(d_\lambda,\epsilon)\right),
\end{gather}
where functions $F_\ast$ are defined in Fact \ref{fact:tail_bounds}. The bounds for Haar random setting are the same albeit we have to exchange $O\leftrightarrow R$ and $U\leftrightarrow C$.
\end{theorem}
\begin{proof}
The first bound is the union bound, i.e using \eqref{model1} we have
\begin{gather}
\mathbb{P}\left(\|T_t\|>2+\epsilon\right )=\mathbb{P}\left(\bigcup_{\lambda\in\Lambda_t^{\mathrm{ess}}}\left(\|T_\lambda\|>2+\epsilon\right)\right)=\sum_{\lambda\in\Lambda_t^{\mathrm{ess}}}\mathbb{P}\left(\|T_\lambda\|>2+\epsilon\right),
\end{gather}
which combined with Fact \ref{fact:tail_bounds} gives \eqref{eq:bounds_T_t}. 
For the second bound let $p_\lambda=\mathbb{P}\left(\|T_\lambda\|> 2+\epsilon\right)$. Then 
\begin{gather}
\mathbb{P}\left(\|T_t\|>2+\epsilon\right )=\mathbb{P}\left(\bigcup_{\lambda\in\Lambda_t^{\mathrm{ess}}}\left(\|T_\lambda\|>2+\epsilon\right)\right )=1-\mathbb{P}\left(\bigcap_{\lambda\in\Lambda_t^{\mathrm{ess}}}\left(\|T_\lambda\|\leq 2+\epsilon\right)\right )=\\\nonumber =1-\prod_{\lambda\in\Lambda_t^{\mathrm{ess}}} \mathbb{P}\left(\|T_\lambda\|\leq 2+\epsilon\right)=1-\prod_{\lambda\in\Lambda_t^{\mathrm{ess}}} \left(1-p_\lambda \right)\leq 1-\exp\left(-\sum_{\lambda\in\Lambda_t^{\mathrm{ess}}}p_\lambda\right)\leq\\\nonumber\leq 1-\exp\left(-\sum_{\lambda\in\Lambda_t^{\mathrm{r}}}F_O(d_\lambda,\epsilon)-\frac{1}{2}\sum_{\lambda\in\Lambda_t^{\mathrm{c}}}F_U(d_\lambda,\epsilon)\right),
\end{gather}
where in the third equality we used the independence of $T_\lambda$'s. We also note that both bounds coincide when the RHS of \eqref{eq:bounds_T_t} is close to zero.
\end{proof}

\begin{example}\label{ex1} 
Before we proceed further we consider an example of a qubit approximate $t$-design in the symmetric Haar random setting, i.e. when $d=2$. In this case $\Lambda_t=\Lambda_t^r=\{(k,-k):k\in \{1, ..., t\}\}$ and the dimension $d_\lambda=d_{(k,-k)}=2k+1$. Thus
\begin{gather}\label{eq:bound_qubit}
\mathbb{P}\left(\|T_t\|>2+\epsilon\right)\leq\sum_{\lambda\in\Lambda_t^{\mathrm{r}}}F_{O}(d_\lambda,\epsilon)=\frac{1}{2}\sum_{k=1}^te^{-\frac{(2k+1)\epsilon^2}{4}}=\frac{e^{-\epsilon^2/4}(1-e^{-t\epsilon^2/2})}{2(e^{\epsilon^2/2}-1)}. 
\end{gather}
Thus by Theorem \ref{thm:bounds_T_t} we get the following bounds:
\begin{gather}
  \mathbb{P}\left(\delta(t)>2+\epsilon\right )\leq\frac{e^{-\epsilon^2/4}(1-e^{-t\epsilon^2/2})}{2(e^{\epsilon^2/2}-1)},\\
    \mathbb{P}\left(\delta>2+\epsilon\right )\leq \frac{e^{-\epsilon^2/4}}{2(e^{\epsilon^2/2}-1)}.
    \end{gather}
For a fixed $\epsilon>0$ let us consider a series of events:
\begin{gather}
    A_k=\{\delta\left(\left(k,-k\right)\right)>2+\epsilon\}.
\end{gather}
We note that the sum 
\begin{gather}
    \sum_{k=1}^\infty\mathbb{P}(A_k)<\frac{e^{-\epsilon^2/4}}{2(e^{\epsilon^2/2}-1)}<\infty.
\end{gather}
Therefore using the Borel-Cantelli lemmas (see Fact \ref{fact:B-C})  we can deduce  that for any given $\epsilon>0$ we have 
\begin{gather}
    \mathbb{P}\left(\{A_k\,\, f.o.\}\right)=1.
\end{gather}
In other words, for any $\epsilon>0$ with the probability $1$ there is  $\lambda \in \Lambda_t$ such that for all $\lambda^\prime \in \Lambda_{t^\prime}$, where $t^\prime>t$ we have $\delta(\lambda^\prime)\leq 2+\epsilon$. Thus in order to determine the value of $\delta$ it is enough to determine $\delta(\nu_\mc{S},t)$ for some finite $t$. Next, we prove similar result for any $d\geq 2$. 
\end{example}
\begin{theorem}\label{thm:spectral-gap-conjecture}
   For both Haar random and symmetric Haar random settings in any dimension $d\geq2$ with the probability $1$ there exists finite $t\in \mathbb{Z}_{+}$ such that  $\delta=\delta(t)$. 
\end{theorem}

\begin{proof}
    Let $\Lambda^{\mathrm{ess}}:=\bigcup_t\Lambda_t^{\mathrm{ess}}$. We equip $\Lambda^{\mathrm{ess}}$ with the order $\prec$ such that form any $\lambda\in\Lambda_t^{\mathrm{ess}}$ and $\lambda^\prime\in \Lambda_{t^\prime}^{\mathrm{ess}}$ we have $\lambda\prec\lambda^\prime$ when $t<t^\prime$ and when $t=t^\prime$ the order is the lexicographic order. For a fixed $\epsilon>0$ let us consider a sequence of events $\{A_\lambda\}_{\lambda\in\Lambda^{\mathrm{ess}}}$, where
    \begin{gather}
        A_\lambda=\{\delta(\lambda)>2+\epsilon\}.
    \end{gather}
Following the discussion presented for $d=2$ we need to show that
    \begin{gather}
   \mathbb{P}\left(\{A_\lambda\,\, f.o.\}\right)=1.
\end{gather}
Using the Borel-Cantelli lemmas (see Fact \ref{fact:B-C}) and Theorem \ref{thm:bounds_T_t}, in the symmetric Haar random setting, it is enough to show that 
\begin{gather}\label{eq:condition}
   \sum_{\lambda\in\Lambda^{\mathrm{ess}}}\mathbb{P}(A_\lambda)\leq\lim_{t\rightarrow \infty}\left(\sum_{\lambda\in\Lambda_t^{\mathrm{r}}}F_{O}(d_\lambda,\epsilon)+\frac{1}{2}\sum_{\lambda\in \Lambda_t^c}F_{U}(d_\lambda,\epsilon)\right)<\infty.
\end{gather}
For the Haar random setting we have to exchange $O\leftrightarrow R$ and $U\leftrightarrow C$. First we note that the number of distinct irreducible representations $\pi_\lambda$ with $\|\lambda\|_1 =2k$ is given by the $\alpha_{2k}$ from Fact \ref{fact:partitions1}. By Fact \ref{fact:partitions-properties} for a fixed $n$ the function $p_n(k)$ grows like $O(k^{n-1})$. Therefore $\alpha_{2k}$ is bounded by $c_dk^{d-2}$, where $c_d$ is a positive constant that can depend on $d$. Next, using Fact \ref{lemma:dim-bound-new} we can bound $d_\lambda\geq 2k$. Thus \eqref{eq:condition} is bounded from above by
\begin{gather}\label{eq:spectral-gap-bound}
\sum_{\lambda\in\Lambda^{\mathrm{ess}}}\mathbb{P}(A_\lambda)\leq\sum_{k=1}^\infty c_dk^{d-2}e^{-\frac{1}{2}dk\epsilon^2}, 
\end{gather}
We note that for $d>2$
\begin{gather}\label{eq:int-bound}
    \int_{0}^\infty dk\left(c_d k^{d-2}e^{-\frac{1}{2}k\epsilon^2}\right)=\frac{2^{d-1}c_d(d-2)!}{\epsilon^{2(d-1)}}.
\end{gather}
Hence \eqref{eq:spectral-gap-bound} is is evidently finite for any $d$. 
\end{proof}

\begin{theorem}
In the (symmetric) Haar random setting in any dimension $d>2$ and for any $\epsilon> c \left(\frac{4\pi^2}{3d(d-1)^2}\right)^{1/4} $ we have 
\begin{gather*}
    \mathbb{P}\left(\delta>2+\epsilon\right ) \leq
    e^{-\frac{d(d+1)}{4c^2} \epsilon^2} \left( e^{2\pi\sqrt{\frac{d+2}{3}}} \frac{60+100\pi\sqrt{3d+6}}{d+2} - b \right) + \\
    +\frac{60}{d^2} \left( 2 + \frac{\sqrt{2\pi}c}{\epsilon} \right) e^{-\frac{3d(d-1)}{4c^2}\epsilon^2 + 2\pi\sqrt{\frac{d}{3}}},
\end{gather*}
where 
\begin{equation*}
    b = 10 \left[8\pi^2 \Ei{2\sqrt{\frac{2}{3}}\pi} - e^{2\sqrt{\frac{2}{3}}\pi}(3+2\sqrt{6}\pi) \right] \approx 3855.93,
\end{equation*}
and $c=1$ in the Haar random setting and $c=\sqrt{2}$ in the symmetric Haar random setting.
\end{theorem}
\begin{proof}
We will prove our thesis for the Haar random setting. The proof for the symmetric Haar random case is analogous.
Recall that  
\begin{gather}\label{eq:proof1}
\mathbb{P}\left(\delta>2+\epsilon\right )\leq  \sum_{\lambda\in\Lambda^{\mathrm{ess}}}\mathbb{P}(\delta(\lambda)>2+\epsilon).
\end{gather}
Let $\alpha_{2k}$ be as in the Fact \ref{fact:partitions1}. Using Fact \ref{fact:partitions-properties} we can write 
\begin{gather}
\alpha_{2k}\leq \frac{30}{k^2}e^{2\pi\sqrt{\frac{2k}{3}}}.
\end{gather}
Using the bound given in Lemma \ref{lemma:dim-bound-new} we can bound  \eqref{eq:proof1} from above by 
\begin{gather}\label{eq:bound2}
% \sum_{\lambda\in\Lambda^{\mathrm{ess}}} \exp(-d\lambda \epsilon^2) \leq
30 \left[ \sum_{k=1}^{\floor{d/2}} \frac{\exp\left(-\frac{d(d+1)}{4} \epsilon^2 + a\sqrt{k}\right)}{k^2} + \sum_{k=\ceil{d/2}}^\infty \frac{\exp \left(-(d-1)\left(k + \frac{d}{4}\right) \epsilon^2 + a\sqrt{k}\right)}{k^2} \right],
\end{gather}
where $a=2\pi\sqrt{\frac{2}{3}}$. First, let us consider the first summation. We note that the function $f(x) = \exp(a\sqrt{x})/x^2$ is growing for all $x$ greater than $\frac{16}{a^2} < 1$. Thus:
\begin{gather} \label{eq:bound_thm_ine1}
    30 e^{-\frac{d(d+1)}{4} \epsilon^2} \sum_{k=1}^{\floor{d/2}} \frac{e^{a\sqrt{k}}}{k^2} \le 30 e^{-\frac{d(d+1)}{4} \epsilon^2} \int_{1}^{(d+2)/2} \frac{e^{a\sqrt{x}}}{x^2} dx.
    % = \frac{60 e^{-\frac{d(d+1)}{4} \epsilon^2}}{a^2} \left( a \sqrt{k} - 1 \right) e^{a\sqrt{k}} %|_1^{\ceil{d/2}} = \\
    %= \frac{60 e^{-\frac{d(d+1)}{4} \epsilon^2}}{a^2} \left( a \sqrt{\ceil{\frac{d}{2}}} - 1 \right)
\end{gather}
The integral above is equal to
\begin{equation*} %\label{eq:bound_thm_int}
    \frac{1}{3} \left[ e^{a} (3 + 2\sqrt{6}\pi) - e^{2\sqrt{\frac{d+2}{3}}\pi} \frac{6+4\sqrt{3d+6}\pi}{d+2} - 8\pi^2 \Ei{a} + 8\pi^2 \Ei{2\sqrt{\frac{d+2}{3}}\pi} \right],
\end{equation*}
where
\begin{equation*}
    \Ei{x} = \int_{-\infty}^x \frac{e^t}{t}dt,
\end{equation*}
is the exponential integral function. We observe that for $g(x) = p\frac{\exp(x)}{x} - \Ei{x}$ and some $p > 1$ we have that:
\begin{equation*}
    g'(x) = \left[ (p-1)x-p \right] \frac{e^x}{x^2}
\end{equation*}
is $0$ at $x = x_p =\frac{p}{p-1}$ and positive for $x \ge x_p$. Moreover for $p=\frac{3}{2}$ we have that $g\left( x_p \right) > 0$ and $x_p = 3 \le 2\sqrt{\frac{d+2}{3}}\pi $ thus for $x>3$ it holds $\frac{3\exp(x)}{2x} > \Ei{x}$ and \eqref{eq:bound_thm_ine1} is bounded by:
\begin{gather*}
      e^{-\frac{d(d+1)}{4} \epsilon^2} \left[ e^{2\sqrt{\frac{d+2}{3}}\pi} \frac{60+100\pi\sqrt{3d+6}}{d+2} - b \right],\\
    b = 10 \left[8\pi^2 \Ei{2\sqrt{\frac{2}{3}}\pi} - e^{2\sqrt{\frac{2}{3}}\pi}(3+2\sqrt{6}\pi) \right] \approx 3855.93.
\end{gather*}
Next, we will deal with the second summation in \eqref{eq:bound2}. In order to change summation to integration we first note that the function $h(x)=\exp \left(-(d-1)\left(x + \frac{d}{4}\right) \epsilon^2 + a\sqrt{x}\right)$ attains maximum at $x=\frac{2\pi^2}{3(d-1)^2\epsilon^4}$. Thus if we assume $\epsilon > \epsilon_d := \left(\frac{4\pi^2}{3d(d-1)^2}\right)^{1/4}$ then the function $h$ is decreasing for $x \ge \frac{d}{2}$. Therefore when $\epsilon > \epsilon_d$ the function $h(x)/x^2$ is also decreasing for $x > \frac{d}{2}$. Assume that $\epsilon>\epsilon_d$. Then the second sum in \eqref{eq:bound2} is bounded by 
\begin{gather}
    \nonumber
    \frac{120}{d^2} h\left(\frac{d}{2}\right) + 60 e^{-\frac{d(d-1)}{4} \epsilon^2} \int_{\sqrt{d/2}}^\infty \frac{e^{-(d-1)\epsilon^2x^2+a x}}{x^3} dx \leq\\
    \nonumber
    \leq \frac{120}{d^2} h\left(\frac{d}{2}\right) + 60  \left( \frac{2}{d} \right)^\frac{3}{2} e^{-\frac{d(d-1)}{4} \epsilon^2} \int_{\sqrt{d/2}}^\infty e^{-(d-1)\epsilon^2x^2+a x} dx= \\
    \nonumber
    = \frac{120}{d^2} h\left(\frac{d}{2}\right) + 60 \left( \frac{2}{d} \right)^\frac{3}{2} e^{-\frac{d(d-1)}{4} \epsilon^2} \int_{\sqrt{d/2}}^\infty \exp{\left(-(d-1)\epsilon^2 \left( x - \frac{a}{2(d-1)\epsilon^2} \right)^2 + \frac{a^2}{4(d-1)\epsilon^2} \right)} dx =\\
    \label{eq:bound_thm_ine2}
    = \frac{120}{d^2} h\left(\frac{d}{2}\right) + 60 \left( \frac{2}{d} \right)^\frac{3}{2} e^{-\frac{d(d-1)}{4} \epsilon^2 + \frac{a^2}{4(d-1)\epsilon^2}} \int_{\sqrt{d/2}}^\infty \exp{\left(-\frac{\left(x-\mu\right)^2}{2\sigma^2}\right)}dx,
\end{gather}
where $\mu=\frac{a}{2(d-1)\epsilon^2}$ and $\sigma^2=\frac{1}{2(d-1)\epsilon^2}$. Next, for a Gaussian random variable $X\sim N(\mu,\sigma^2)$ we use Fact \ref{fact:gauss_concentration} and obtain
\begin{gather}\label{eq:gauss_concent}
    \mathbb{P}(X\geq \mu+\beta)\leq \frac{1}{2}e^{-\frac{\beta^2}{2\sigma^2}}.
\end{gather}
Thus using \eqref{eq:gauss_concent} we get
\begin{gather*}
    \int_{\sqrt{d/2}}^\infty \exp{\left(-\frac{\left(x-\mu\right)^2}{2\sigma^2}\right)}dx= \sqrt{2\pi}\sigma \mathbb{P} \left( X \geq \sqrt{\frac{d}{2}} \right) \leq \\
    \leq \frac{\sqrt{\pi}}{2\sqrt{d-1}\epsilon} \exp{\left( - (d-1) \epsilon^2 \left( \sqrt{\frac{d}{2}} - \frac{a}{2(d-1)\epsilon^2} \right)^2 \right)},
\end{gather*}
and \eqref{eq:bound_thm_ine2} is bounded by
\begin{equation*}
    \frac{120}{d^2} e^{-\frac{3}{4}d(d-1)\epsilon^2 + a\sqrt{\frac{d}{2}}} + \frac{30\sqrt{\pi}}{\sqrt{d-1}\epsilon} \left( \frac{2}{d} \right)^\frac{3}{2} e^{-\frac{3}{4}d(d-1)\epsilon^2 + a\sqrt{\frac{d}{2}}} 
    \le \frac{60}{d^2} \left( 2 + \frac{\sqrt{2\pi}}{\epsilon} \right) e^{-\frac{3}{4}d(d-1)\epsilon^2 + a\sqrt{\frac{d}{2}}}
    % \frac{120\sqrt{2}}{d^{\frac{3}{2}}}e^{-\frac{1}{2}d\epsilon^2+a\sqrt{\frac{d}{2}}} \left( \frac{1}{\sqrt{2d}} e^{-\frac{3}{4}d \left(d - \frac{5}{3}\right)\epsilon^2} + \frac{\sqrt{\pi}}{3} \right) \le \frac{20\sqrt{2}\left( 3 + 2\sqrt{\pi} \right)}{d^{\frac{3}{2}}}e^{-\frac{1}{2}d\epsilon^2+a\sqrt{\frac{d}{2}}}  .
\end{equation*}
Finally, we have:
\begin{gather*}
    \mathbb{P}\left(\delta>2+\epsilon\right ) \leq
    e^{-\frac{d(d+1)}{4} \epsilon^2} \left( e^{2\pi\sqrt{\frac{d+2}{3}}} \frac{60+100\pi\sqrt{3d+6}}{d+2} - b \right) + \\
    +\frac{60}{d^2} \left( 2 + \frac{\sqrt{2\pi}}{\epsilon} \right) e^{-\frac{3}{4}d(d-1)\epsilon^2 + 2\pi\sqrt{\frac{d}{3}}}.
\end{gather*}

\end{proof}
Combining Theorem \ref{thm:spectral-gap-conjecture} with the Conjecture \ref{conjecture} we can say that for any $\epsilon>0$ and any (symmetric) Haar random $\mc{S}$ there exists, with the probability $1$ a weight $\lambda\in\Lambda_t$ such that for all $\lambda^\prime \in\Lambda_{t^\prime}$, where $t^\prime>t$ we have $\delta(\nu_\mc{S},\lambda^\prime)\leq \frac{2+\epsilon}{\sqrt{\mc{S}}}$. Thus in order to determine the value of $\delta(\nu_\mc{S})$ it is enough to determine $\delta(\nu_\mc{S},t)$ for some finite $t$. We note, however, that the required value $t$ can depend on $\mc{S}$ and using our methods we cannot decide if $\mathrm{sup}_{\mc{S}}t(\mc{S})$ is finite.
\begin{corollary}
Assume Conjeture \ref{conjecture} holds. Let $\mc{S}$ be a (symmetric) Haar random gate-set. Then with the probability $1$ there exists finite $t\in \mathbb{Z}_{+}$ such that  $\delta(\nu_{\S})=\delta(\nu_{\mc{S}},t)$. Moreover, for any $d \ge 2$ we have
\begin{gather*}
    \mathbb{P}\left(\delta(\nu_\mc{S})>\frac{2+\epsilon}{\sqrt{\mc{S}}}\right )\leq \frac{e^{-\epsilon^2/(2c^2)}}{2(e^{\epsilon^2/c^2}-1)},\;\mathrm{if}\;d=2\;\mathrm{and}\;\epsilon>0,\\
    \mathbb{P}\left( \delta(\nu_\mc{S}) > \frac{2 + \epsilon}{\sqrt{\mc{S}}} \right) \leq
    e^{-\frac{d(d+1)}{4c^2} \epsilon^2} \left( e^{2\pi\sqrt{\frac{d+2}{3}}} \frac{60+100\pi\sqrt{3d+6}}{d+2} - b \right) + \\
    +\frac{60}{d^2} \left( 2 + \frac{\sqrt{2\pi}c}{\epsilon} \right) e^{-\frac{3d(d-1)}{4c^2}\epsilon^2 + 2\pi\sqrt{\frac{d}{3}}}, \; \mathrm{if} \; d \ge 2 \; \mathrm{and} \; \epsilon>c \left(\frac{4\pi^2}{3d(d-1)^2}\right)^{1/4},
\end{gather*}
where 
\begin{equation*}
    b = 10 \left[8\pi^2 \Ei{2\sqrt{\frac{2}{3}}\pi} - e^{2\sqrt{\frac{2}{3}}\pi}(3+2\sqrt{6}\pi) \right] \approx 3855.93,
\end{equation*}
and $c=1$ in the Haar random setting and $c=\sqrt{2}$ in the symmetric Haar random setting.
\end{corollary}

Using above inequalities we can make some conclusions about the benefits of adding inverses to our sets of gates. The inequalities are of the form:
\begin{gather*}
    \mathbb{P}\left(\delta(\nu_\mc{S})>\frac{2+\epsilon}{\sqrt{\mc{S}}}\right ) \leq F \left( \frac{\epsilon}{c} \right),
\end{gather*}
where $F$ is an appropriate function. We can make the dependence on $\mc{S}$ more explicit by a substitution $\epsilon \mapsto \frac{\epsilon}{\sqrt{\mc{S}}}$:
\begin{gather*}
    \mathbb{P}\left(\delta(\nu_\mc{S})>\frac{2}{\sqrt{\mc{S}}} + \epsilon \right ) \leq F \left( \frac{\sqrt{S}}{c} \epsilon \right).
\end{gather*}
Note that the constant $c$ is such that, in both settings i.e. for a Haar random set $\mc{S} = n$ and for a symmetric Haar random set $\mc{S} = 2 \times n$, we get the same result:
\begin{gather*}
    \mathbb{P}\left(\delta(\nu_\mc{S})>\frac{2}{\sqrt{\mc{S}}} + \epsilon \right ) \leq F \left( \sqrt{n} \epsilon \right).
\end{gather*}
Thus we can conclude that the concentration around $\frac{2}{\sqrt{\mc{S}}}$ is the same in both scenarios. The advantage of the symmetric set lays in the fact that this concentration is around the value that is $\sqrt{2}$ times smaller than in the Haar random setting.

%The main conjecture of this paper is
%\begin{Conjecture}\label{conjecture}
%There is universal constant $N_0\in \mathbb{Z}_+$ which might depend on $d$ such that for any (symmetric) Haar random $\mc{S}\subset U(d)$ with $|\mc{S}|>N_0$ we have: 
%\begin{gather}
 %   \mathbb{P}\left(\|T_{\nu_\mc{S},\lambda}\|>\frac{2+\epsilon}{\sqrt{\mc{S}}}\right )\leq\mathbb{P}\left(\|T_\lambda\|>2+\epsilon\right)\\
%    \mathbb{P}\left(\|T_{\nu_{\mc{S}},t}\|>\frac{2+\epsilon}{\sqrt{\mc{S}}}\right )\leq\mathbb{P}\left(\|T_t\|>2+\epsilon\right )\\
 %   \mathbb{P}\left(\|T_{\nu_{\mc{S}}}\|>\frac{2+\epsilon}{\sqrt{\mc{S}}}\right )\leq\mathbb{P}\left(\|T\|>2+\epsilon\right )
%\end{gather}
%\end{Conjecture}

%\section{Moment operators an{gat}nd efficiency of quantum gates}

%-- introduction to momemt operators
%-- proof of kesten-McKay for fixed set of gates
%-- proof of Kesten-McKay for expected value
%-- convergance of KM to semicircle -- show it is fast with S using TV distance
%-- convergance of T to Gaussian RMT with growing S + assymtoti independence of blocks
%-- discuss speed of coonvergnece using berry-esseen type theorem
%-- formulate cojecture and its implications -- concrete bounds on $\delta$ and designs
%-- concrete random S-K theorem
%-- bounds on number of generators
%--discussion with fourier transform etc
%\section{Maybe circuits with structure -- how do they differ from fully random Haar}
%-- 2-mode beamsplitters
%-- 2-qubit gates in n-qubits

\section{Numerical computations}\label{sec:numerics}

In order to justify Conjecture \ref{conjecture} we numerically sampled $\|T_{\nu_\mc{S},\lambda}\|$ and $\|T_{\nu_\mc{S},t}\|$ in the following scenarios:

\begin{itemize}
    \item $d=2$, all $\lambda$ from $\Lambda_{500}^\mathrm{ess}$, $t=500$ {\color{black} which means irreducible representations of $U(2)$ with the dimension $3 \leq d_\lambda \leq 1001$}:
    \begin{enumerate}
        \item Haar random sets with sizes from $3$ to $13$,
        \item symmetric Haar random sets with sizes from $2\times2$ to $2\times13$,
    \end{enumerate}
    \item $d=4$, all $\lambda$ from $\Lambda_{6}^\mathrm{ess}$, $t=6$ {\color{black} which means irreducible representations of $U(4)$ with the dimension $15 \leq d_\lambda \leq 5200$}:
    \begin{enumerate}
        \item Haar random sets with sizes from $3$ to $10$,
        \item symmetric Haar random sets with sizes from $2\times2$ to $2\times10$,
    \end{enumerate}
\end{itemize}
with sample sizes $2500-5000$. In the subsequent section we explain how we performed those computations and present their results.

% present in this section results of our numerical computations. 

\subsection{Algorithm}
A sketch of the algorithm we used to perform our computations is shown in Listing \ref{code} in a form of a Python code.

\begin{lstlisting}[language=Python, caption=Caption, label=code]
def t_design_norms(d, t, sample_size, set_size, is_symmetric):
    weights = t_design_weights(d, t)
    reps = [URepresentation(w) for w in weights]
    all_norms = []
    
    for i in range(sample_size):
        norms = []
        gates = [get_random_U(d) for j in range(set_size)]
        
        for pi in reps:
            T = sum(pi(U) for U in gates) / set_size
            
            if is_symmetric:
                T = (T + T.transpose().conjugate()) / 2
                
            norms.append(norm(T))
        
        all_norms.append(norms)
        
    return weights, all_norms
\end{lstlisting}
The main function \textbf{t\_design\_norms} computes \textbf{sample\_size} many times norms of $T_{\nu_\mathcal{S}, \lambda}$ for a random $\mathcal{S}$ and for all $\lambda \in \Lambda^\mathrm{ess}_t$. Then it returns two lists. One is a list of weights - \textbf{weights} used in a computation and the other is a two-dimensional list of norms with \textbf{sample\_size} many rows and $|\tilde{\Lambda}_t|$ many columns - \textbf{all\_norms}. Each row in \textbf{all\_norms} represents a different $\mathcal{S}$ and contains norms $\delta(\nu_S, \lambda^1)$, $\delta(\nu_S, \lambda^2)$, $\delta(\nu_S, \lambda^3)$, ... where the order of $\lambda^i$ is the same as in the list \textbf{weights}.\\
More precisely, \textbf{t\_design\_norms} function takes arguments:
\begin{itemize}
    \item \textbf{d} - group dimension, an integer bigger than $1$,
    \item \textbf{t} - scale of a $t$-design that defines set of weights $\tilde{\Lambda}_t$, an integer bigger than $0$,
    \item \textbf{sample\_size} - the number of different random $\mathcal{S}$ for which the computation is performed, an integer bigger than $0$,
    \item \textbf{set\_size} - the size of $\mathcal{S}$ if \textbf{is\_symmetric} is False and a half of the size of $\mathcal{S}$ otherwise, an integer bigger than $0$,
    \item \textbf{is\_symmetric} - an answer to a question: \textit{Is the set $\mathcal{S}$ symmetric?}, a boolean,
\end{itemize}
and returns:\\

\begin{tabular}{ccccccccc}
    weights & \multicolumn{2}{c}{=} & [  & $\lambda^1$, & $\lambda^2$, & $\lambda^3$, & ... & ] \\
    all\_norms & = & [ & \multicolumn{6}{c}{} \\
    &&& {[} & $\delta(\nu_{S_1}, \lambda^1)$, & $\delta(\nu_{S_1}, \lambda^2)$, & $\delta(\nu_{S_1}, \lambda^3)$ , & ... &  {]}, \\
    &&& {[} & $\delta(\nu_{S_2}, \lambda^1)$, & $\delta(\nu_{S_2}, \lambda^2)$, & $\delta(\nu_{S_2}, \lambda^3)$ , & ... &  {]}, \\
    &&&&& \vdots &&& \\
    &&& {[} & $\delta(\nu_{S_N}, \lambda^1)$, & $\delta(\nu_{S_N}, \lambda^2)$, & $\delta(\nu_{S_N}, \lambda^3)$ , & ... &  {]} \\
    && {]} &&&&&& \\
    \multicolumn{9}{c}{}
\end{tabular}
\\
where $N$ = \textbf{sample\_size}.
Then when we wanted to compute $\delta(\nu_{S}, t)$ we used formula \ref{eq:delta_t_sup}.
\\
To use the code from Listing \ref{code} one needs to implement also:
\begin{itemize}
    \item \textbf{t\_design\_weights} - a function that for a given \textbf{d} and  \textbf{t} returns a list of weights from $\Lambda_t^\mathrm{ess}$. Its implementation can be easily derived from the definition of $\Lambda_t$ \eqref{eq:big-lambda} and $\Lambda_t^\mathrm{ess}$ \eqref{eq:def_lam_ess}.
    \item \textbf{URepresentation} - a class (or a function) that takes as argument a weight $\lambda$ and computes $\pi_\lambda$. It was implemented using the fact that $\pi_\lambda= \exp \circ \, \rho_\lambda \circ \log $ where $\rho_\lambda$ is a $\mathfrak{u}(d)$ algebra representation with the highest weight $\lambda$. The representation $\rho_\lambda$ was implemented using the Gelfand-Tsetlin construction from \cite{barut}.
    \item \textbf{get\_random\_U} - a function that for a given \textbf{d} computes a Haar-random matrix from $U(d)$. The implementation of such a function was described in \cite{mezzadri}.
    \item \textbf{norm} - a function that computes an operator norm. We used the \textbf{linalg.norm} from the scipy package.
\end{itemize}

\subsection{Results}\label{sec:results}

% - plots of delta for qubit 2qubits and 3 qubits for some values of $t$
% - plots of moments and conjecture regarding moments - exact calculation of low order moments

In Figure \ref{fig:conj} one can see the relations between probabilities from Conjecture \ref{conjecture}:
\begin{enumerate}
    \item[A] $\mathbb{P}\left(  \delta\left( \nu_\mc{S},\lambda \right) > \frac{2 + \epsilon}{\sqrt{\mc{S}}} \right)$,
    \item[B] $\mathbb{P}\left( \delta\left( \lambda \right) > 2 + \epsilon \right)$,
    \item[C] tail bounds $F_*(d_\lambda,\epsilon)$ from Fact \ref{fact:tail_bounds}.
\end{enumerate}
Our data clearly indicates that the inequalities from our conjecture A $\le$ B $\le$ C hold. Moreover, the data from Figures \ref{fig:averages} and \ref{fig:renorm} shows that for $d_\lambda \rightarrow \infty$ we have

\begin{gather*}
    \mathbb{E} \left( \sqrt{\mathcal{S}} \delta\left( \nu_\mc{S},\lambda \right) \right) \rightarrow 2\sqrt{\frac{\mathcal{S} - 1}{\mathcal{S}}}, \\
    \mathbb{E} \left( \delta\left( \lambda \right) \right) \rightarrow 2.
\end{gather*}
Although both limits are the same when $\mathcal{S} \rightarrow \infty$ the first one is always strictly smaller than the second one for any finite $\mc{S}$. On the other hand, standard deviations - "widths" - of those distributions in the same limit go to $0$. Combining those two facts we get that with growing $d_\lambda$ distributions of $\sqrt{\mathcal{S}} \delta\left( \nu_\mc{S},\lambda \right) $ and $ \delta\left( \lambda \right) $ concentrate around $2\sqrt{\frac{\mathcal{S} - 1}{\mathcal{S}}}$ and $2$ respectively.
This leads us to the conclusion that our conjecture is true even for representations with much higher dimensions than those we considered.
{\color{black} Note that the argument here is not concentrated solely on the fact that $ \sqrt{\mathcal{S}} \delta\left( \nu_\mc{S},\lambda \right) \xrightarrow[\S \rightarrow \infty]{} \delta\left( \lambda \right)$ but rather on the observation that this limit is approached from below. Thus the slower rate of convergence for large $d_\lambda$ actually works in favor of our conjecture because it sharpens the inequality A $\leq$ B.}

\err{
While the above reasoning gives arguments for our conjecture for $\epsilon \approx 0$ it is not clear that those inequalities still hold if $\epsilon$ is large. First, we observe that the median (see Figure \ref{fig:medians}):
\begin{equation*}
     M_{ \sqrt{\mathcal{S}} \delta\left( \nu_\mc{S},\lambda \right) } \le 2,
\end{equation*}
 similarly to the random matrix ensembles (Fact \ref{fact:med_exp}). In Figure \ref{fig:conj} we can see that value of A is not only smaller but also that it converges to $0$ with growing $\epsilon$ faster than B and C. To put a number on this observation we fitted function of the form:
\begin{equation}
    G_L(\epsilon) = \frac{1}{2}\left( 1 - \mathrm{erf}\left(\frac{\epsilon}{\sqrt{2}L} \right) \right),
\end{equation}
to A and B shifted such that they represent concentration around the median instead of $2$:
\begin{enumerate}
    \item[A'] $\mathbb{P}\left( \sqrt{\mathcal{S}} \delta\left( \nu_\mc{S},\lambda \right) > M_{ \sqrt{\mathcal{S}} \delta\left( \nu_\mc{S},\lambda \right) } + \epsilon \right)$,
    \item[B'] $\mathbb{P}\left( \delta\left( \lambda \right) > M_{ \delta\left( \lambda \right) } + \epsilon \right)$,
\end{enumerate}
where $L$ was treated as a fitting parameter. From Facts \ref{fact:gauss_concentration} and \ref{fact:tail_bounds} we know that for an appropriate value of $L$ denoted - $L_*$ - there holds:
\begin{equation}
    \mathrm{B'} \le G_{L_*}(\epsilon) \le \mathrm{C}.
\end{equation}
Moreover, we found out that this function can be very well fitted to our numerical data. The results of this computation (see Figure \ref{fig:Ls}) show that indeed A converges to $0$ faster than C and that this difference grows with $d_\lambda$.
}

% The other argument for our conjecture comes from the observation that the median (see Figure \ref{fig:medians})
% \begin{equation*}
%      M_{ \sqrt{\mathcal{S}} \delta\left( \nu_\mc{S},\lambda \right) } \le 2,
% \end{equation*}
%  similarly to the random matrix ensembles (Fact \ref{fact:med_exp}). Since median is a value at which tail probability is equal to $\frac{1}{2}$ we obtain that for $\epsilon = 0$ it holds A $\le \frac{1}{2}$ while from the Theorem \ref{thm:bounds_T_t} we know that for the same $\epsilon$ we have C$= \frac{1}{2}$. So at least in $\epsilon=0$ it holds A $\le$ C and unless C is decreasing with $\epsilon$ much faster than A, which we observe not to be the case, the inequality holds for all $\epsilon > 0$.

In Figures \ref{fig:conj_t} and \ref{fig:conj_t_d2} we show similar data as in \ref{fig:conj} but this time for various $t$. In those Figures we compare:
\begin{itemize}
    \item[A] $\mathbb{P}\left(  \delta\left( \nu_\mc{S},t \right) > \frac{2 + \epsilon}{\sqrt{\mc{S}}} \right)$,
    \item[B] $\mathbb{P}\left(  \delta\left( t \right) > 2 + \epsilon \right)$,
    \item[C] bounds \ref{eq:bounds_T_t} and \ref{eq:bounds_T_t_2} from Theorem \ref{thm:bounds_T_t}
\end{itemize}
for various values of $t$ and $\mathcal{S}$ and again we observe that A $\le$ B $\le$ C. We also note that the bound $\ref{eq:bounds_T_t_2}$ is much better for small $\epsilon$ than the bound $\ref{eq:bounds_T_t}$.

\subsection{Rescaling}

In Section $\ref{sec:spectral_properties}$ we showed that when $t\rightarrow \infty$ the distribution of singular values of $T_{\nu_\mc{S},t}$ converges to \eqref{eq:kesten-final}:
\begin{equation*}
    \frac{|\S| \sqrt{\delta_{\mathrm{opt}}^2(\S) - x^2}}{\pi (1-x^2)}, \quad x \in \left[0, \delta_{\mathrm{opt}}(\mathcal{S})\right].
\end{equation*}
where $\delta_\mathrm{opt}(\mc{S})=\frac{2\sqrt{\mathcal{S}-1}}{\mathcal{S}}$. We note that, by standard dimensionality arguments, for (symmetric) Haar random gate-set $\mc{S}$ the group generated by $\mc{S}\cup \S^{-1}$ is free with the probability $1$. Then in Section \ref{sec:asymp_beh} we obtained that in the same limit the distribution of singular values of $\sqrt{\mathcal{S}}T_{\nu_\mc{S},t}$ converges to:
\begin{equation}\label{eq:ST_sv_dist}
   \frac{1}{\pi(1-\frac{y^2}{\mc{S}})}\sqrt{\frac{4\left(\mc{S}-1\right)}{\mc{S}}-y^2}\,\,, \quad y \in \left[ 0, 2\sqrt{\frac{\mathcal{S} - 1}{\mathcal{S}}} \right],
\end{equation}
which in a limit $\mathcal{S} \rightarrow \infty$ converges to the quarter-circle distribution:
\begin{gather*}
   \sigma_{QC}(y)=\frac{1}{\pi}\sqrt{4-y^2}, \quad y \in [0, 2].
\end{gather*}
From Figure \ref{fig:averages} we know that for $d_\lambda \rightarrow \infty$ expected values $\mathbb{E}\left( \sqrt{\mathcal{S}} \| T_{\nu_\mathcal{S}, \lambda} \| \right)$ and $\mathbb{E}\left( \| T_{\lambda} \| \right)$ concentrate around the right endpoints of supports of \eqref{eq:ST_sv_dist} and $\sigma_{QC}$ respectively. The fact that those endpoints are different for \eqref{eq:ST_sv_dist} and $\sigma_{QC}$ is causing the divergence of the distributions $\sqrt{\mathcal{S}} \delta\left( \nu_\mc{S},\lambda \right)$ and $\delta\left( \lambda \right)$ for large $d_\lambda$ what we observed in Figures \ref{fig:averages} and \ref{fig:renorm} and discussed in previous subsection.\\
If we want our random matrix model to match $T_{\nu_\mathcal{S}, \lambda}$ also in a regime of big $d_\lambda$ we have to multiply our operator by $\frac{2}{\delta_\mathrm{opt}(\mathcal{S})}$ instead of $\sqrt{\mathcal{S}}$. We can easily compute that the distribution of singular values of $\frac{2}{\delta_\mathrm{opt}(\mathcal{S})} T_{\nu_\mathcal{S}, t}$ is:
\begin{gather}
    \frac{\mc{S} - 1}{\mc{S} (1 - \frac{\mc{S} - 1}{\mc{S}^2}y^2)}\frac{1}{\pi}\sqrt{4 - y^2}, \quad y \in [0, 2],
\end{gather}
so we kept the asymptotic behaviour for $\mathcal{S} \rightarrow \infty$ and the endpoints are now the same. In order to show that $\frac{2}{\delta_\mathrm{opt}(\mathcal{S})} T_{\nu_\mathcal{S}, t}$ converges to the same random matrix ensemble model as $\sqrt{\mathcal{S}} T_{\nu_\mathcal{S}, t}$ we can either repeat the steps from Section \ref{sec:asymp_beh} or note that $\frac{2}{\delta_\mathrm{opt}(\mathcal{S})} = \frac{\mathcal{S}}{\sqrt{\mathcal{S} - 1}} = \Theta(\sqrt{\mathcal{S}})$ so asymptotically those two operators are the same.\\
Figure \ref{fig:renorm} shows how $\sqrt{\mathcal{S}}$ and $\frac{2}{\delta_\mathrm{opt}(\mathcal{S})}$ scalings change with $d_\lambda$. Clearly, the second one is much more similar to the random matrix ensemble model. On the other hand we can also see that the distribution of $\frac{2}{\delta_\mathrm{opt}(\mathcal{S})}  \delta\left( \nu_\mc{S},\lambda \right)$ is shifted relative to the distribution of $\delta\left( \lambda \right)$ towards bigger values. In other words, our data indicates that:
\begin{equation*}
    \mathbb{P}\left( \sqrt{\mc{S}} \delta\left( \nu_\mc{S},\lambda \right) > 2 + \epsilon \right) \le \mathbb{P}\left( \delta\left( \lambda \right) > 2 + \epsilon \right) \le \mathbb{P}\left( \frac{2}{\delta_\mathrm{opt}(\mathcal{S})} \delta\left( \nu_\mc{S},\lambda \right) > 2 + \epsilon \right).
\end{equation*}
In summary, the $\frac{2}{\delta_\mathrm{opt}(\mathcal{S})}$ scaling has better asymptotic behaviour for $d_\lambda \rightarrow \infty$ but it is less useful for theoretical purposes where we need an upper bound on a tail probability.

\begin{figure}[hp]
    \centering
    \includegraphics{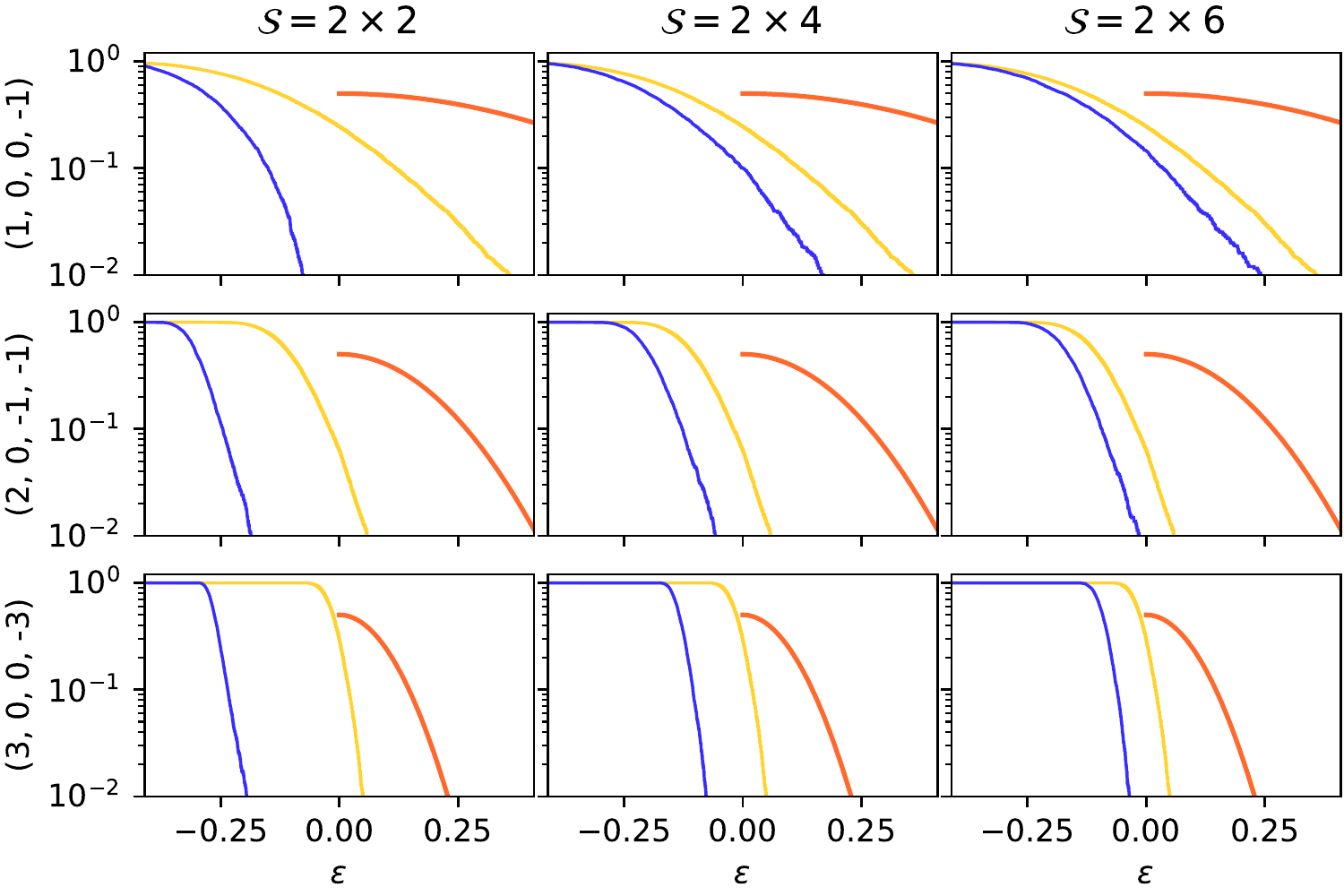}
    \caption{Comparison of $ \mathbb{P}\left(\sqrt{\mc{S}}\delta\left( \nu_\mc{S},\lambda \right)>2+\epsilon\right)$ - \textbf{blue}, $\mathbb{P}\left(\delta\left( \lambda \right)>2+\epsilon\right )$ - \textbf{yellow} and tail bounds $F_\ast(d_\lambda,\epsilon)$ from Fact \ref{fact:tail_bounds} - \textbf{orange} for Haar random gate-sets of dimension $d=4$, sizes: $\mathcal{S}=2\times2$ - \textbf{left}, $\mathcal{S}=2\times4$ - \textbf{middle}, $\mathcal{S}=2\times6$ - \textbf{right} and for weights: $\lambda = (1, 0, 0, -1)$ {\color{black}of dimension $d_\lambda=15$} - \textbf{top}, $\lambda=(2, 0, -1, -1)$ {\color{black}of dimension $d_\lambda=45$} - \textbf{middle}, $\lambda=(3, 1, -1, -3)$ {\color{black}of dimension $d_\lambda=729$} - \textbf{bottom}.}
    \label{fig:conj}
\end{figure}

\begin{figure}[hp]
    \centering
    \includegraphics{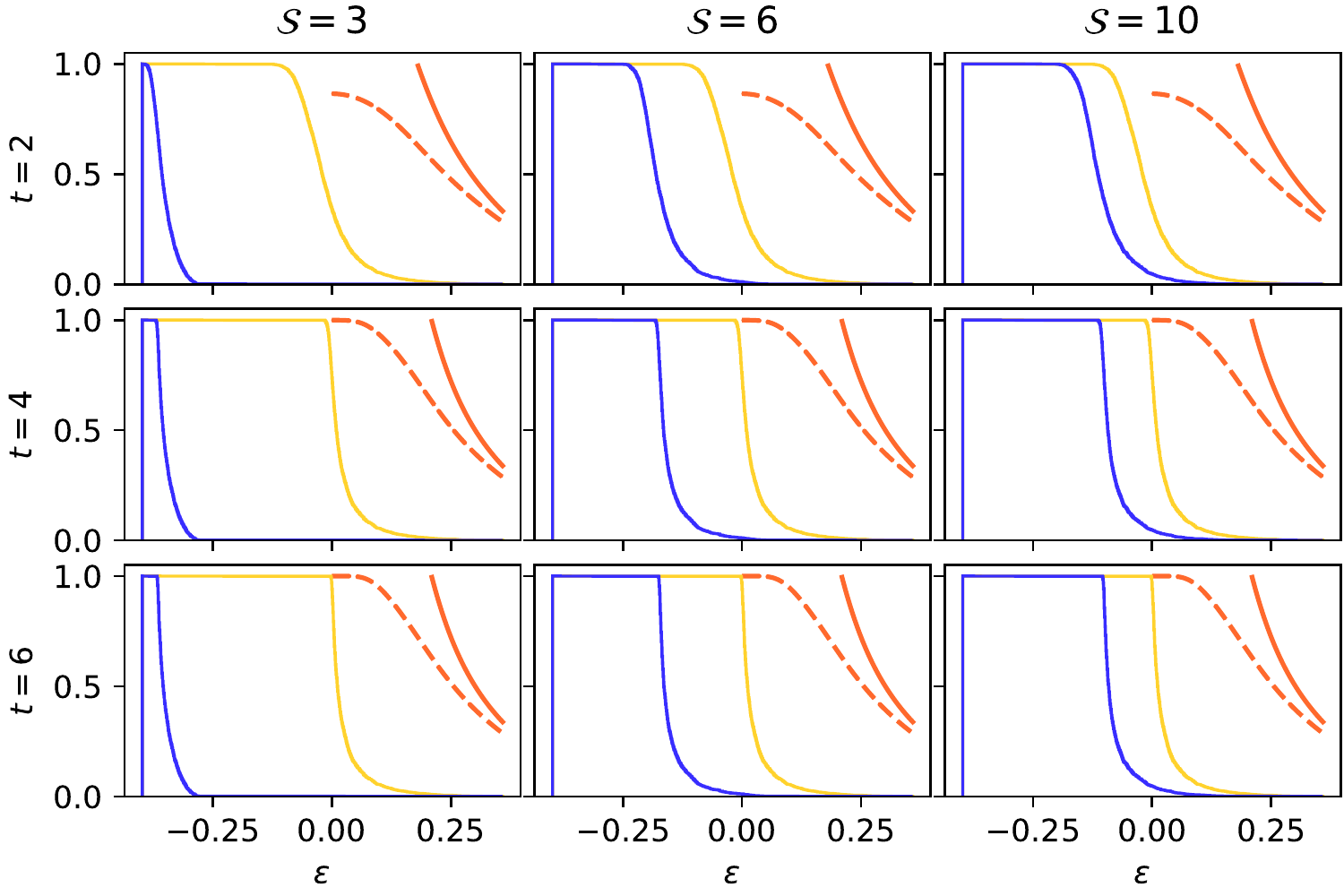}
    \caption{Comparison of $ \mathbb{P}\left(\sqrt{\mc{S}}\delta\left( \nu_\mc{S},t \right) >2+\epsilon\right)$ - \textbf{blue}, $\mathbb{P}\left(\delta\left( t \right)>2+\epsilon\right )$ - \textbf{yellow}, bound \eqref{eq:bounds_T_t} - \textbf{solid orange} and bound \eqref{eq:bounds_T_t_2} - \textbf{dashed orange} for symmetric Haar random gate-sets of dimension $d=4$, sizes: $\mathcal{S}=3$ - left, $\mathcal{S}=6$ - \textbf{middle}, $\mathcal{S}=10$ - \textbf{right} and for values of $t$: $t=2$ - \textbf{top}, $t=4$ - \textbf{middle},  $t=6$ - \textbf{bottom}.}
    \label{fig:conj_t}
\end{figure}

\begin{figure}[hp]
    \centering
    \includegraphics{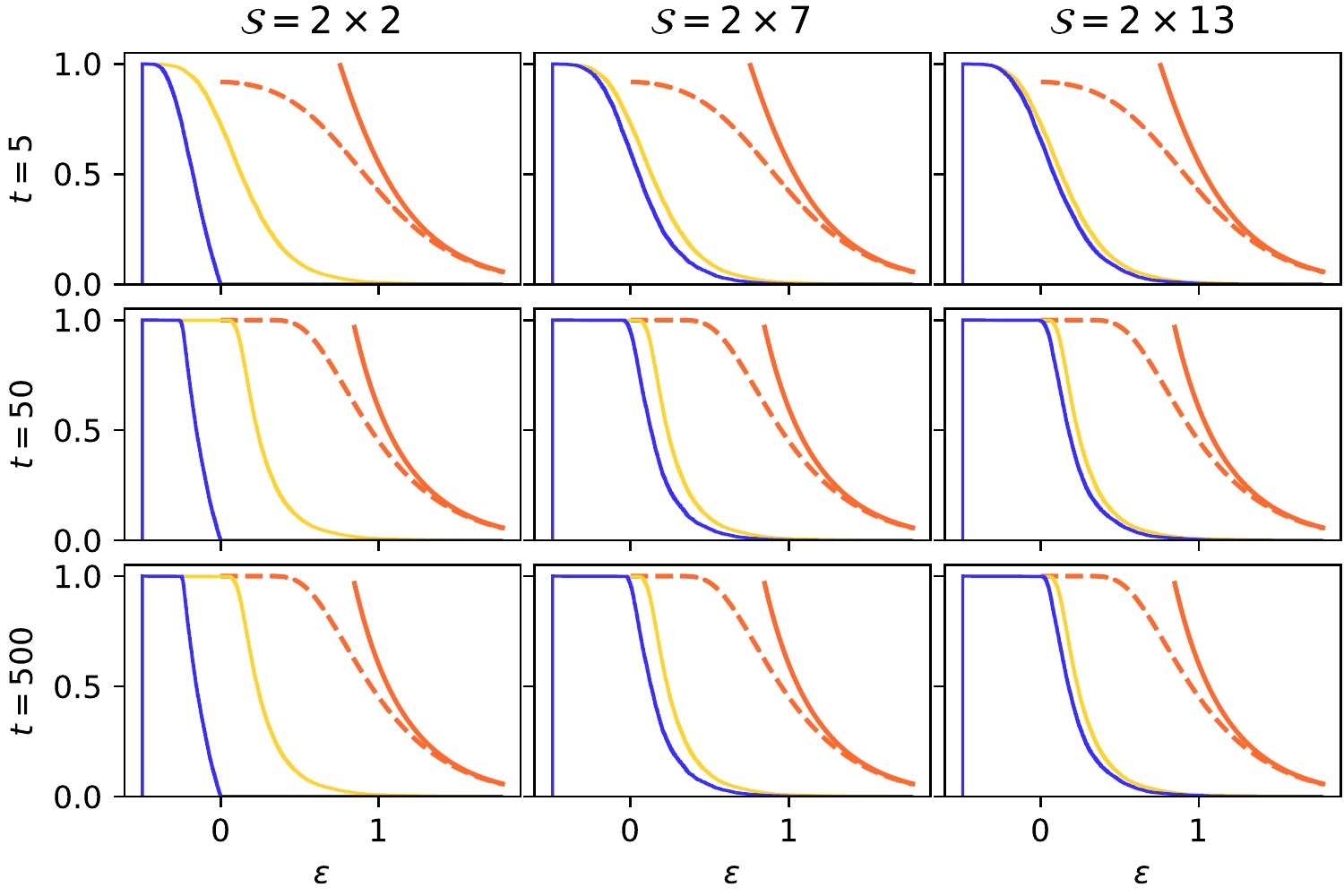}
    \caption{Comparison of $ \mathbb{P}\left(\sqrt{\mc{S}}\delta\left( \nu_\mc{S},t \right) >2+\epsilon\right)$ - \textbf{blue}, $\mathbb{P}\left(\delta\left( t \right)>2+\epsilon\right )$ - \textbf{yellow}, bound \eqref{eq:bounds_T_t} - \textbf{solid orange} and bound \eqref{eq:bounds_T_t_2} - \textbf{dashed orange} for symmetric Haar random gate-sets of dimension $d=2$, sizes: $\mathcal{S}=2\times2$ - \textbf{left}, $\mathcal{S}=2\times7$ - \textbf{middle}, $\mathcal{S}=2\times13$ - \textbf{right} and for values of $t$: $t=5$ - \textbf{top}, $t=50$ - \textbf{middle},  $t=500$ - \textbf{bottom}.}
    \label{fig:conj_t_d2}
\end{figure}

\begin{figure}[hp]
    \centering
    \includegraphics{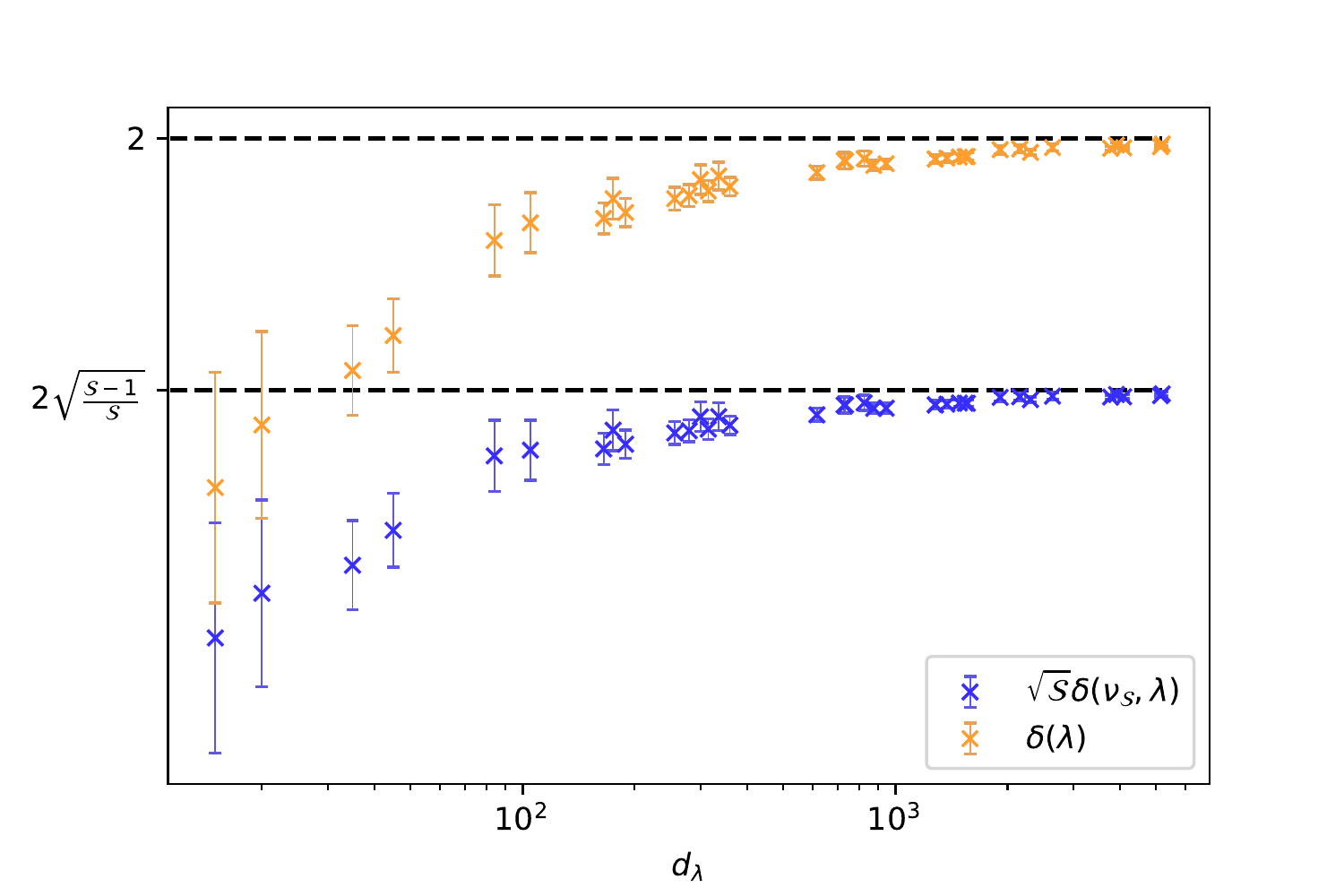}
    \includegraphics{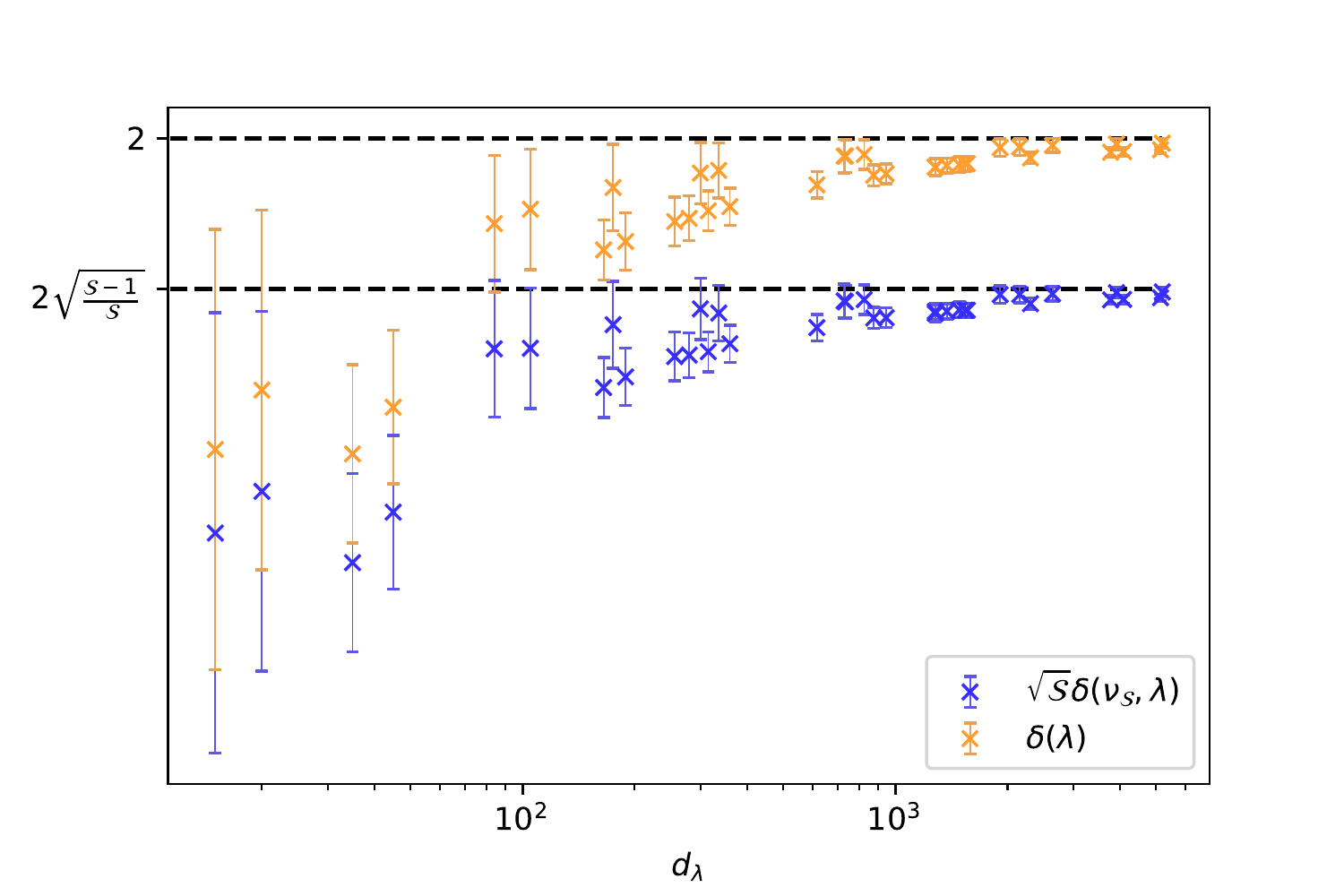}
    \caption{Averages and standard deviations of the distributions of $\sqrt{\mc{S}}\delta\left( \nu_\mc{S}, \lambda \right)$ - \textbf{blue} and $\delta\left( \lambda \right)$ - \textbf{orange} for $\lambda \in \Lambda^\mathrm{ess}_6$, $d = 4$, Haar random sets with $\mathcal{S}=9$ - \textbf{top} and symmetric Haar random sets with $\mathcal{S}=2\times9$ - \textbf{bottom}.}
    \label{fig:averages}
\end{figure}

\begin{figure}[hp]
    \centering
    \includegraphics{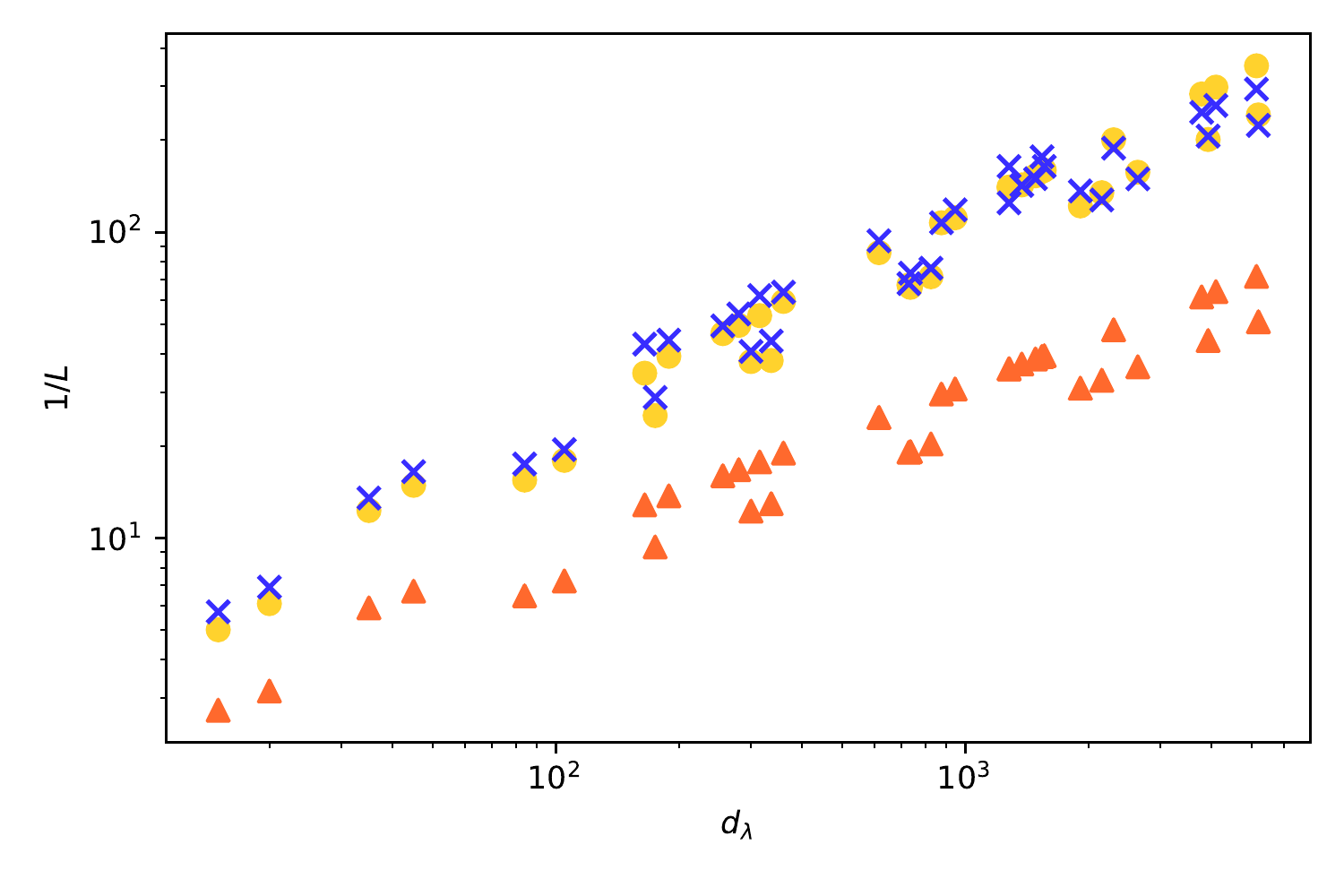}
    \caption{The speed of convergence to $0$ (defined in Section \ref{sec:results}) of tail distributions of $\sqrt{\mc{S}}\delta\left( \nu_\mc{S}, \lambda \right)$ - \textbf{blue crosses}, $\delta\left( \lambda \right)$ - \textbf{yellow dots} and tail bounds from Fact \ref{fact:tail_bounds} - \textbf{orange triangles} for $\lambda \in \Lambda_6^\mathrm{ess}$, $d=4$ and symmetric Haar random sets with $\mathcal{S}=2 \times 7$.}
    \label{fig:Ls}
\end{figure}

\begin{figure}[hp]
    \centering
    \includegraphics{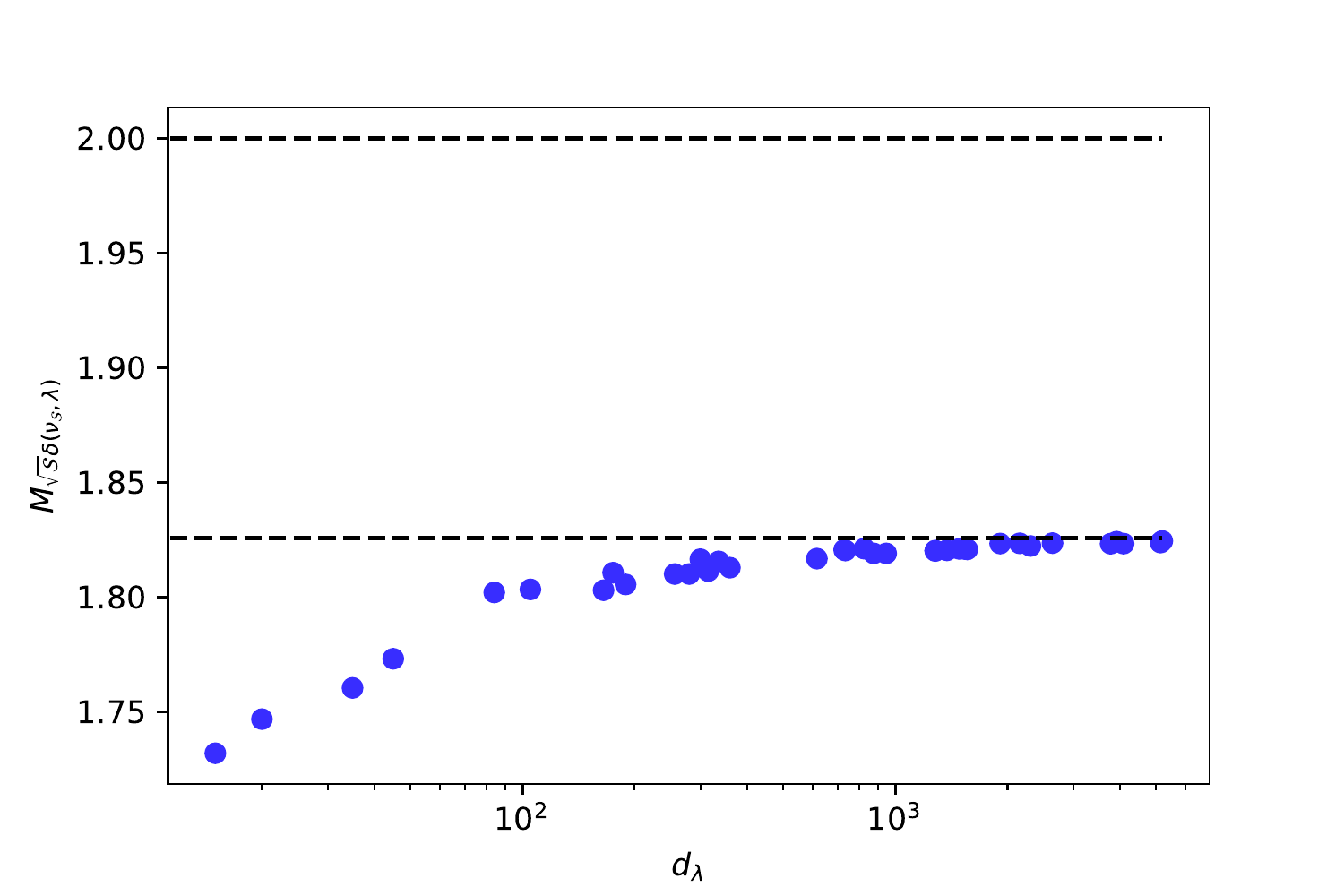}
    \includegraphics{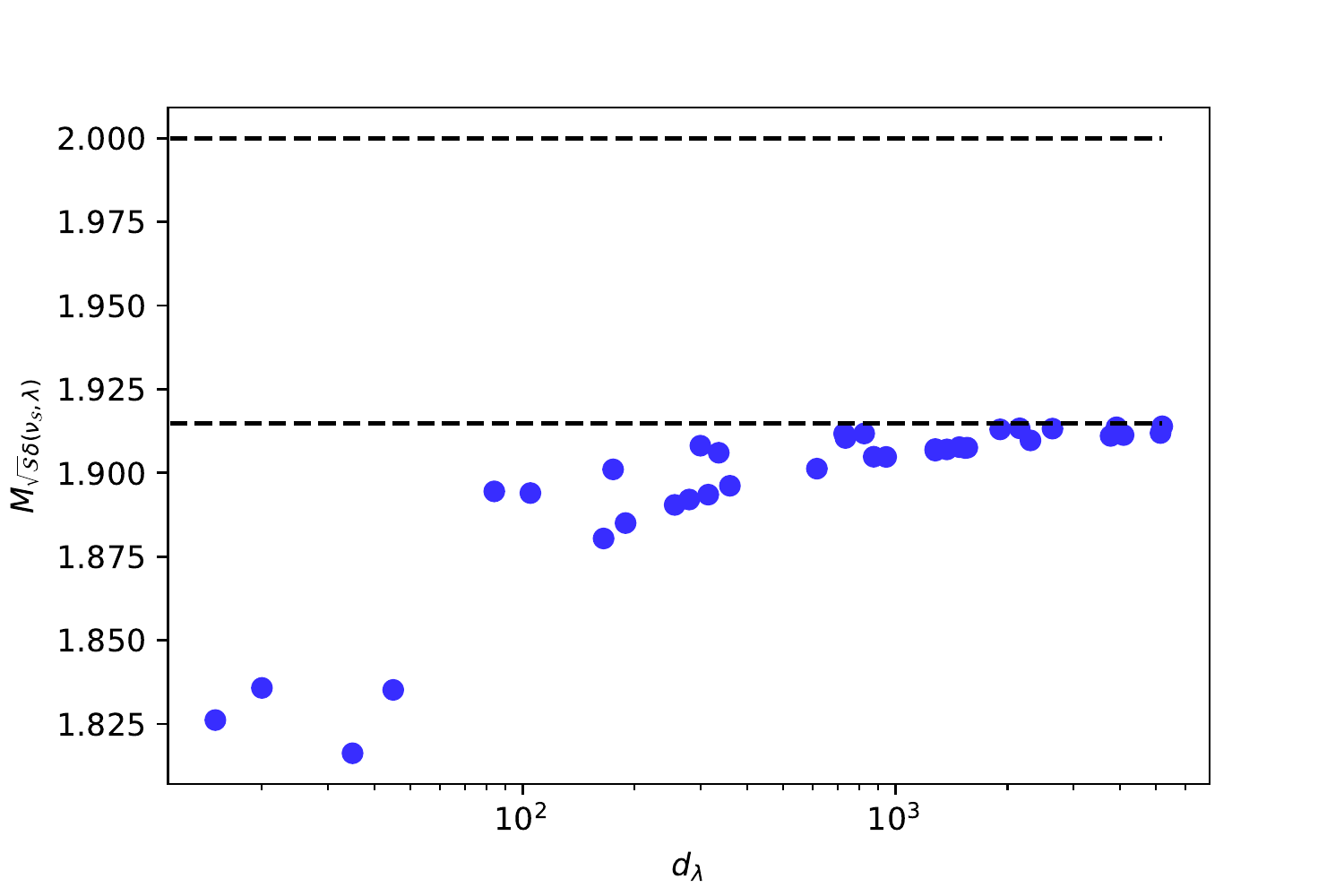}
    \caption{Medians $M_{\sqrt{\mc{S}}\delta\left( \nu_\mc{S},\lambda \right)}$ for $\lambda \in \Lambda^\mathrm{ess}_6$, $d=4$, Haar random sets with $\mathcal{S}=6$ - \textbf{top} and symmetric Haar random sets with $\mathcal{S}=2\times6$ - \textbf{bottom}. With dashed lines we denoted $2 \sqrt{\frac{\mathcal{S} - 1}{\mathcal{S}}}$ and $2$.}
    \label{fig:medians}
\end{figure}

\begin{figure}[hp]
    \centering
    \includegraphics{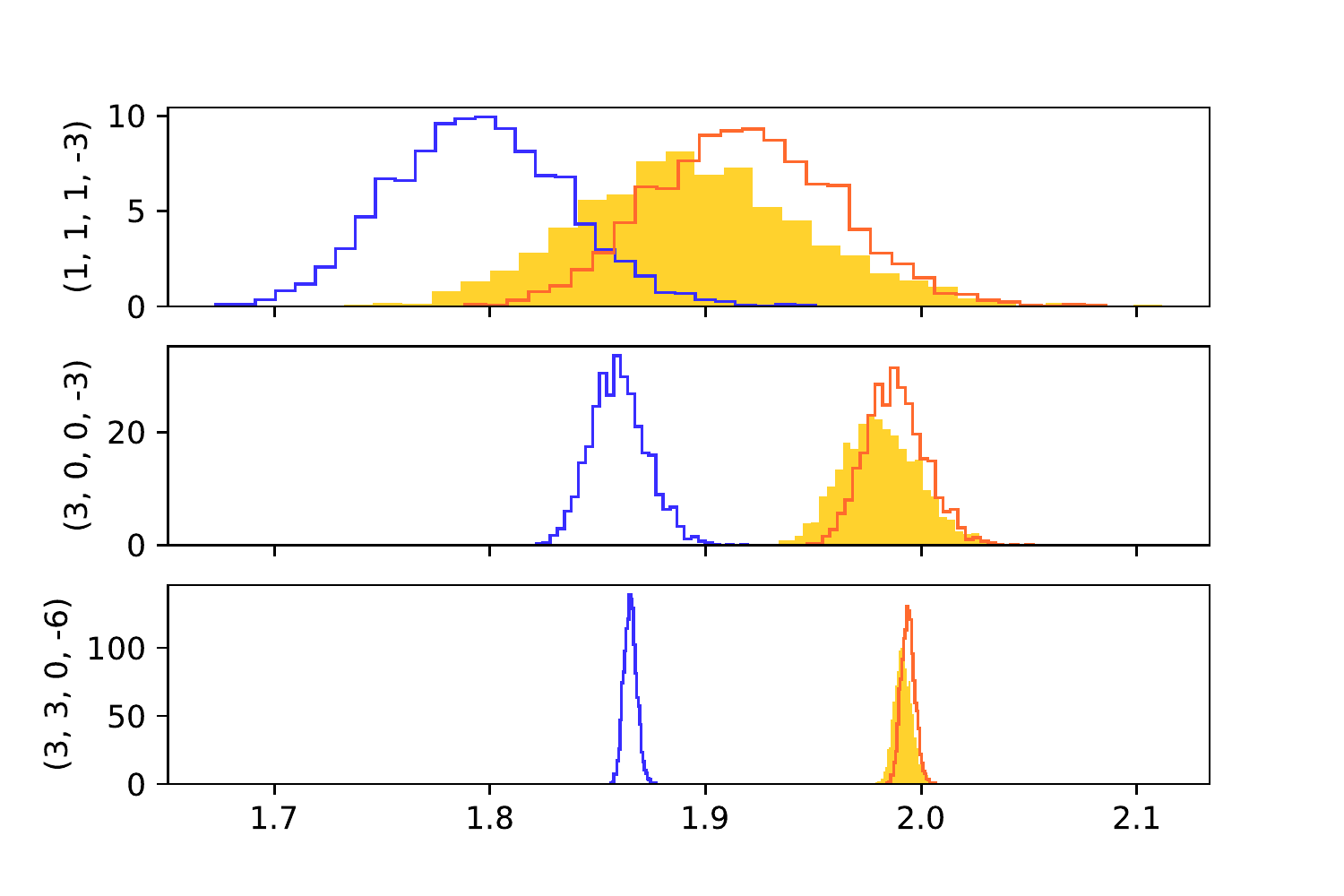}
    \includegraphics{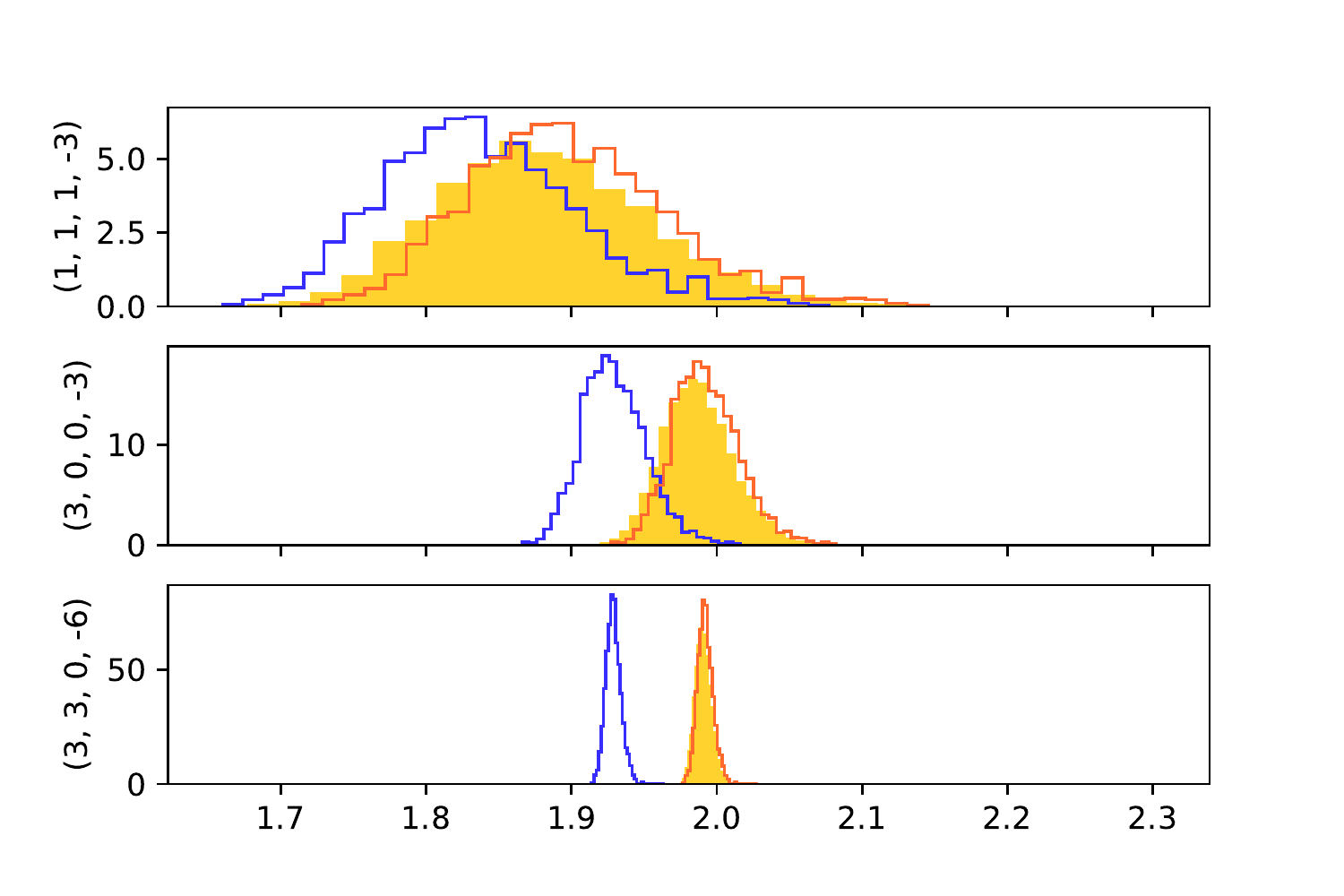}
    \caption{Histograms of $\sqrt{\mc{S}}\delta\left( \nu_\mc{S}, \lambda \right)$ - \textbf{blue}, $\frac{2}{\delta_{opt}(\mathcal{S})} \delta\left( \nu_\mc{S}, \lambda \right)$ - \textbf{orange} and $\delta\left( \lambda \right)$ - \textbf{yellow} for dimension $d=4$, Haar random gate-set with size $\mathcal{S}=8$ - \textbf{top}, symmetric Haar random gate-set with size $\mathcal{S}=2\times8$ - \textbf{bottom} and for  weights: $\lambda = (1, 1, 1, -3)$ {\color{black}of dimension $d_\lambda=35$}, $\lambda=(3, 0, 0, -3)$ {\color{black}of dimension $d_\lambda=300$}, $\lambda=(3, 3, 0, -6)$ {\color{black}of dimension $d_\lambda=1540$}.}
    \label{fig:renorm}
\end{figure}

\section*{Acknowledgments}
This research was funded by the National Science Centre, Poland under the grant OPUS: UMO-2020/37/B/ST2/02478 and supported in part by PLGrid Infrastructure. AS would like to thank Yossi Avron for inspiring discussions.

\printbibliography

@book{tropp2015,
    author = "J. A. Tropp",
    title = "An Introduction to Matrix Concentration Inequalities",
    year = "2015",
    publisher = "Now Publishers Inc",
    series = "Foundations and Trends in Machine Learning",
    isbn = "978-1601988386",
    url = "https://arxiv.org/abs/1501.01571",
    archivePrefix = {arXiv},
    eprint = {1501.01571}
}

@article{Bourgain2011,
  title={A spectral gap theorem in SU(d)},
  author={J. Bourgain and A. Gamburd},
  journal={J. Eur. Math. Soc.},
  volume = {14},
  number = {5},
  year={2012},
  pages={1455–1511},
  DOI={10.4171/JEMS/337},
  archivePrefix = {arXiv},
    eprint = {1108.6264},
}

@article{lubotzky,
   title={Hecke operators and distributing points on S2. II},
   journal={Communications on Pure and Applied Mathematics},
   volume={40},
   author={A. Lubotzky and R. Phillips and P. Sarnak},
   year={1987},
   pages={401-420},
   publisher={Wiley Subscription Services, Inc., A Wiley Company}
}

@book{Dieck,
    title     = {Representations of Compact Lie Groups},
    author    = {Br{\"o}cker, T. and Dieck, T.},
    isbn      = {9783540136781},
    lccn      = {95227046},
    series    = {Graduate Texts in Mathematics},
    url       = {https://books.google.pl/books?id=AfBzWL5bIIQC},
    year      = {2003},
    publisher = {Springer Berlin Heidelberg}
}

@article{Stroomer,
    author = "G. Benkart and M. Chakrabarti and T. Halverson and R. Leduc and C. Lee and J. Stroomer",
    title = "Tensor product representations of general linear groups and their connections with Brauer algebras",
    year = "1994",
    journal = "J. Algebra",
    volume = "166",
    pages = "529–567"
}

@misc{https://doi.org/10.48550/arxiv.2201.11774,
  doi = {10.48550/ARXIV.2201.11774},
  
  url = {https://arxiv.org/abs/2201.11774},
  
  author = {Słowik, O. and Sawicki, A.},
  
  title = {Calculable lower bounds on the efficiency of universal sets of quantum gates},
  
  publisher = {arXiv},
  
  year = {2022},

    journal = {arXiv e-prints},
    archivePrefix = {arXiv},
    eprint = {2201.11774},
  
  copyright = {Creative Commons Attribution 4.0 International}
}

@book{feller1957introduction,
  title={An Introduction to Probability Theory and Its Applications},
  author={Feller, W. and Feller, V.},
  number={t. 1-2},
  isbn={9780471257097},
  lccn={57010805},
  % series={An Introduction to Probability Theory and Its Applications},
  url={https://books.google.pl/books?id=BsSwAAAAIAAJ},
  year={1957},
  publisher={Wiley}
}

@article{kesten,
 ISSN = {00029947},
 URL = {http://www.jstor.org/stable/1993160},
 author = {H. Kesten},
 journal = {Transactions of the American Mathematical Society},
 number = {2},
 pages = {336--354},
 publisher = {American Mathematical Society},
 title = {Symmetric Random Walks on Groups},
 % urldate = {2022-09-23},
 volume = {92},
 year = {1959}
}

@article{concentration,
    author = "P. Dulian and A. Sawicki",
    title = "Matrix concentration inequalities and efficiency of
    random universal sets of quantum gates",
    year = "2022",
    journal = {arXiv e-prints},
    archivePrefix = {arXiv},
    eprint = {2202.05371},
    % TODO
    % jornal = "",
    % volume = "",
    % pages = ""
}

@article{karol1,
	author = {{\.Z}yczkowski,K. and Penson,K. A. and Nechita,I. and Collins,B.},
	journal = {Journal of Mathematical Physics},
	number = {6},
	pages = {062201},
	title = {Generating random density matrices},
	volume = {52},
	year = {2011},
        archivePrefix = {arXiv},
        eprint = {1010.3570}
 }

@book{barut,
    author = "A. O. Barut and R. Rączka",
    title = "Theory of group representations and applications",
    year = "1986",
    publisher = "World Scientific Publishing Co Pte Ltd.",
    isbn = "978-9971502164"
}

@article{mezzadri,
    author = "F. Mezzadri",
    title = "How to generate random matrices from the classical compact groups",
    year = "2007",
    month = {10},
    journal = {Notices of the American Mathematical Society},
    volume = "54",
    pages = "592-604",
    archivePrefix = {arXiv},
    eprint = {math-ph/0609050}
}

@book{fulton1991representation,
  title={Representation Theory: A First Course},
  author={Fulton, W. and Harris, J.},
  isbn={9780387974958},
  lccn={96165313},
  series={Graduate Texts in Mathematics},
  url={https://books.google.pl/books?id=TuQZAQAAIAAJ},
  year={1991},
  publisher={Springer New York}
}

@article{EffLearning2020,
    author = {H. Huang and R. Kueng and J. Preskill},
    title = {Predicting Many Properties of a Quantum System from Very Few Measurements},
    
    doi = {10.1038/s41567-020-0932-7},
	url = {https://doi.org/10.1038\%2Fs41567-020-0932-7},
	year = {2020},
	month = {6},
	publisher = {Springer Science and Business Media {LLC}},
	volume = {16},
	number = {10},
	pages = {1050--1057},
	journal = {Nature Physics},
    
    % keywords = {Quantum Physics, Computer Science - Information Theory, Computer Science - Machine Learning},
    % eid = {arXiv:2002.08953},
    % pages = {arXiv:2002.08953},
    archivePrefix = {arXiv},
    eprint = {2002.08953},
    % primaryClass = {quant-ph},
    % adsurl = {https://ui.adsabs.harvard.edu/abs/2020arXiv200208953H},
    % adsnote = {Provided by the SAO/NASA Astrophysics Data System}
}

@Article{BHH2016,
  author        = {{Brand{\~a}o}, F. G.~S.~L. and {Harrow}, A. W. and {Horodecki}, M.},
  title         = {{Local Random Quantum Circuits are Approximate Polynomial-Designs}},
  journal       = {Communications in Mathematical Physics},
  year          = {2016},
  volume        = {346},
  number        = {2},
  pages         = {397-434},
  month         = {9},
  adsurl        = {https://ui.adsabs.harvard.edu/abs/2016CMaPh.346..397B},
  archiveprefix = {arXiv},
  doi           = {10.1007/s00220-016-2706-8},
  eprint        = {1208.0692},
  primaryclass  = {quant-ph},
}

@article{Gambetta2014,
    author = {J. M. Epstein and A. W. Cross and E. Magesan and J. M. Gambetta},
    
    title = {Investigating the limits of randomized benchmarking protocols},
    
    doi = {10.1103/physreva.89.062321},
	url = {https://doi.org/10.1103\%2Fphysreva.89.062321},
	year = {2014},
	month = {6},
	publisher = {American Physical Society ({APS})},
	volume = {89},
	number = {6},
 	journal = {Physical Review A},

    % keywords = {03.67.Ac, 03.67.Pp, 03.65.Wj, 03.67.Lx, Quantum algorithms protocols and simulations, Quantum error correction and other methods for protection against decoherence, State reconstruction quantum tomography, Quantum computation, Quantum Physics},
    
    % year = "2014",
    % month = "Jun",
    % volume = {89},
    % number = {6},
    % eid = {062321},
    % pages = {062321},
    % doi = {10.1103/PhysRevA.89.062321},
    archivePrefix = {arXiv},
    eprint = {1308.2928},
    % primaryClass = {quant-ph},
    % adsurl = {https://ui.adsabs.harvard.edu/abs/2014PhRvA..89f2321E},
    % adsnote = {Provided by the SAO/NASA Astrophysics Data System}
}

@ARTICLE{Decoupl2013,
       author = {O. Szehr and F. Dupuis and M. Tomamichel and R. Renner},
        title = "{Decoupling with unitary approximate two-designs}",
      journal = {New Journal of Physics},
         year = "2013",
        month = "5",
       volume = {15},
       number = {5},
          eid = {053022},
        pages = {053022},
          doi = {10.1088/1367-2630/15/5/053022},
archivePrefix = {arXiv},
       eprint = {1109.4348},
 primaryClass = {quant-ph},
       adsurl = {https://ui.adsabs.harvard.edu/abs/2013NJPh...15e3022S}
}

@ARTICLE{StateDiscrimination2005,
       author = {Radhakrishnan, J. and Rötteler, M. and Sen, P.},
        title = "{Random measurement bases, quantum state distinction and applications to the hidden subgroup problem}",
      journal = {Algorithmica},
      volume = {55},
     keywords = {Quantum Physics},
         year = "2009",
        month = "10",
          % eid = {quant-ph/0512085},
        pages = {490–516},
archivePrefix = {arXiv},
       % eprint = {quant-ph/0512085},
 primaryClass = {quant-ph},
       adsurl = {https://ui.adsabs.harvard.edu/abs/2005quant.ph.12085S},
      adsnote = {Provided by the SAO/NASA Astrophysics Data System}
}

@ARTICLE{ChaosDesign2017,
       author = {{Roberts}, Daniel A. and {Yoshida}, Beni},
        title = "{Chaos and complexity by design}",
      journal = {Journal of High Energy Physics},
         year = "2017",
        month = "Apr",
       volume = {2017},
       number = {4},
          eid = {121},
        pages = {121},
          doi = {10.1007/JHEP04(2017)121},
archivePrefix = {arXiv},
       eprint = {1610.04903},
 primaryClass = {quant-ph},
       adsurl = {https://ui.adsabs.harvard.edu/abs/2017JHEP...04..121R}
}

@article{HArrowHasstings2008,
author = {Hastings, M. B. and Harrow, A. W.},
title = {Classical and Quantum Tensor Product Expanders},
year = {2009},
issue_date = {March 2009},
publisher = {Rinton Press, Incorporated},
address = {Paramus, NJ},
volume = {9},
number = {3},
issn = {1533-7146},
journal = {Quantum Info. Comput.},
month = mar,
pages = {336–360},
numpages = {25},
archivePrefix = {arXiv},
    eprint = {0804.0011}
}

@ARTICLE{Harrow2018,
       author = {{Harrow}, A. and {Mehraban}, S.},
        title = "{Approximate unitary $t$-designs by short random quantum circuits using nearest-neighbor and long-range gates}",
      journal = {arXiv e-prints},
     keywords = {Quantum Physics},
         year = "2018",
        month = "9",
          eid = {arXiv:1809.06957},
        % pages = {arXiv:1809.06957},
archivePrefix = {arXiv},
       eprint = {1809.06957},
 primaryClass = {quant-ph},
       adsurl = {https://ui.adsabs.harvard.edu/abs/2018arXiv180906957H},
      adsnote = {Provided by the SAO/NASA Astrophysics Data System}
}

@ARTICLE{Qhomeopathy2020,
       author = {{Haferkamp}, J. and {Montealegre-Mora}, F. and
         {Heinrich}, M. and {Eisert}, J. and {Gross}, D. and
         {Roth}, I.},
        title = "{Quantum homeopathy works: Efficient unitary designs with a system-size independent number of non-Clifford gates}",
      journal = {arXiv e-prints},
     keywords = {Quantum Physics, Mathematical Physics},
         year = 2020,
        month = feb,
          % eid = {arXiv:2002.09524},
        % pages = {arXiv:2002.09524},
archivePrefix = {arXiv},
       eprint = {2002.09524},
 primaryClass = {quant-ph},
       adsurl = {https://ui.adsabs.harvard.edu/abs/2020arXiv200209524H},
      adsnote = {Provided by the SAO/NASA Astrophysics Data System}
}

@ARTICLE{InfTransmission2009,
       author = {{Abeyesinghe}, A. and {Devetak}, I. and {Hayden}, P. and {Winter}, A.},
        title = "{The mother of all protocols: restructuring quantum information's family tree}",
      journal = {Proceedings of the Royal Society of London Series A},
         year = "2009",
        month = "6",
       volume = {465},
       number = {2108},
        pages = {2537-2563},
          doi = {10.1098/rspa.2009.0202},
archivePrefix = {arXiv},
       eprint = {quant-ph/0606225},
 primaryClass = {quant-ph},
       adsurl = {https://ui.adsabs.harvard.edu/abs/2009RSPSA.465.2537A}
}

@article{Yoshifumi1,
  title = {Efficient Quantum Pseudorandomness with Nearly Time-Independent Hamiltonian Dynamics},
  author = {Nakata, Y. and Hirche, C. and Koashi, M. and Winter, A.},
  journal = {Phys. Rev. X},
  volume = {7},
  issue = {2},
  pages = {021006},
  numpages = {20},
  year = {2017},
  month = {Apr},
  publisher = {American Physical Society},
  doi = {10.1103/PhysRevX.7.021006},
  url = {https://link.aps.org/doi/10.1103/PhysRevX.7.021006},
  archivePrefix = {arXiv},
  eprint = {1609.07021}
}

@article{Yoshifumi2,
  title = {Quantum Circuits for Exact Unitary $t$-Designs and Applications to Higher-Order Randomized Benchmarking},
  author = {Nakata, Y. and Zhao, D. and Okuda, T. and Bannai, E. and Suzuki, Y. and Tamiya, S. and Heya, K. and Yan, Z. and Zuo, K. and Tamate, S. and Tabuchi, Y. and Nakamura, Y.},
  journal = {PRX Quantum},
  volume = {2},
  issue = {3},
  pages = {030339},
  numpages = {35},
  year = {2021},
  month = {9},
  publisher = {American Physical Society},
  doi = {10.1103/PRXQuantum.2.030339},
  url = {https://link.aps.org/doi/10.1103/PRXQuantum.2.030339},
  archivePrefix = {arXiv},
  eprint = {2102.12617}
}

@article{jonas2,
	Author = {Haferkamp, J. and Faist, P. and Kothakonda, N. B. T. and Eisert, J. and Yunger Halpern, N.},
	Journal = {Nature Physics},
	Number = {5},
	Pages = {528--532},
	Title = {Linear growth of quantum circuit complexity},
	Volume = {18},
	Year = {2022},
    archivePrefix = {arXiv},
    eprint = {2106.05305}
 }

@article{ORUC2016355,
title = {On number of partitions of an integer into a fixed number of positive integers},
journal = {Journal of Number Theory},
volume = {159},
pages = {355-369},
year = {2016},
issn = {0022-314X},
doi = {https://doi.org/10.1016/j.jnt.2015.06.023},
url = {https://www.sciencedirect.com/science/article/pii/S0022314X1500236X},
author = {A. Yavuz Oruc}
}

@book{Suskind2018,
        author = {{Susskind}, L.},
        title = "{Three Lectures on Complexity and Black Holes}",
      % journal = {arXiv e-prints},
        publisher = {Springer Cham},
        isbn = {978-3-030-45109-7},
     keywords = {High Energy Physics - Theory},
         year = "2020",
        % month = "Oct",
          % eid = {arXiv:1810.11563},
        % pages = {arXiv:1810.11563},
archivePrefix = {arXiv},
       eprint = {1810.11563},
 primaryClass = {hep-th},
       adsurl = {https://ui.adsabs.harvard.edu/abs/2018arXiv181011563S},
      adsnote = {Provided by the SAO/NASA Astrophysics Data System}
}

@article{MO1,
  doi = {10.48550/ARXIV.2205.09734},
  
  url = {https://arxiv.org/abs/2205.09734},
  
  author = {Oszmaniec, M. and Horodecki, M. and Hunter-Jones, N.},
  
  title = {Saturation and recurrence of quantum complexity in random quantum circuits},
  
  publisher = {arXiv},

  journal = {arXiv e-prints},
  archivePrefix = {arXiv},
  eprint = {2205.09734},
  
  year = {2022},
  
  copyright = {Creative Commons Attribution 4.0 International}
}

@ARTICLE{nets,
  author={Oszmaniec, M. and Sawicki, A. and Horodecki, M.},
  journal={IEEE Transactions on Information Theory}, 
  title={Epsilon-Nets, Unitary Designs, and Random Quantum Circuits}, 
  year={2022},
  volume={68},
  number={2},
  pages={989-1015},
  doi={10.1109/TIT.2021.3128110},
    archivePrefix = {arXiv},
    eprint = {2007.10885}
}

@article{Leone2021quantumchaosis,
  doi = {10.22331/q-2021-05-04-453},
  url = {https://doi.org/10.22331/q-2021-05-04-453},
  title = {Quantum {C}haos is {Q}uantum},
  author = {Leone, L. and Oliviero, S. F. E. and Zhou, Y. and Hamma, A.},
  journal = {{Quantum}},
  issn = {2521-327X},
  publisher = {{Verein zur F{\"{o}}rderung des Open Access Publizierens in den Quantenwissenschaften}},
  volume = {5},
  pages = {453},
  month = may,
  year = {2021},
  archivePrefix = {arXiv},
  eprint = {2102.08406}
}

@article{OLIVIERO2021127721,
title = {Transitions in entanglement complexity in random quantum circuits by measurements},
journal = {Physics Letters A},
volume = {418},
pages = {127721},
year = {2021},
issn = {0375-9601},
doi = {https://doi.org/10.1016/j.physleta.2021.127721},
url = {https://www.sciencedirect.com/science/article/pii/S0375960121005855},
author = {S. F. E. Oliviero and L. Leone and A. Hamma},
archivePrefix = {arXiv},
eprint = {2103.07481}
}

@article{Sawicki22,
  title = {Universality verification for a set of quantum gates},
  author = {Sawicki, A. and Mattioli, L. and Zimbor\'as, Z.},
  journal = {Phys. Rev. A},
  volume = {105},
  issue = {5},
  pages = {052602},
  numpages = {6},
  year = {2022},
  month = {5},
  publisher = {American Physical Society},
  doi = {10.1103/PhysRevA.105.052602},
  url = {https://link.aps.org/doi/10.1103/PhysRevA.105.052602},
  archivePrefix = {arXiv},
    eprint = {2111.03862},
}

@article{Haferkamp2022randomquantum,
  doi = {10.22331/q-2022-09-08-795},
  url = {https://doi.org/10.22331/q-2022-09-08-795},
  title = {Random quantum circuits are approximate unitary {$t$}-designs in depth {$O\left(nt^{5+o(1)}\right)$}},
  author = {Haferkamp, J.},
  journal = {{Quantum}},
  issn = {2521-327X},
  publisher = {{Verein zur F{\"{o}}rderung des Open Access Publizierens in den Quantenwissenschaften}},
  volume = {6},
  pages = {795},
  month = sep,
  year = {2022},
}

@article{PhysRevA.104.022417,
  title = {Improved spectral gaps for random quantum circuits: Large local dimensions and all-to-all interactions},
  author = {Haferkamp, J. and Hunter-Jones, N.},
  journal = {Phys. Rev. A},
  volume = {104},
  issue = {2},
  pages = {022417},
  numpages = {18},
  year = {2021},
  month = {8},
  publisher = {American Physical Society},
  doi = {10.1103/PhysRevA.104.022417},
  url = {https://link.aps.org/doi/10.1103/PhysRevA.104.022417},
    archivePrefix = {arXiv},
    eprint = {2012.05259},
}

@article{szarek2001Chapter8,
    title={Chapter 8 - Local Operator Theory, Random Matrices and Banach Spaces},
    author={K. R. Davidson and S. J. Szarek},
    year={2001},
    month = {3},
    volume = {1},
    isbn = {9780444828422},
    journal = {Handbook of the Geometry of Banach Spaces},
    doi = {10.1016/S1874-5849(01)80010-3}
}

@book{SzarekBook,
      author        = "Aubrun, G. and Szarek, S. J",
      title         = "{Alice and Bob meet Banach: the interface of asymptotic geometric analysis and quantum information theory}",
      publisher     = "American Mathematical Society",
      address       = "Providence, RI",
      series        = "Mathematical surveys and monographs",
      year          = "2017",
      url           = "https://cds.cern.ch/record/2296008",
      isbn          = "1470434687"
}

@book{HH09,
    author="Harrow, A. W. and Low, R. A.",
    editor="Dinur, Irit and Jansen, Klaus and Naor, Joseph and Rolim, Jos{\'e}",
    title="Efficient Quantum Tensor Product Expanders and k-Designs",
    booktitle="Approximation, Randomization, and Combinatorial Optimization. Algorithms and Techniques",
    year="2009",
    publisher="Springer Berlin Heidelberg",
    address="Berlin, Heidelberg",
    pages="548--561",
    isbn="978-3-642-03685-9",
    archivePrefix = {arXiv},
    eprint = {0811.2597},
}

@article{goldstein2016,
  doi = {10.48550/ARXIV.1603.03076},
  
  url = {https://arxiv.org/abs/1603.03076},
  
  author = {Goldstein, D. and Guralnick, R. and Stong, R.},
  
  keywords = {Representation Theory (math.RT), FOS: Mathematics, FOS: Mathematics},
  
  title = {A lower bound for the dimension of a highest weight module},
  
  %TODO
  journal = {Representation Theory of the American Mathematical Society 21(20)},
  
  publisher = {arXiv},
  
  year = {2016},
  
  copyright = {arXiv.org perpetual, non-exclusive license},

    archivePrefix = {arXiv},
    eprint = {1603.03076},

}

@misc{naud2016,
    title={Hecke operators and spectral gaps on compact lie groups},
    author={F. Naud},
    month={5},
    year={2016}
}
\end{document}